\pdfoutput=1 

\documentclass[11pt,paper=A4]{scrartcl}

\usepackage{libertinus}
\usepackage[T1]{fontenc}
\usepackage[utf8]{inputenc}
\usepackage{csquotes}
\usepackage[american]{babel}

\usepackage{mathtools}
\mathtoolsset{centercolon,mathic} 
\DeclareRobustCommand{\colon}{\nobreak\mskip2mu\mathpunct{}\nonscript
  \mkern-\thinmuskip{\ordinarycolon}\mskip6muplus1mu\relax}
\allowdisplaybreaks[3]

\renewcommand{\eqref}[1]{(\ref{#1})}
\usepackage{microtype}
\usepackage{amssymb}
\usepackage{tikz}
\usetikzlibrary{calc,intersections,cd}
\tikzset{%
    myLine/.style={line width=#1*0.89pt},
    myLine/.default=1,
}

\usepackage{enumitem}

\let\Delta\varDelta

\let\Phi\varPhi

\let\Psi\varPsi

\let\Lambda\varLambda

\let\Gamma\varGamma

\let\Pi\varPi

\let\Omega\varOmega

\usepackage{booktabs}
\usepackage[
    hidelinks,
    bookmarksnumbered,
    bookmarksopen,
    bookmarksopenlevel=2,
    bookmarksdepth=3,
    pdfusetitle=true,
    breaklinks=true,
]{hyperref}

\newcommand{\statement}[1]{\paragraph{#1}\pdfbookmark[1]{#1}{#1}} 
\usepackage{caption}
\usepackage{floatrow}
\usepackage{stmaryrd}
\usepackage[rightcaption]{sidecap}
\usepackage{orcidlink}

\definecolor{myColor}{HTML}{887733}

\usepackage{stackengine} 
\stackMath 

\makeatletter
\newcommand{\arrow@text}[1]{\scriptstyle\mathrm{#1}}
\newcommand{\labeledRightArrow}[1]{ %
    \mathrel{ 
        \stackon[-3.5pt]{ 
            \xrightarrow{\hphantom{\arrow@text{#1}}} 
        }{ %
            \arrow@text{#1\;} %
        } %
    } %
}
\makeatother

\usepackage[
    bibstyle=phys,
    citestyle=numeric,
    biblabel=brackets,
    mincitenames=1,
    maxcitenames=3,
    minalphanames=3,
    maxalphanames=4,
    minbibnames=3,
    maxbibnames=10,
    arxiv=abs,
    eprint=true,
    doi=false,
    url=false,
    sorting=nty,
    useprefix=true,
    alldates=year,
    urldate=long,
]{biblatex}
\renewbibmacro*{note+pages}{%
    \printfield{note}%
    \setunit{\bibpagespunct}%
    \printfield{pages}%
    \newunit%
}
\makeatletter
\DeclareFieldFormat{eprint:arxiv}{%
    \ifhyperref
    {\href{https://arxiv.org/\abx@arxivpath/#1}{%
            arXiv\addcolon\linebreak[2]#1}}
    {arXiv\addcolon
        \nolinkurl{#1}%
    }
}
\makeatother

\renewbibmacro*{maintitle+title}{%
    \iffieldsequal{maintitle}{title}{%
        \clearfield{maintitle}%
        \clearfield{mainsubtitle}%
        \clearfield{maintitleaddon}%
    }{%
        \iffieldundef{maintitle}{}{%
            \printtext[doi/url-link]{%
                \usebibmacro{maintitle}%
            }%
            \newunit\newblock
            \iffieldundef{volume}{}{%
                \printfield{volume}%
                \printfield{part}%
                \setunit{\addcolon\space}%
            }%
        }%
    }%
    \printtext[doi/url-link]{%
        \usebibmacro{title}%
    }%
    \newunit%
}

\DeclareBibliographyDriver{article}{%
    \usebibmacro{bibindex}%
    \usebibmacro{begentry}%
    \usebibmacro{author/translator+others}%
    \setunit{\labelnamepunct}\newblock
    \usebibmacro{title}%
    \newunit
    \printlist{language}%
    \newunit\newblock
    \usebibmacro{byauthor}%
    \newunit\newblock
    \usebibmacro{bytranslator+others}%
    \newunit\newblock
    \printfield{version}%
    \newunit\newblock
    \printtext[doi/url-link]{%
        \usebibmacro{journal+issuetitle}%
        \newunit
        \usebibmacro{byeditor+others}%
        \newunit
        \usebibmacro{note+pages}%
        \newunit\newblock
        \iftoggle{bbx:isbn}
        {\printfield{issn}}
        {}%
        \setunit{\addspace}%
    }%
    \printfield{year}%
    \setunit{\addspace}%
    \usebibmacro{doi+eprint+url}%
    \newunit\newblock
    \usebibmacro{addendum+pubstate}%
    \setunit{\bibpagerefpunct}\newblock
    \usebibmacro{pageref}%
    \newunit\newblock
    \iftoggle{bbx:related}
        {\usebibmacro{related:init}%
        \usebibmacro{related}}
        {}%
    \usebibmacro{finentry}%
}

\usepackage{xurl} 

\makeatletter

\usepackage{bookmark}
\bookmarksetup{open,numbered}
\usepackage{amsthm}
\usepackage{thmtools}

\def\assumptionthm@envname{assumptionthm}

\declaretheorem[
    title=Proposition,%
    parent=section,%
    style=plain,%
]{prop}
\declaretheorem[
    title=Theorem,%
    sibling=prop,%
    style=plain,%
]{theorem}
\declaretheorem[
    title=Lemma,%
    sibling=prop,%
    style=plain,%
]{lem}

\declaretheorem[
    title=Remark,%
    sibling=prop,%
    style=definition,%
]{rmk}
\declaretheorem[
    title=Main Result,%
    numbered=no,%
    style=plain,%
]{mainres}

\def\thmt@headstyle@assumption{%
    \ABBREVIATION\ \LONGNAME%
}

\declaretheoremstyle[
    headfont=\normalfont\bfseries,%
    notefont=\normalfont\bfseries,%
    notebraces={}{},%
    headpunct={},%
    headformat=assumption,%
]{assumption}

\declaretheorem[
    name=Assupmtion,%
    numbered=no,%
    style=assumption,%
    qed=\(\diamond\),%
]{assumptionthm}

\newenvironment{assumption}[2]{%
    \def\LONGNAME{#2}%
    \def\ABBREVIATION{#1}%
    \begin{assumptionthm}[%
        name=\upshape #1,%
    ]%
}{
    \end{assumptionthm}
}

\makeatother

\DeclarePairedDelimiter{\intervaloo}{\lparen}{\rparen}
\DeclarePairedDelimiter{\intervaloc}{\lparen}{\rbrack}

\DeclarePairedDelimiter{\intervalcc}{\lbrack}{\rbrack}

\DeclarePairedDelimiter{\abs}{\lvert}{\rvert}
\DeclarePairedDelimiter{\paren}{\lparen}{\rparen}

\DeclarePairedDelimiter{\norm}{\lVert}{\rVert}
\DeclarePairedDelimiter{\Norm}{\lVert}{\rVert}
\DeclarePairedDelimiter{\commutator}{\lbrack}{\rbrack}

\DeclarePairedDelimiter{\List}{\{}{\}}

\DeclarePairedDelimiter{\Ab}{\llbracket}{\rrbracket}

\DeclarePairedDelimiter{\braket}{\langle}{\rangle}
\DeclarePairedDelimiter{\expectation}{\langle}{\rangle}
\DeclarePairedDelimiter{\bra}{\langle}{\rvert}
\DeclarePairedDelimiter{\ket}{\lvert}{\rangle}
\DeclarePairedDelimiterXPP{\dist}[1]{d}{\lparen}{\rparen}{}{#1}
\newcommand\cmeansym{\mathbb{E}}
\DeclarePairedDelimiterXPP{\cmean}[2]{\cmeansym_{#1}}{\lparen}{\rparen}{}{#2}

\DeclareDocumentCommand{\trace}{s o e{_} m}{
    \IfValueTF{#3}
        {\tr_{#3}}
        {\tr}%
    \IfBooleanTF{#1}
        {\paren*{#4}}
        {
            \IfValueTF{#2}
                {\paren[#2]{#4}}
                {\paren{#4}}%
        }%
}

\providecommand\given{}

\newcommand\SetSymbol[1][]{%
    \nonscript\,#1\vert
    \allowbreak
    \nonscript\,
    \mathopen{}}
\DeclarePairedDelimiterX\Set[1]\{\}{%
    \renewcommand\given{%
        \SetSymbol[\delimsize]}
    \nonscript\,
    #1
    \nonscript\,
}

\ExplSyntaxOn
\makeatletter
\cs_generate_variant:Nn \tl_if_eq:nnT {onT}
\tl_new:N \__scale_new_delimsize_tl
\cs_new:Nn \__scale_make_bigger_tl:N {
    \tl_if_eq:onT {#1} {\@MHempty} { \tl_set:Nn \__scale_new_delimsize_tl { big } }
    \tl_if_eq:onT {#1} {}          { \tl_set:Nn \__scale_new_delimsize_tl { big } }
    \tl_if_eq:onT {#1} {\big}      { \tl_set:Nn \__scale_new_delimsize_tl { Big } }
    \tl_if_eq:onT {#1} {\Big}      { \tl_set:Nn \__scale_new_delimsize_tl { bigg } }
    \tl_if_eq:onT {#1} {\bigg}     { \tl_set:Nn \__scale_new_delimsize_tl { Bigg } }
    \tl_if_eq:onT {#1} {\middle}   { \tl_set:Nn \__scale_new_delimsize_tl { middle } }
}
\cs_new:Nn \scale_make_bigger_m:N {
    \__scale_make_bigger_tl:N #1
    \use:c { \tl_use:N \__scale_new_delimsize_tl }
}
\cs_new:Nn \scale_make_bigger_l:N {
    \__scale_make_bigger_tl:N #1
    \tl_if_eq:NnTF \__scale_new_delimsize_tl {middle} {
        \tl_set:Nn \__scale_new_delimsize_tl {left}
    }{
        \tl_put_right:Nn \__scale_new_delimsize_tl {l}
    }
    \use:c { \tl_use:N \__scale_new_delimsize_tl }
}
\cs_new:Nn \scale_make_bigger_r:N {
    \__scale_make_bigger_tl:N #1
    \tl_if_eq:NnTF \__scale_new_delimsize_tl {middle} {
        \tl_set:Nn \__scale_new_delimsize_tl {right}
    }{
        \tl_put_right:Nn \__scale_new_delimsize_tl {r}
    }
    \use:c { \tl_use:N \__scale_new_delimsize_tl }
}
\DeclarePairedDelimiterXPP{\pdist}[1]{\scale_make_bigger_l:N\delimsize\lparen d}{\lparen}{\rparen}{\scale_make_bigger_r:N\delimsize\rparen}{#1}
\DeclarePairedDelimiterXPP{\pdiam}[1]{\scale_make_bigger_l:N\delimsize\lparen \diam}{\lparen}{\rparen}{\scale_make_bigger_r:N\delimsize\rparen}{#1}

\makeatother
\ExplSyntaxOff

\newcommand{\mathup}[1]{\mathrm{#1}}
\newcommand{\quadtext}[1]{\quad\text{#1}\quad}
\newcommand{\qquadtext}[1]{\quad\quadtext{#1}\quad}
\newcommand{\Alignindent}{\hspace*{2em}&\hspace*{-2em}}
\newcommand{\alignindent}{\hspace{-1em}}

\newcommand{\sumstack}[2][]{\ifstrempty{#1}{\sum_{\substack{#2}}}{\smashoperator[#1]{\sum_{\substack{#2}}}}}
\newcommand{\colonpunct}{\mathpunct{:}}

\DeclareMathOperator{\diam}{diam}
\DeclareMathOperator{\rank}{rank}
\DeclareMathOperator{\support}{supp}
\DeclareMathOperator{\spectrum}{spec}

\let\Gammar\Gamma
\renewcommand{\Gamma}{{\mathbb{Z}^d}}

\newcommand{\epsi}{\varepsilon}
\newcommand{\E}{\mathrm{e}}
\newcommand{\e}{\E}
\newcommand{\I}{\mathrm{i}}
\newcommand{\evol}[2]{\E^{#1}\,#2\,\E^{-#1}}
\newcommand{\evoli}[2]{\E^{-#1}\,#2\,\E^{#1}}
\newcommand{\N}{ \mathbb{N} }
\newcommand{\Z}{ \mathbb{Z} }

\newcommand{\R}{ \mathbb{R} }
\newcommand{\C}{ \mathbb{C} }

\newcommand{\HSpace}{\mathcal{H}}
\newcommand{\calL}{\mathcal{L}}
\newcommand{\calI}{\mathcal{I}}
\newcommand{\calJ}{\mathcal{J}}
\newcommand{\unit}{\mathbf{1}}

\newcommand{\Alg}{\mathcal{A}}
\newcommand{\alg}{\Alg}

\newcommand{\D}{\mathrm{d}}

\DeclareMathOperator{\adjoint}{ad}
\DeclareMathOperator{\tr}{tr}
\newcommand{\interpunct}{\,\cdot\,}

\newcommand{\vts}{v \, \abs{t-s}}
\newcommand{\Wf}{\mathcal{W}}
\newcommand*{\Rt}{\tilde R}

\newcommand{\Cvol}{C_{\mathup{vol}}}
\newcommand{\Cint}{C_{\mathup{int}}}
\newcommand{\Ccod}[1]{\mathcal{C}_{\mathup{vol},#1}}

\newcommand{\gap}{\mathup{gap}}

\newcommand{\Lambdag}{{\Lambda^{\mathup{gap}}}}
\newcommand{\Lambdap}{{\Lambda^{\mathup{pert}}}}
\newcommand{\ee}{^{\varepsilon,\eta}}

\newcommand{\LR}{\mathup{LR}}

\newcommand{\vertiii}[1]{{\left\vert\kern-0.25ex\left\vert\kern-0.25ex\left\vert #1 
        \right\vert\kern-0.25ex\right\vert\kern-0.25ex\right\vert}}

\title{Response theory for locally gapped systems}

\author{
    Joscha Henheik%
    \texorpdfstring{%
        \,\orcidlink{0000-0003-1106-327X}
        \footnote{
        Institute of Science and Technology Austria (IST Austria),
        Am Campus 1,
        3400 Klosterneuburg,
        Austria.
        \mbox{Email}:~joscha.henheik@ist.ac.at}
    }{}%
    \and 
    Tom Wessel%
    \texorpdfstring{%
        \,\orcidlink{0000-0001-7593-0913}
        \footnote{
        Fachbereich Mathematik,
        Universität Tübingen,
        Auf der Morgenstelle 10,
        72076 Tübingen,
        Germany. 
        \mbox{Email}:~tom.wessel@uni-tuebingen.de}
    }{}%
}

\date{October 14, 2024}

\begin{document}

\maketitle

\begin{abstract}
    We introduce a notion of a \emph{local gap} for interacting many-body quantum lattice systems and prove the validity of response theory and Kubo's formula for localized perturbations in such settings.
    On a high level, our result shows that the usual spectral gap condition, concerning the system as a whole, is not a necessary condition for understanding local properties of the system.

    More precisely, we say that an equilibrium state \(\rho_0\) of a Hamiltonian \(H_0\) is locally gapped in \(\Lambdag \subset \Lambda\), whenever the Liouvillian \(- \I \, [H_0, \, \cdot \, ]\) is almost invertible on local observables supported in \(\Lambdag\) when tested in \(\rho_0\).
    To put this into context, we provide other alternative notions of a local gap and discuss their relations.

    The validity of response theory is based on the construction of \emph{non-equilibrium almost stationary states} (NEASSs).
    By controlling locality properties of the NEASS construction, we show that response theory holds to any order, whenever the perturbation \(\epsilon V\) acts in a region which is further than \(\abs{\log \epsilon}\) away from the non-gapped region \(\Lambda \setminus \Lambdag\).
\end{abstract}

\clearpage

\tableofcontents

\clearpage

\section{Introduction}
\label{sec:intro}

Spectral gaps lie at the heart of many areas of physics and mathematics, and their existence has several profound consequences.
Classical mathematical examples include (i)~Cheeger’s inequality~\cite{Cheeger1970, Mohar1989}, relating the spectral gap of the (discrete) Laplacian on a Riemannian manifold (a graph) to geometric properties of the underlying space, (ii)~exponential heat kernel estimates~\cite{BGV2003}, (iii)~classical Poincaré or Nash inequalities~\cite{LL2001}, or (iv)~the logarithmic Sobolev inequality and hypercontractivity for Markov semigroups~\cite{Gross1975a, Gross1975b, BE1985}.

In physics, questions concerning the spectral gap lie at the center of several of the most challenging problems, e.g.~the Berry-Tabor~\cite{BT1977} and Bohigas-Giannoni-Schmit~\cite{BGS1984} conjectures in quantum chaos, Haldane's conjecture on integer valued antiferromagnetic Heisenberg chains~\cite{Haldane1983a, Haldane1983b}, or the Yang-Mills mass gap problem.
The great interest in spectral gaps, far beyond these famous conjectures, is rooted in the fact that its existence has tremendous effects on fundamental properties of the system: It is textbook material in condensed matter theory, that the (non)-existence of a band gap determines the isolating (conducting) behavior of a material, and in superconductivity, the existence of a non-zero solution of the BCS gap equation~\cite{BCS1957} decides whether a system is superconducting. 

Moreover, in quantum many-body physics, the existence of a spectral gap above the ground state eigenvalue has far-reaching consequences for entanglement properties~\cite{NS2006, Hastings2007} and ground state correlations~\cite{HK2006}.
On the other hand, the \emph{closing} of a spectral gap is related to the occurrence of a (topological) quantum phase transition~\cite{BMNS2012, NSY2019}.
Finally, the assumption of a spectral gap played a crucial role in recent proofs of adiabatic theory and linear response for many-body systems (see Section~\ref{sec:response-theory} for a detailed discussion).

However, an overall drawback in all of the above examples is that the desired spectral property is of \emph{global} nature, i.e.~involving the studied system as a whole, and thus seldomly compatible with a notion of \emph{locality} in an underlying physical space.
Therefore, a natural question to ask is:
\begin{quote}
    How can one express that a system is \emph{locally} gapped, and which consequences that one has for globally gapped systems persist?
\end{quote}

In this paper, we study this question in the setting of locally interacting many-body quantum spin lattice systems; see Section~\ref{sec:setup} for precise definitions.
More precisely, in the above spirit of our guiding question, this paper has two main goals:
\begin{itemize}
    \item[(i)] We propose a notion of a \emph{local gap} via an (equivalent) dynamical characterization and exemplary prove that local perturbations of Hamiltonians with a frustration free product ground state satisfy this condition.
        Moreover, we study possible alternative notions of local gaps and their relations among each other (see Sections~\ref{subsec:LDG} and~\ref{sec:localgap}).
    \item[(ii)] As an application to a physically relevant problem, we show that for Hamiltonians with a local gap, response theory approximately holds to any order and thus justify \emph{Kubo's formula} (see Sections~\ref{subsec:discMR} and~\ref{sec:response-theory-main}).
\end{itemize}

There are only few works in the literature studying many-body quantum systems under a non-standard gap condition, i.e.~one differing from the clean separation of eigenvalues: In~\cite{DEF2024}, the authors derive Kubo's formula for two-dimensional disordered systems having only a \emph{mobility gap} (cf.~also the recent paper~\cite{DGHP24} dealing with spectral gaps in presence of disorder).
Together with Teufel, one of us in~\cite[Theorem~4.1]{HT2020bulk}, proved that finite systems, whose analog in the thermodynamic limit has a usual spectral gap, approximately obey an adiabatic theorem.
In another recent work~\cite{YSL24}, the authors developed a theory of metastable states, characterized by the requirement that local operators raise the energy of such a state by a certain minimal amount (their condition is similar to an alternative notion of a local gap given in~\eqref{eq:GAPdecay} below).
Finally, we remark that, in the context of Lie group theory, the notion of a “local gap” has recently been introduced~\cite{BIG2017} and proved itself to have profound consequences~\cite{BIG2017, BI2018}.

Next, in Section~\ref{sec:response-theory}, we discuss the problem of justifying linear (and higher order) response theory and Kubo’s formula based on adiabatic theory.
Afterwards, in Section~\ref{subsec:LDG}, we introduce our \emph{local dynamical gap condition} \nameref{ass:localGAP}.
Finally, in Section~\ref{subsec:discMR}, we discuss our main result on response theory.

\subsection{Response theory in many-body quantum systems}
\label{sec:response-theory}
The purpose of response theory is to express how quantum expectation values change, after a small perturbation is slowly turned on. 
More precisely, one considers an unperturbed Hamiltonian \(H_0\) with equilibrium state (usually a ground state) \(\rho_{0}\) and slowly turns on a small additive perturbation \(\epsi V\). Denoting by \(\rho_\epsi\) the state of the system after \(\epsi V\) has been turned on, one aims to understand, how the expectation value of an observable \(B\) changes, i.e.~determine
\begin{equation}
    \label{eq:introresponse}
    \expectation{B}_{\rho_\epsi} - \expectation{B}_{\rho_0}
    =
    \epsi \, \sigma_B + o(\epsi)
\end{equation}
at least to leading order in the strength \(\epsi\). 

A central piece of response theory is \emph{Kubo's formula}~\cite{Kubo1957}, which provides a simple expression for the so-called linear response coefficient \(\sigma_B\) in~\eqref{eq:introresponse}.
Despite the simplicity and empirical success of Kubo's formula, the problem of justifying it in a very general framework has so far escaped rigorous treatment.
The fundamental difficulty lies in the fact that, in general, the state \(\rho_\epsi\) is no longer an equilibrium state, and therefore determining it is outside the powerful realm of equilibrium statistical mechanics.
This problem has prominently been pointed out by Simon~\cite{Simon1984} in 1984 in his “Fifteen problems in mathematical physics”, containing a rigorous justification of Kubo's formula from first principles as problem (4B).

However, in the particular setting of many-body lattice systems with a spectral gap at zero temperature, it has recently been possible~\cite{BRF2018, MT2019, Teufel2020, HT2020} to actually prove Kubo's formula and justify the applicability of linear response theory to compute the change in expectation values~\eqref{eq:introresponse}.
The more general underlying results establish \emph{generalized super-adiabatic theorems}%
\footnote{%
    \label{ftn:superad}
    This term describes adiabatic theorems for time-dependent Hamiltonians of the form \(H_\epsi(t) = H_0(t) + \epsi V(t)\), where \(H_0(t)\) is assumed to have a spectral gap.
    Now “super-adiabatic” means that for \(\epsi = 0\), there exists a state \(\rho_0^\eta(t)\) close to the instantaneous ground state \(\rho_0(t)\) of \(H_0(t)\), such that the time-evolution generated by \(\eta^{-1}H_0(\interpunct)\) intertwines \(\rho_0^\eta(t_0)\) and \(\rho_0^\eta(t)\) to any order in \(\eta\).
    The term “generalized” means that, even for a gap-closing perturbation \(V\), there exist super-adiabatic non-equilibrium almost-stationary states (NEASSs) \(\Pi^{\epsi, \eta}(t)\), which are intertwined by the time evolution generated by \(\eta^{-1} H_\epsi(\interpunct)\) to any order in \(\epsi\) and \(\eta\).
    We refer to~\cite{MT2019, Teufel2020, HT2020unif, HT2020bulk, HW2022} for details (see also Sections~\ref{subsec:NEASS} and~\ref{subsec:adpert}).
} for short range interacting Hamiltonians, which can be written as a sum of local terms, and are hence called \emph{SLT operators}~\cite{HT2020unif, HT2020bulk, HW2022}.

The recent breakthrough, which paved the way for these results, was achieved by Bachmann, De Roeck, and Fraas~\cite{BRF2018} (see~\cite{MT2019} for an adaptation to fermionic systems), who proved the first adiabatic theorem for extended (but finite) quantum lattice systems.
One key difficulty is that, for macroscopic systems, typical operator norm bounds in adiabatic theory deteriorate due to the \emph{orthogonality catastrophe} and one instead has to formulate the result in a weaker topology by testing against local observables.
On a high level, the main ingredient for their proof are the well-known \emph{Lieb-Robinson bounds} (LRBs)~\cite{LR1972}, which ensure a finite speed of correlation and prevent build-up of long-range entanglement.
These LRBs, enabled to prove that the generator of the \emph{spectral flow}, introduced by \textcite{HW2005}, is in fact an SLT operator and hence maintains good locality properties~\cite{HW2005, BMNS2012} (showing so-called \emph{automorphic equivalence} of ground states).

However, the work~\cite{BRF2018} had the limitation that the spectral gap of \(H_0\) is assumed to remain open after adding the perturbation~\(\epsi V\).
To allow~\(\epsi V\) to close the gap, \textcite{Teufel2020} combined ideas from space-time adiabatic perturbation theory~\cite{PST2003, PST2003b} with locality estimates from~\cite{BRF2018}.
The underlying perturbative scheme is an iterative application of locality preserving Schrieffer-Wolff transformations (a.k.a.~Lie-Schwinger block diagonalization~\cite{FP2020}), which proved to be a powerful approach in several rather recent works in that direction~\cite{BRF2018, Teufel2020, HT2020unif, HT2020bulk, YL23, DGHP24, YSL24}.
In this paper, we carefully exploit locality properties of the operations involved in the perturbative scheme, which allows to deal with \emph{locally} gapped systems (see Section~\ref{subsec:LDG} below).

The rough physical picture underlying~\cite{Teufel2020, PST2003, PST2003b} is that a gap that is locally intact after adding perturbation should be sufficient for adiabatic theory to be valid (cf.~\cite[Figure~1]{HT2020}).
The results from~\cite{BRF2018, MT2019, Teufel2020} for large but finite systems were subsequently extended to the thermodynamic limit~\cite{HT2020unif, HT2020bulk}, building on an extension of the spectral flow techniques to infinite systems by \textcite{MO2020}.
We point out that, contrary to~\cite{HT2020unif}, the papers~\cite{MO2020, HT2020bulk} assumed a spectral gap \emph{only} for the GNS Hamiltonian of the infinite system (a \emph{gap in the bulk}).
More comprehensive reviews of the developments discussed in this section are given in~\cite{HT2020, HW2022}.

In view of the linear response problem and the second of our principal goals, the contribution of this paper is to extend the previous results for uniformly or bulk gapped systems to systems where \(H_0\) is locally gapped.
This important extension allows to rigorously treat systems with impurities of gap-closing edge modes (see the discussions in Sections~\ref{subsec:LDG}--\ref{subsec:discMR} below).
Technically, our contribution is to control operations on SLT operators, which are localized on a subregion of the whole lattice; see Section~\ref{subsec:interactions} and Appendix~\ref{app:technical}.

Lastly, we remark that we only consider finite-dimensional spaces and bounded operators for simplicity of the presentation.

\subsection{A local dynamical gap condition}
\label{subsec:LDG}
All the results on linear response and adiabatic theory mentioned in Section~\ref{sec:response-theory} above, heavily rely on the range of the equilibrium state \(\rho \equiv \rho_0\) being contained in a gapped part of the spectrum of the unperturbed Hamiltonian \(H \equiv H_0\).%
\footnote{%
    For ease of notation and since there will be no perturbation \(\epsi V\) in the current Section~\ref{subsec:LDG}, we will drop the subscript \(0\) here.
}
More precisely, assume
\begin{equation}
    \label{eq:globspecgap}
    \spectrum(H)
    =
    \sigma_1 \, \dot{\cup} \, \sigma_2
    \qquad \text{with}
    \qquad
    \dist{\sigma_1, \sigma_2} \ge g
\end{equation}
for some gap size \(g> 0\).
Then, denoting the spectral projection%
\footnote{%
    We will follow the convention that \emph{(orthogonal) projections} will be denoted by \(P\) (i.e.~satisfying \(P^2 = P\) and \(P^* = P\)), while \emph{states} are denoted by \(\rho\) (i.e.~satisfying \(\rho = \rho^*\) and \(0 \le \rho \le 1\) with \(\tr \rho = 1\)).
    Clearly, if \(P\) is an orthogonal projection, then \(\rho:= P/\dim \rank P\) is a state.
}
associated to \(H\) onto \(\sigma_1\) by \(P\), we have that \(P \rho P = \rho\).

On a technical level, in all of the works~\cite{BRF2018, MT2019, Teufel2020, HT2020unif, HT2020bulk}, the crucial importance of the gap of \(H\) lies in the \emph{local invertibility of the Liouvillian} \(\calL_{H}\Ab{\interpunct} := -\I \, [H, \interpunct]\).
That is, there exists a locality preserving (usually called \emph{quasi-local}) map \(\calI_{H, g} = \calI_{H, g}\Ab{\interpunct}\) on the observable algebra, depending on the SLT Hamiltonian \(H\) and the gap size \(g\), which inverts the Liouvillian in the projection \(\expectation{\interpunct}_{P} = \trace[\big]{P \interpunct}\)
onto the gapped part of \(H\).
More precisely (see Proposition~\ref{prop:GDG}), for all local observables \(A,B\) it holds that
\begin{equation}
    \label{eq:invertintro}
    \expectation[\big]{[\calL_{H} \circ \calI_{H, g} \Ab{A} - A , B ]}_{P}
    =
    0
    .
\end{equation}
Note that such a map cannot be uniquely characterized as a “weak” right inverse of \(\calL_{H}\), and is thus clearly not unique.
However, mostly for concreteness, we will always work with an \emph{explicitly} constructed~\cite{HW2005,BMNS2012} variant (see Remarks~\ref{rmk:weight} and~\ref{rmk:abstract} for relaxed, rather abstract conditions on \(\calI\)), denoted by
\begin{equation}
    \label{eq:invliouintro}
    \calI_{H, g}\Ab{A}
    :=
    \int_\R \D t \, w_g(t) \int_0^t \D s \, \e^{\I Hs}
    \, A \, \e^{-\I Hs}
    ,
\end{equation}
and henceforth called the \emph{quasi-local inverse of the Liouvillian}.
The positive weight function \(w_g \in L^1(\R)\), normalized to \(\int w_g = 1\), is required to have Fourier transform%
\footnote{%
    We use the convention that \(\widehat{f}(p) := (2 \pi)^{-1/2} \int_\R \D x \, \E^{- \I p x} f(x)\) for the Fourier transform.
} \(\widehat{w_g}\) with support
\begin{equation}
    \label{eq:Fouriercpct}
    \support(\widehat{w_g}) \subset [-g,g]
    .
\end{equation}
Moreover, for the explicitly constructed \(w_g\) (see~Lemma~\ref{lem:weightfunctions} in Appendix~\ref{app:inverse-liouvillian}), we additionally have the bound
\begin{equation}
    \label{eq:wbound}
    \abs{w_g(t)} \le C \, \e^{-\abs{t}^q} \quad \text{for all} \quad q<1
    .
\end{equation}
This estimate~\eqref{eq:wbound} together with classical Lieb-Robinson bounds~\cite{LR1972} for the dynamics generated by \(H\) ensure that \(\calI_{H, g}\) acts as a quasi-local operator.
In Appendix~\ref{app:inverse-liouvillian} we will briefly recall the construction of \(\calI_{H,g}\) and report on its properties in more detail.

\subsubsection{Dynamical characterization of a spectral gap} Interestingly, the Hamiltonian \(H\) having a spectral gap is actually \emph{equivalent} to the invertibility of the Liouvillian.
In particular, the spectral property~\eqref{eq:globspecgap} can be \emph{dynamically} characterized via~\eqref{eq:invertintro}.
A short proof of this fact is given in Section~\ref{subsec:dyncharGDG}.

\begin{prop}[Dynamical characterization of a spectral gap]
    \label{prop:GDG}
    Let \(H\) be a self-adjoint operator on a finite dimensional Hilbert space \(\HSpace\).
    Let \(g > 0\), \(w_g \in L^1(\R)\) be positive, normalized to \(\int w_g = 1\) and satisfy~\eqref{eq:Fouriercpct} with \(\widehat{w_g}\vert_{(-g,g)} > 0\).
    Decompose the spectrum of \(H\) as \(\spectrum(H) = \sigma_1 \, \dot{\cup} \, \sigma_2\) and let \(P\) be the spectral projection onto \(\sigma_1\).
    Then, denoting \(\calL_{H} \Ab{\interpunct} = -\I \, [H, \interpunct]\) and \(\calI_{H,g}\) as in~\eqref{eq:invliouintro}, we have that
    \begin{equation}
        \label{eq:dynchar}
        \expectation[\big]{
            [\calL_{H} \circ \calI_{H, g} \Ab{A} - A , B ]
        }_{P}
        =
        0
        \quad \forall A,B \in \mathcal{B}(\HSpace)
        \quad \iff \quad
        \dist{\sigma_1, \sigma_2} \ge g
        .
    \end{equation}
\end{prop}

The goal of this article is to relax the requirement of a globally spectrally gapped Hamiltonian~\(H\) and instead work with a so-called \emph{local dynamical gap condition} \nameref{ass:localGAP}.
This condition roughly asserts that, the Hamiltonian \(H\) behaves as if it had a gap in a spatially localized region, i.e.~that the Liouvillian can (almost) be locally inverted in that region.
A more formal version of \nameref{ass:localGAP} is formulated in Assumption \nameref{ass:localGAP_main}.

\begin{assumption}{(LDG\textsubscript{intro})}{Local dynamical gap condition -- informal version}
    \label{ass:localGAP}
    \noindent
    Let \(H\) be an SLT Hamiltonian and \(\rho\) an equilibrium state of \(H\), i.e.~\([H, \rho] = 0\).
    We say that \(\rho\) is \emph{locally dynamically gapped} of size at least \(g > 0\) in a region \(\Lambdag \subset \Lambda\) if and only if for all observables%
    \footnote{%
        Throughout this paper, \(\Alg_X\) denotes the algebra of observables with support in \(X\).
    }
    \(A \in \alg_X\) and \(B \in \alg_Y\) localized in \(X\subset \Lambda\) and \(Y \subset \Lambda\), it holds that
    \begin{equation}
        \label{eq:local gap}
        \abs[\Big]{\expectation[\big]{\commutator[\big]{\calL_{H} \circ \calI_{H,g}\Ab{A} - A , B}}_{\rho}}
        \leq
        \begin{multlined}[t]
            C \, \norm{A} \, \norm{B} \, \paren[\big]{\diam(X) + \diam(Y)}^\ell
            \\\times
            \exp\paren[\Big]{-\paren[\big]{\dist{X, \Lambda \setminus \Lambdag}+ \dist{Y, \Lambda \setminus \Lambdag}}^q}
        \end{multlined}
    \end{equation}
    for some fixed \(\ell \in \N_0\) and constants \(C\), \(q>0\), independent of the sizes \(\abs{\Lambda}\) and \(\abs{\Lambdag}\).
\end{assumption}

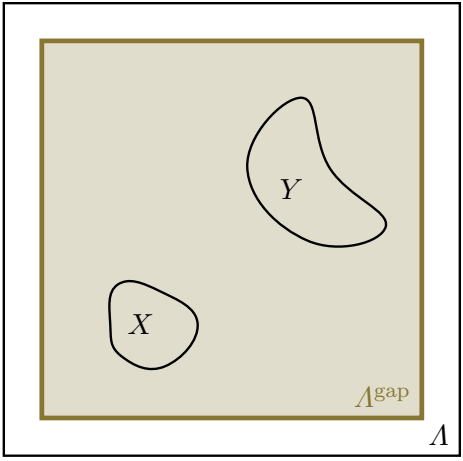
\begin{SCfigure}[1.2]
    \centering
    \begin{tikzpicture}[scale=.5, myLine]
        \newcommand{\pathL}{(-6,6) rectangle (6,-6)}
        \newcommand{\pathLo}{(-5,5) rectangle (5,-5)}
        \newcommand{\posX}{(-2.4,-2.5)}
        \newcommand{\pathX}[1]{plot[smooth cycle, tension=.8]%
            coordinates{+(0:#1+1.5) +(60:#1+1) +(120:#1+1.2) +(180:#1+.8) +(240:#1+1) +(300:#1+1.3)}}
        \newcommand{\pathY}{plot[smooth cycle, tension=0.8] coordinates{+(10:.6) +(95:2) +(175:1.6) +(270:1.8) +(325:2.5)}}

        \newcommand{\drawXYL}[1]{
            \draw \posX{} node {\(X\)} \pathX{0};
            \draw (2,1.5) +(225:.6) node {\(Y\)} \pathY;
            \draw \pathL node[anchor=south east] {#1};
        }

        \begin{scope}
            \begin{scope}[myColor]
                \begin{scope}[even odd rule, myLine=2, fill=myColor!25, draw=myColor]
                    \path[draw, fill] \pathLo (0,0) ;
                \end{scope}
            \end{scope}
            \begin{scope}[myLine=2, even odd rule]
                \clip \pathL \pathLo;
            \end{scope}
            \node[anchor=south east, myColor] at (5,-5) {\(\Lambda^{\mathrm{gap}}\)};
        \end{scope}
        \drawXYL{\(\Lambda\)}
    \end{tikzpicture}
    \caption{
        Illustrated is the local dynamical gap condition \nameref{ass:localGAP} in the case where the system is gapped in the bulk of \(\Lambda\), e.g., due to gap closing edge modes.
        If the observables \(A \in \Alg_X\) and \(B \in \Alg_Y\) are supported well inside \(\Lambda \setminus \Lambdag\), the rhs.~of~\eqref{eq:local gap} is small, i.e.~the Liouvillian is locally almost invertible.
    }
    \label{fig:edge}
\end{SCfigure}

In a nutshell, this means that, within \(\Lambdag\), the Hamiltonian \(H\) approximately behaves as if it were spectrally gapped -- up to an error vanishing (stretched) exponentially fast in the distance to \(\Lambda \setminus \Lambdag\).
On the physical level, one might think of \(\Lambda \setminus \Lambdag\) as some impurity region causing the global spectral gap to close, or the boundary of \(\Lambda\) and hence allowing for gap-closing edge modes (see Figure~\ref{fig:edge}).
Indeed, as we will show in Section~\ref{subsec:examples}, the local gap condition \nameref{ass:localGAP} is satisfied for ground states of locally in \(\Lambda \setminus \Lambdag\) (but arbitrarily strongly) perturbed Hamiltonians of certain quantum spin systems, which have a globally gapped ground state.

\subsubsection{Verifying the local dynamical gap condition}
Despite the supportive examples above, our concrete formulation of a local gap condition \nameref{ass:localGAP} might still seem a bit \emph{ad hoc} at the moment.
Therefore, we will outline several alternative ways to formulate such a condition and discuss their respective features and relations in Section~\ref{sec:localgap}.
In particular, in Proposition~\ref{prop:mechforlocalgap} we show the following (the constants \(q, C, \ell\) have the same meaning as in~\eqref{eq:local gap} but might take different values):
\begin{itemize}
    \item[(1)] Let the SLT Hamiltonian \(H\) (with equilibrium state \(\rho\)) be obtained from a globally gapped SLT Hamiltonian \(H_*\) (with equilibrium state \(\rho_*\)) by an SLT perturbation \(J\) localized in \(\Lambda\setminus \Lambdag\), i.e.~\(H = H_* + J\). Then, if the locally tested difference \(\rho- \rho_*\) is small in trace norm, i.e.
        \begin{equation*}
            \norm{(\rho- \rho_*) \, A}_{\tr} + \norm{A \, (\rho- \rho_*)}_{\tr}
            \le
            C \, \norm{A} \, \diam(X)^\ell \exp\paren[\big]{-\dist{X, \Lambda\setminus \Lambdag}^{q}}
            ,
        \end{equation*}
        then \(\rho\) is locally dynamically gapped (cf.~Proposition~\ref{prop:mechforlocalgap}~(i) and (vi)).
    \item[(2)] In the same setting as in (1), it holds that, whenever there exists a norm-preserving automorphism \(\tau\) on the observable algebra, i.e.~\(\expectation{\interpunct}_\rho = \expectation{\tau \Ab{\interpunct}}_{\rho_*}\), which satisfies
        \begin{equation*}
            \norm[\big]{( \tau - \unit )\Ab{A}}
            \le
            C \, \norm{A} \, \diam(X)^\ell \, \exp\paren[\big]{-\dist{X, \Lambda\setminus \Lambdag}^{q}}
            ,
        \end{equation*}
        then \(\rho\) is locally dynamically gapped (cf.~Proposition~\ref{prop:mechforlocalgap}~(vii)).
    \item[(3)] Let \(H\) be an SLT Hamiltonian and \(\rho = \ket{\psi} \bra{\psi}\) its pure \emph{product ground state}.
        Then, if one has an effective gap well inside \(\Lambdag\) of the form%
        \footnote{%
            For \(\Lambdag = \Lambda\), this condition for all observables, is actually equivalent to the usual spectral gap condition.
        }
        \begin{equation*}
            \I \, \expectation[\big]{A^* \calL\Ab{A}}_\rho
            \ge
            g
            \, \paren[\Big]{
                1- C \diam(X)^\ell \, \exp\paren[\big]{-\dist{X, \Lambda\setminus \Lambdag}^{q}}
            }
            \, \paren[\Big]{
                \expectation{A^* A}_\rho - \abs[\big]{\expectation{A}_\rho}^2
            }
            ,
        \end{equation*}
        for all \(A \in \Alg_X\), then \(\rho\) is locally dynamically gapped (cf.~Proposition~\ref{prop:mechforlocalgap}~(viii)).
\end{itemize}
Items (1) and (2) will be used in Section~\ref{subsec:frustfree}, to show that ground states of perturbations of gapped frustration free Hamiltonians have a local dynamical gap.

\subsection{Discussion of our main result}
\label{subsec:discMR}
We can now formulate an informal version of our main result as a showcase application of our local gap condition \nameref{ass:localGAP} to a physically relevant problem -- the validity of response theory.
In a nutshell, it says that, even after relaxing the usual condition of a global gap to \nameref{ass:localGAP}, we have \emph{response theory to all orders} for a perturbation localized in \(\Lambdap\) -- provided that \(\dist{\Lambdap, \Lambda \setminus \Lambdag}\) is sufficiently large compared to \(\abs{\log(\epsi)}^{1/q}\), where \(\epsi>0\) is the strength of the perturbation and \(q\in \intervaloo{0,1}\) is some small constant (see~\eqref{eq:distespeffective} and Figure~\ref{fig:main}).

\begin{SCfigure}[1.2][h]
    \centering
    \begin{tikzpicture}[scale=.5, myLine]
        \newcommand{\pathouter}{(-6.2,6.2) rectangle (6.2,-6.2)}
        \newcommand{\pathL}{(-6,6) rectangle (6,-6)}
        \newcommand{\pathLp}{(-6,6) -- plot[smooth, tension=0.8] coordinates{(3,6) (5,2) (3.5,-1) (2,-4) (-2,-4) (-5,-6)} -- (-6,-6) -- cycle}
        \newcommand{\posY}{(-3,0.5)}
        \newcommand{\pathY}{plot[smooth cycle, tension=0.9] coordinates{+(20:5) +(60:1.5) +(150:1.6) +(230:2.8) +(325:3.5)}}

        \begin{scope}[myLine=2, even odd rule, fill=myColor!25, draw=myColor] 
            \clip \pathLp;
            \path[draw, fill] \pathLp (0,0) ;
            \draw \pathLp;
        \end{scope}
        \node[anchor=north east, myColor] at (3,6) {\(\Lambda^{\mathrm{gap}}\)}; 

        \draw \posY node {\(\Lambda^{\mathrm{pert}}\)} \pathY; 

        \begin{scope}[myLine=2, even odd rule] 
            \clip \pathouter \pathL;
            \draw \pathL;
        \end{scope}
        \draw (6,-6) node[anchor=south east] {\(\Lambda\)};

        \draw[latex-latex, shorten <=2pt, shorten >=3pt] \posY{} ++(-65:2.8) -- node[anchor=west] { \(\scriptstyle \gg \abs{\log \epsi}^{1/q}\)} (-1.5,-4);
    \end{tikzpicture}
    \caption{
        Let \(H_0\) be locally dynamically gapped in \(\Lambda \setminus \Lambdag\) and consider a small \(\Lambda^{\mathrm{pert}}\)-localized perturbation \(\epsi V\), which is adiabatically turned on.
        Then, if the distance between \(\Lambda^{\mathrm{pert}}\) and \(\Lambda \setminus \Lambdag\) is large compared to \(\abs{\log \epsi}^{1/q}\), response theory~\eqref{eq:linearresponseintro} holds to any order.
    }
    \label{fig:main}
\end{SCfigure}
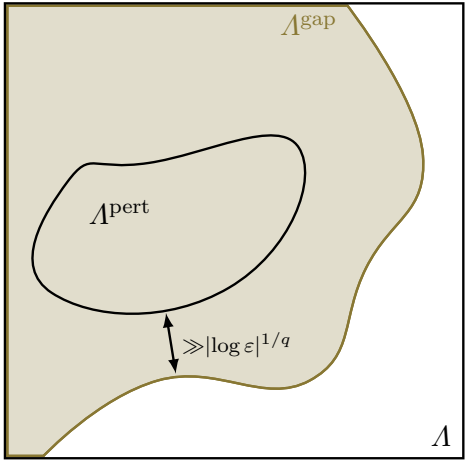

More precisely, let \(H_0\) be an SLT Hamiltonian and \(\rho_0\) an equilibrium state of \(H_0\) that is locally dynamically gapped in \(\Lambdag\) (according to Assumption \nameref{ass:localGAP}).
Let \(V\) be a \(\Lambdap\)-localized perturbation by an SLT Hamiltonian (see~\eqref{eq:interaction-norm-localization} and Lemma~\ref{lem:SLTlocal} for details), \(f:\R \to \intervalcc{0,1}\) a smooth switching function with \(f(t) = 0\) for \(t \le -1\) and \(f(t) = 1\) for \(t \ge 0\), and define
\begin{equation}
    \label{eq:fundHamiltonian}
    H_\epsi(t) := H_0 + \epsi f(t) V
    .
\end{equation}
Moreover, let \(\rho^{\epsi, \eta, f}(t)\) be the solution of the time-dependent adiabatic Schrödinger equation
\begin{equation}
    \I \, \eta \frac{\D }{\D t} \rho^{\epsi, \eta, f}(t)
    =
    \commutator[\big]{H_\epsi(t) , \rho^{\epsi, \eta, f}(t)}
\end{equation}
with adiabatic parameter \(\eta \in (0,1]\) and initial datum \(\rho^{\epsi, \eta, f}(t) = \rho_0\) for all \(t \le -1\).
Finally, for an observable \(B \in \Alg_{Y}\) define the response to the perturbation as
\begin{equation}
    \Sigma_B^{\epsi, \eta, f}(t)
    :=
    \expectation{B}_{\rho^{\epsi, \eta, f}(t)} - \expectation{B}_{\rho_0}
    .
\end{equation}

\begin{mainres}[see Theorem~\ref{thm:linear-response}]
    \label{thm:mainres}
    For every \(j \in \N\) there exists a response coefficient \(\sigma_{B,j}\), independent of \(\epsi, \eta\) and \(f\), such that the following holds: There exists a constant \(q\in \intervaloo{0,1}\) and for every \(n,m \in \N\) and \(T > 0\), there exists a constant \(C>0\), independent of \(\epsi\), such that for every \(t \in [0,T]\) we have that
    \begin{equation}
        \begin{split}
            \label{eq:linearresponseintro}
            \sup_{\eta \in \intervalcc[\big]{\epsi^m , \epsi^{\frac{1}{m}}}}\abs*{\Sigma_B^{\epsi, \eta, f}(t) - \sum_{j=1}^{n} \epsi^j \sigma_{B,j}}
            \le
            C \, \vertiii{B} \, \epsi^{n+1} \paren*{1+ \E^{-\dist{\Lambdap, \Lambda \setminus \Lambdag}^{q} - C \log (\epsi)}}\,,
        \end{split}
    \end{equation}
    where \(\vertiii{B}\) measures the norm of \(B\) and its support \(Y\).

    The first order coefficient is given by \emph{Kubo's formula}
    \begin{equation}
        \label{eq:Kubo}
        \sigma_{B,1}
        =
        - \I \, \expectation[\big]{\commutator[\big]{\calI_{H_0,g}\Ab{V} , B}}_{\rho_0}
        .
    \end{equation}
\end{mainres}
When \(\Lambdag = \Lambda\), the equilibrium state \(\rho_0\) of \(H_0\) is \emph{globally} dynamically gapped and the exponential in~\eqref{eq:linearresponseintro} is absent since \(\dist{\Lambdap, \emptyset} : = \infty\).
This special case of our result (when \(\rho_0 = P_0/\dim \rank P_0\) and \(P_0\) projects onto a gapped spectral patch) has already been proven in~\cite{Teufel2020} with extensions to infinite systems in~\cite{HT2020unif,HT2020bulk}; see also Remark~\ref{rmk:mainresult} below.
Moreover, observe that, whenever
\begin{equation}
    \label{eq:distespeffective}
    \dist{\Lambdap , \Lambda \setminus \Lambdag}
    \gg
    \abs{\log(\epsi)}^{1/q}
    ,
\end{equation}
with \(q\) from our main result, the exponential in~\eqref{eq:linearresponseintro} is small compared to~\(1\).
In particular, this is the case, if \(
    \dist{\Lambdap , \Lambda \setminus \Lambdag}
    \gg
    \epsi^{-\delta}
\) for some (arbitrarily small) \(\delta > 0\).
Hence, our main result gives an effective condition on the distance from the perturbation region, \(\Lambdap\), to the non-gapped region, \(\Lambda \setminus \Lambdag\), in comparison to the perturbation strength \(\epsi\), which ensures that response theory to all orders is valid (see Figure~\ref{fig:main}).

It is interesting to compare our main result to the adiabatic theorem~\cite[Theorem 4.1]{HT2020bulk} for finite systems with a gap in the bulk, which -- in some sense -- are \emph{locally gapped} systems.
By routine arguments (see~\cite[Proof of Theorem 4.1]{Teufel2020}),~\cite[Theorem~4.1]{HT2020bulk} yields the analog of~\eqref{eq:linearresponseintro} again with an \emph{additional error term}: Instead of the exponential in~\eqref{eq:linearresponseintro}, one obtains%
\footnote{%
    The bound in~\cite[eq.~(4.1)]{HT2020bulk} essentially means that the unperturbed system on \(\Lambda \equiv \Lambda_k\) is gapped within \(\Lambdag \equiv \Lambda_{\lfloor k(1-o(1)) \rfloor }\), where \(\Lambda_l\) denotes the box of side length \(2l+1\) in \(\Gamma\) centered at zero.
}
\(
    \mathcal{O}\paren[\big]{\paren[\big]{\epsi^C \dist{Y, \Lambda \setminus \Lambdag}}^{-\infty} }
\).
Therefore, in order to have response theory to all orders in this setting, one needs that
\begin{equation}
    \label{eq:distespeffectiveBULK}
    \dist{Y, \Lambda \setminus \Lambdag} \gg \epsi^{-C}
    .
\end{equation}
We point out the following two differences between the conditions in~\eqref{eq:distespeffective} and~\eqref{eq:distespeffectiveBULK}: First, in~\eqref{eq:distespeffective}, the relevant distance is between `where the perturbation \(V\) acts', i.e.~\(\Lambdap\), and `where we do not have a gap', i.e.~\(\Lambda \setminus \Lambdag\).
In contrast to that, the relevant distance in~\eqref{eq:distespeffectiveBULK} is between `where the observable \(B\) tests', i.e.~\(Y\), and `where we do not have a gap', i.e.~\(\Lambda \setminus \Lambdag\).
Second, while in~\eqref{eq:distespeffective} the distance must be much bigger than a power of the \emph{logarithm} of the perturbation strength \(\epsi\), in~\eqref{eq:distespeffectiveBULK}, the distance must be much bigger than a large \emph{inverse power} of \(\epsi\).

\subsection{Structure of the paper}
The rest of this paper is structured as follows.
After introducing the mathematical framework, in particular the underlying space and the concept of locality of SLT operators, in Section~\ref{sec:setup}, we precisely formulate our main result in Theorem~\ref{thm:linear-response} in Section~\ref{sec:response-theory-main} and also give its proof based on the NEASS construction.
In Section~\ref{sec:localgap} we discuss the problem of formulating a local gap condition, formulate different variants and explain their connections, and, moreover, show certain exemplary systems to have a local dynamical gap.
Afterwards, in Section~\ref{sec:proof}, we perform the NEASS construction under Assumption \nameref{ass:localGAP} and prove the necessary inputs for Theorem~\ref{thm:linear-response}.
Additional proofs concerning the formulation of the local gap condition are given in Section~\ref{app:LPPLlocalgap}; several technical lemmata and auxiliary results needed for the arguments in Section~\ref{sec:proof} are deferred to Appendix~\ref{app:technical}.

\section{Mathematical framework}
\label{sec:setup}

In this section, we briefly introduce the (standard) mathematical framework used in the formulation of the adiabatic theorems.
For similar setups see~\cite{Teufel2020,HT2020unif,HT2020bulk,BRF2018}.

\subsection{Spatial structure and algebra of observables}

We consider a quantum spin system on a finite graph~\(\Lambda\) equipped with the graph distance \(d(\interpunct, \interpunct)\).
Let \(B_r(x):= \Set{y \in \Lambda \given d(x,y) \le r}\) be the ball of radius \(r\) centered at \(x \in \Lambda\).
The graph is assumed have dimension (at most) \(d >0\), i.e.~there exists a constant \(\Cvol > 0\) such that
\begin{equation}
    \label{eq:dimension}
    \sup_{x \in \Lambda} \abs{B_r(x)} \le 1 + \Cvol \, r^d
    ,
\end{equation}
where \(\abs{X}\) denotes the number of sites in \(X \subset \Lambda\).
The set of all such graphs \(\Lambda\) is denoted by
\begin{equation}
    \label{eq:Gset}
    \mathcal{G}(d, \Cvol)
    :=
    \Set*{
        \Lambda \text{ finite graph}
        \given
        \sup_{x \in \Lambda} \, \abs{B_r(x)} \le 1 + \Cvol \, r^d
        \text{ for all } r > 0
    }
    .
\end{equation}

To each vertex \(x \in \Lambda\), we associate a single-particle Hilbert space \(\HSpace_x\), which we assume to be of finite dimension, \(\sup_{x \in \Lambda} \dim \HSpace_x < \infty\).
Moreover, for each \(X\subset \Lambda\) let \(\HSpace_X := \otimes_{x \in X} \HSpace_x\) be the many-particle Hilbert space on \(X\) and denote the associated \(C^*\)-algebra of observables by \(\Alg_{X} := \calL(\HSpace_X)\).
Whenever \(X\subset X'\), then \(\Alg_{X}\) is naturally embedded as a subalgebra of \(\Alg_{X'}\) and we set \(\Alg := \Alg_\Lambda\).

Since a very similar construction is common for fermionic lattice systems (see, e.g.,~\cite{NSY2017, HW2022}), all the results almost immediately translate to this setting.

\subsection{Interactions and SLT operators}
\label{subsec:interactions}
An \emph{interaction} is a map
\begin{equation}
    \Phi \colon \Set{X \subset \Lambda} \to \Alg^N \,, \ X \mapsto \Phi(X)
    =
    \Phi(X)^*\in \Alg_X
    .
\end{equation}
With any interaction, one associates a \emph{sum of local terms (SLT) operator} \(A\) via
\begin{equation}
    \label{eq:intSLT}
    A := \sum_{X \subset \Lambda} \Phi(X) \in \Alg
    .
\end{equation}
Note that, while every interaction defines a unique operator, there are multiple interactions realizing the same operator, i.e.~the assignment \(\Phi \mapsto A(\Phi)\) is not invertible.
Note that all interactions and SLT operators are by definition always self-adjoint.

For any \(b > 0\) and \(p \in (0,1]\), we consider the stretched exponential function
\begin{equation}
    \label{eq:stretched-exponential}
    \chi_{b,p}: [0,\infty) \to (0,1] \,, \ x \mapsto \E^{-bx^p}
\end{equation}
as \emph{localization functions} and define the associated \emph{SLT interaction norm}
\begin{equation}
    \label{eq:interaction-norm}
    \norm{\Phi}_{b,p}
    :=
    \sup_{z \in \Lambda} \sumstack{Z\subset \Lambda:\\z\in Z}
    \frac{\norm{\Phi(Z)}}{\chi_{b,p}\pdiam{Z}}
\end{equation}
on interactions.
The quality of the localization for an interaction will be expressed by the finiteness of a norm \(\norm{\Phi}_{b,p}\), \emph{independent} of the (size of the) graph \(\Lambda\).
Operators with this property will often be referred to as \emph{\((b,p)\)-localized SLT operators}.

Moreover, in order to further quantify, how well an interaction is localized in a region \(\Omega \subset \Lambda\), we introduce the \emph{localized SLT interaction norm}
\begin{equation}
    \label{eq:interaction-norm-localization}
    \norm{\Phi}_{b,p; \Omega}
    :=
    \sup_{z \in \Lambda} \sumstack{Z\subset \Lambda:\\z\in Z}
    \frac{\norm{\Phi(Z)}}{\chi_{b,p}\pdiam{Z} \, \chi_{b,p}\pdist{z,\Omega}}
\end{equation}
and refer to operators with the property that \(\norm{\Phi}_{b,p; \Omega} \le C\) as \emph{\((b,p,\Omega)\)-localized SLT operators}.
Whenever it is clear from the context, or irrelevant for the discussion, we will often also omit the arguments \((b,p)\), and simply refer to \emph{\(\Omega\)-localized SLT operators}.
Finally, observe that \(\norm{\Phi}_{b,p; \Lambda} = \norm{\Phi}_{b,p}\).

The following simple lemma, whose proof is given in Appendix~\ref{app:proof-SLTlocal}, relates conceptually easier notions of locality of SLT operators to boundedness of the norm~\eqref{eq:interaction-norm-localization}.

\begin{lem} \label{lem:SLTlocal}
    Let \(A\) be an SLT operator stemming from an interaction \(\Phi\), for which we assume that \(\norm{\Phi}_{b,p} \le C\) for some constant \(C >0\), and let \(\Omega \subset \Lambda\).
    \begin{itemize}
        \item[(i)] Let \(A\) be \emph{strictly \(\Omega\)-localized}, i.e.~\(\Phi(Z) = 0\) whenever \(Z \cap (\Lambda\setminus\Omega) \neq \emptyset\).
            Then it holds that \(\norm{\Phi}_{b,p; \Omega} \le C\)
        \item[(ii)] Let \(A\) be \emph{strongly \(\Omega\)-localized}, i.e.~\(\Phi(Z) = 0\) whenever \(Z \cap \Omega = \emptyset\).
            Then it holds that \(\norm{\Phi}_{b/2,p; \Omega} \le C\).
    \end{itemize}
\end{lem}

\section{Main result: Response theory for locally gapped systems}
\label{sec:result}

In this section, we formulate our main result, the validity of response theory to any order, in Theorem~\ref{thm:linear-response} in Section~\ref{sec:response-theory-main}.
Its proof, based on the NEASS construction, is given in Section~\ref{subsec:NEASS}.

\subsection{Response theory}
\label{sec:response-theory-main}

We recall the assumption \nameref{ass:localGAP} from Section~\ref{sec:intro} of a local gap, now formulated preparatory of our main result, Theorem~\ref{thm:linear-response} below.

\begin{assumption}{(LDG\textsubscript{main})}{Local dynamical gap condition -- formal version}
    \label{ass:localGAP_main}
    \noindent
    We say that an equilibrium state \(\rho_0\) of the SLT-operator \(H_0\), i.e.~with \([H_0, \rho_0] = 0\), is \emph{locally dynamically gapped} of size at least \(g > 0\) in a region \(\Lambdag \subset \Lambda\) with respect to \(C_\gap\), \(b\), \(p > 0\) and \(\ell\in \N_0\), if and only if for all \(X, Y\subset \Lambda\) and \(A \in \alg_X\) and \(B \in \alg_Y\), it holds that
    \begin{equation}
        \label{eq:local gap_main}
        \abs*{\expectation[\big]{\commutator*{\calL_{H_0} \circ \calI_{H_0}\Ab{A} - A , B}}_{\rho_0}}
        \le
        \begin{multlined}[t]
            C_\gap \, \norm{A} \, \norm{B} \, \big[ \diam(X) + \diam(Y)\big]^\ell
            \\\times
            \chi_{b,p}\paren[\big]{\dist{X, \Lambda \setminus \Lambdag} + \dist{Y, \Lambda \setminus \Lambdag}}
            .
        \end{multlined}
        \qedhere
    \end{equation}
\end{assumption}

We can now formulate our main result.
In a nutshell, it says the following: We have validity of \emph{response theory to all orders} under the assumption of \nameref{ass:localGAP_main} for a perturbation localized in \(\Lambdap\) -- provided that the distance to \(\Lambda \setminus \Lambdag\) is sufficiently large compared to \(\abs{\log(\epsi)}^{1/q}\), where \(\epsi>0\) is the strength of the perturbation and \(q>0\) is some small constant.

The proof of Theorem~\ref{thm:linear-response} is given at the end of Section~\ref{subsec:NEASS}.

\begin{theorem}[Response theory to all orders]
    \label{thm:linear-response}
    Fix \(n,m \in \N\) and let \(d\in \N\), \(\Cvol>0\), \(b>0\), \(p\in \intervaloc{0,1}\), \(\Cint>0\) and \(g > 0\), \(C_\gap>0\), \(\ell\in \N_0\), and \(C_{\mathrm{switch}} > 0\).
    Take any \(q \in (0, p)\).
    Then there exist a constant \(C_{n,m}> 0\) (in particular depending on \(n\) and \(m\)) such that for all lattices \(\Lambda\in \mathcal{G}(d,\Cvol)\) (recall~\eqref{eq:Gset}), subsets \(\Lambdap\subset \Lambda\) and SLT-operators \(H_0\) and \(V\), with corresponding interactions that satisfy \(\norm{\Phi_{H_0}}_{b,p}<\Cint\) and \(\norm{\Phi_{V}}_{b,p;\Lambdap}<\Cint\), respectively, the following holds:

    Assume that the equilibrium state \(\rho_0\) of \(H_0\) is locally dynamically gapped in \(\Lambdag\) of size at least \(g\) and with respect to \(C_\gap\), \(b\), \(p\) and \(\ell\) according to Assumption \nameref{ass:localGAP_main}.
    Let \(Y\subset \Lambda\) and \(B\in \alg_Y\).
    Then there exist response coefficients \(\sigma_{B,j}\) in the following sense:
    For \(\varepsilon>0\) and smooth switching function \(f\colon \R \to \intervalcc{0,1}\) satisfying \(f(t) = 0\) for \(t \le -1\), \(f(t) = 1\) for \(t \ge 0\), and \(\norm{f}_{C^{\mathfrak{C}_{n,m}}(\R)}<C_\mathup{switch}\) with \(\mathfrak{C}_{n,m} := \lceil m(n + 1 +(2d+\ell)/p) \rceil\), consider the time-dependent Hamiltonian
    \begin{equation}
        \label{eq:Hamiltonian-switching}
        H_\epsi(t) := H_0 + \epsi \, f(t) \, V
        .
    \end{equation}
    Let \(\rho^{\epsi, \eta, f}(t)\) be the solution of the time-dependent adiabatic Schrödinger equation
    \begin{equation}
        \label{eq:Schrödinger-equation-switching}
        \I \, \eta \, \frac{\D }{\D t} \rho^{\epsi, \eta, f}(t)
        =
        \commutator[\big]{H_\epsi(t) , \rho^{\epsi, \eta, f}(t)}
    \end{equation}
    with adiabatic parameter \(\eta \in (0,1]\) and initial datum \(\rho^{\epsi, \eta, f}(t) = \rho_0\) for all \(t \le -1\).

    Then, the response to the perturbation, \(
        \Sigma_B^{\epsi, \eta, f}(t)
        :=
        \expectation{B}_{\rho^{\epsi, \eta, f}(t)}
        - \expectation{B}_{\rho_0}
    \), satisfies
    \begin{equation}
        \label{eq:linearresponse}
        \sup_{\eta \in \intervalcc[\big]{\epsi^m , \epsi^{1/m}}} \abs*{\Sigma_B^{\epsi, \eta, f}(t) - \sum_{j=1}^{n} \epsi^j \sigma_{B,j}}
        \le
        \begin{multlined}[t]
            C_{n,m} \, \norm{B} \, \diam(Y)^{(3+n)d+ \ell} \paren[\big]{1 + t}^{(2d+\ell)/p + 1} \, \epsi^{n+1}
            \\\times
            \paren*{1 + \E^{-\dist{\Lambdap, \Lambda \setminus \Lambdag}^{q} - (\mathfrak{C}_{n,m} + 1) \log(\epsi)}}\,,
        \end{multlined}
    \end{equation}
    for all \(t\geq 0\).
    The first order coefficient is given by \emph{Kubo's formula}~\eqref{eq:Kubo}.
\end{theorem}

Our result can also be extended to infinite systems.

\begin{rmk}[Extension to infinite systems]
    \label{rmk:mainresult}
    Following the arguments from~\cite{HT2020unif}, it is reasonably straightforward to extend our result to the case of infinite systems.
    More precisely, in order to do so, we need to
    \begin{itemize}
        \item consider \(\Lambda\) to be part of a sequence of graphs exhausting an infinite graph \(\Gammar\), e.g.~\(\Lambda \equiv \Lambda_k := \{-k,\dotsc,k\}^d \subset \Z^d\) and \(\Gammar := \Gamma\);
        \item assume that the interactions associated to \(H_0\) and \(V\) \emph{have a thermodynamic limit} (see~\cite[Definition 2.1]{HT2020unif});
        \item require that the sequence of equilibrium states \(\rho_0 \equiv \rho_0^{\Lambda}\) satisfies the local dynamical gap condition \nameref{ass:localGAP_main} with constants independent of \(\Lambda\) and converges (in the weak\(^*\) sense in the dual to the algebra of quasi-local observables, see~\cite[Section 2.5]{HT2020unif}) as \(\Lambda \nearrow \Gammar\);
        \item and suppose that also the perturbation region \(\Lambdap\) as well as the gapped region \(\Lambdag\) \emph{converge} (in a suitable sense) to some \(\Gammar^{\mathrm{pert}}, \Gammar^\gap \subset \Gammar\), respectively, at least ensuring that \(
                \dist{\Lambdap, \Lambda \setminus \Lambdag} \to \dist{\Gammar^{\mathup{pert}}, \Gammar \setminus \Gammar^\gap}
            \) as \(\Lambda \nearrow \Gammar\).
    \end{itemize}
    Then,~\eqref{eq:linearresponse} also holds for the infinite system but with \(
        \dist{\Gammar^{\mathup{pert}}, \Gammar \setminus \Gammar^\gap}
    \) in the exponential.

    In the special case, where the gapped region exhausts the entire graph, i.e.~\(\Lambdag \nearrow \Gammar\), we find~\eqref{eq:linearresponse} \emph{without} the additional exponential error term.
    This corresponds to \(H_0\) having a \emph{gap in the bulk} (see~\cite{HT2020bulk, HW2022} and also the discussion around~\eqref{eq:GAPdecay} in Section~\ref{sec:localgap}), where in the finite systems \(H_0\) could have gap closing edge modes (see Figure~\ref{fig:edge}).
\end{rmk}

\subsection{Non-equilibrium almost stationary states and proof of Theorem~\ref{thm:linear-response}}
\label{subsec:NEASS}

The main underlying idea of our justification of linear response theory is the construction of so-called \emph{non-equilibrium almost-stationary states (NEASS)}~\cite{MT2019,Teufel2020,HT2020,HT2020unif,HT2020bulk} for the dynamics of the perturbed Hamiltonian
\begin{equation}
    \label{eq:perturbed-Hamiltonian-without-switching}
    H_\epsi := H_0 + \epsi V
    .
\end{equation}
More precisely, for every \(n \in \N\) there exists%
\footnote{%
    \label{ftn:resum}
    Following the \emph{resummation procedure} in~\cite[Appendix E]{HT2020unif}, one could even construct a single (i.e.~\(n\) independent) state \(\Pi^{\epsi}\), which, for every fixed \(n\), has the same properties as \(\Pi_n^{\epsi}\).
    We will, however, refrain from doing so for brevity of the presentation.
}
a state \(\Pi_n^\epsi\), which is obtained from the equilibrium state \(\rho_0\) of \(H_0\) by a unitary transformation with a small SLT generator, in such a way, that it is almost stationary under the dynamics generated by~\eqref{eq:perturbed-Hamiltonian-without-switching}.

Similarly to~\cite[Theorem~3.1]{Teufel2020}, one can obtain Proposition~\ref{prop:NEASS} by following the proof of Proposition~\ref{prop:adswitch}, given in Section~\ref{subsec:adpert}, and combining it with Proposition~\ref{prop:expansionNEASS}.
As the arguments in~\cite[Theorem~3.1]{Teufel2020} are rather standard, we omit the details for brevity.

\begin{prop}[Non-equilibrium almost-stationary states]
    \label{prop:NEASS}
    Fix \(n\in \N\) and let \(d\in \N\), \(\Cvol>0\), \(b>0\), \(p\in \intervaloc{0,1}\), \(\Cint>0\) and \(g > 0\), \(C_\gap>0\), \(\ell\in \N_0\).
    Take any \(q \in (0, p)\).
    Then there exist a constant \(C_{n}> 0\) (in particular depending on \(n\))
    such that for all lattices \(\Lambda\in \mathcal{G}(d,\Cvol)\) (recall~\eqref{eq:Gset}), subsets \(\Lambdap\subset \Lambda\) and SLT-operators \(H_0\) and \(V\), with corresponding interactions that satisfy \(\norm{\Phi_{H_0}}_{b,p}<\Cint\) and \(\norm{\Phi_{V}}_{b,p;\Lambdap}<\Cint\), respectively, the following holds:

    Assume that the equilibrium state \(\rho_0\) of \(H_0\) is locally dynamically gapped in \(\Lambdag\) of size at least \(g > 0\) and with respect to \(C_\gap\), \(b\), \(p\) and \(\ell\) according to Assumption \nameref{ass:localGAP_main}.
    Then, there exists a sequence \((A_\mu)_{\mu \in \N}\) of SLT operators, which are \((1, p', \Lambda^{\mathrm{pert}})\)-localized for any \(p' < p\), such that the state
    \begin{equation}
        \label{eq:NEASS}
        \Pi_n^\epsi
        :=
        \evol{\I \epsi S_n^\epsi}{\rho_0}
        \quad \text{with} \quad
        S_n^\epsi
        :=
        \sum_{\mu= 1}^n \epsi^{\mu-1} A_\mu
    \end{equation}
    is almost-stationary for the dynamics generated by \(H_\epsi = H_0 + \epsi V\) in the following sense: Let \(\rho^\epsi(t)\) be the solution to the Schrödinger equation
    \begin{equation}
        \label{eq:prop-NEASS-Schrödinger-equation}
        \I \frac{\D }{\D t} \rho^\epsi(t)
        =
        \commutator[\big]{H_\epsi , \rho^\epsi(t)}
        \quad \text{with} \quad
        \rho^\epsi(0)
        =
        \Pi_n^\epsi
        .
    \end{equation}

    Under these conditions, for all \(B\in \alg_Y\) with \(Y\subset \Lambda\) and \(t \geq 0\), it holds that
    \begin{equation}
        \label{eq:timeindepNEASS}
        \abs*{
            \expectation{B}_{\rho^\epsi(t)}
            - \expectation{B}_{\Pi_n^\epsi}}
        \le
        \begin{multlined}[t]
            C_{n,m} \, \norm{B} \diam(Y)^{3d+\ell} \, \abs{t} \paren[\big]{1 + \abs{t}^{(2d+\ell)/p}} \, \epsi^{n+1}
            \\\times
            \paren[\Big]{1 + \E^{-\dist{\Lambdap, \Lambda \setminus \Lambdag}^{q} - (n+1) \log(\epsi)}}
            .
        \end{multlined}
    \end{equation}
\end{prop}

We point out that, similarly to~\cite[Theorem~3.1]{Teufel2020}, one can improve the bound in~\eqref{eq:timeindepNEASS} by rescaling every \(t\) by \(\epsi^m\) for some \(m \in \N\), at the cost of increasing the constant in front of \(\log (\epsi)\).

From a technical perspective, our proof of response theory for locally gapped quantum spin systems, Theorem~\ref{thm:linear-response}, rests on the following two propositions, the proof of which shall be given in Section~\ref{sec:proof}.
The first, Proposition~\ref{prop:adswitch}, states that the NEASS constructed in Proposition~\ref{prop:NEASS} is also almost-stationary under the dynamics generated by the time-dependent Hamiltonian~\eqref{eq:Hamiltonian-switching}
whose perturbation \(V\) is turned on by the switching function \(f\) on the adiabatic time scale \(1/\eta\) (see, e.g.,~\cite[Prop.~3.2]{Teufel2020}).
The proof of Proposition~\ref{prop:adswitch} is given in Section~\ref{subsec:adpert}.

\begin{prop}[Adiabatic switching and the NEASS]
    \label{prop:adswitch}
    Under the assumptions of Theorem~\ref{thm:linear-response} (in particular recalling~\eqref{eq:Schrödinger-equation-switching}) it holds that
    \begin{equation}
        \abs*{
            \expectation{B}_{\rho^{\epsi, \eta, f}(t)}
            - \expectation{B}_{\Pi_n^\epsi}
        }
        \le
        \begin{multlined}[t]
            C \, \norm{B} \diam(Y)^{3d+\ell} \paren[\big]{1 + t}^{(2d+\ell)/p+1} \, \frac{\epsi^{n+1} + \eta^{n+1}}{\eta^{(2d+\ell)/p+1}}
            \\\times
            \paren[\bigg]{1 + \E^{-\dist{\Lambdap, \Lambda \setminus \Lambdag}^{q} - (n+1) \log(\epsi)}}
            ,
        \end{multlined}
    \end{equation}
    for \(t \geq 0\), where \(\Pi_n^\epsi\) is the NEASS~\eqref{eq:NEASS} constructed in Proposition~\ref{prop:NEASS}.
\end{prop}

For our application to response theory, it is important to have an explicit expansion of expectation values in the NEASS in powers of \(\epsi\) with coefficients given by expectations in the unperturbed equilibrium state, the linear term constituting the celebrated \emph{Kubo formula}.
This is the content of the following proposition, whose proof works in the exact same way as in~\cite[Theorem 3.3]{MT2019} or~\cite[Proposition 5.2]{Teufel2020} and is hence omitted.

\begin{prop}[Asymptotic expansion of the NEASS]
    \label{prop:expansionNEASS}
    Under the assumptions of Proposition~\ref{prop:NEASS}, there exist linear maps \(\mathcal{K}_j \colon \Alg \to \Alg\), \(j \in \N\), given by nested commutators with the \(A_\mu\)'s in~\eqref{eq:NEASS}, such that for \(n \ge m\) it holds that
    \begin{equation}
        \label{eq:NEASSexpand}
        \abs[\bigg]{
            \expectation{B}_{\Pi_n^\epsi}
            - \sum_{j=1}^{m} \epsi^j \expectation{\mathcal{K}_j\Ab{B}}_{\rho_0}
        }
        \le
        C_{n,m} \, \epsi^{m+1} \, \norm{B} \, \abs{Y}^m
        \, \chi_{b',p}\pdist{\Lambdap,Y}
        ,
    \end{equation}
    with \(b',p\) the parameters of the localization of the \(A_\mu\) given in Proposition~\ref{prop:NEASS}.
    The first two orders of the expansion~\eqref{eq:NEASSexpand} are explicitly given by
    \begin{equation}
        \expectation{\mathcal{K}_0\Ab{B}}_{\rho_0}
        =
        \expectation{B}_{\rho_0}
        \quadtext{and}
        \expectation{\mathcal{K}_1\Ab{B}}_{\rho_0}
        =
        - \I \expectation*{\commutator[\big]{\calI_{H_0,g}\Ab{V} , B}}_{\rho_0}
        ,
    \end{equation}
    where \(\calI_{H_0,g}\Ab{\interpunct}\) is the inverse Liouvillian from~\eqref{eq:invliouintro}.
\end{prop}

We can finally give the proof of Theorem~\ref{thm:linear-response}.

\begin{proof}[Proof of Theorem~\ref{thm:linear-response}] Armed with Propositions~\ref{prop:adswitch}--\ref{prop:expansionNEASS}, the proof of Theorem~\ref{thm:linear-response} follows by a simple application of the triangle inequality when applying Proposition~\ref{prop:adswitch} with \(n \to \tilde{n} := \lceil m(n+1+(2d+\ell)/p)\rceil\) and Proposition~\ref{prop:expansionNEASS} for \(n \to \tilde{n}\) and \(m \to n\); cf.~\cite[Theorem~4.1]{Teufel2020}.
\end{proof}

\section{How to formulate a local gap condition?}
\label{sec:localgap}
Formulating a (i) physically \emph{and} mathematically transparent, (ii) practically applicable, and (iii) reasonably restrictive condition of a \emph{local gap} for a Hamiltonian \(H\) is a non-trivial task.
In this section, we discuss several different possible approaches to do so.
In principle, there are two main ways: Either we compare \(H\) to a globally gapped reference Hamiltonian \(H_*\) (\emph{extrinsic formulation}, see Section~\ref{subsec:extrinsic}), or the condition involves the original Hamiltonian \(H\) alone (\emph{intrinsic formulation}, see Section~\ref{subsec:intrinsic}).
We shall first list all the different variants and afterwards compare their respective features in view of the above listed three requirements and explain their relations in Section~\ref{subsec:compare}.
Finally, in Section~\ref{subsec:examples} we discuss two exemplary systems, which we prove to have a local dynamical gap in the sense of Assumption \nameref{ass:localGAP_main}.

Throughout the entire section (except differently stated explicitly), we will use \(C, \ell, b,p\) as generic constants satisfying \(C >0\), \(\ell \ge 0\), \(b >0\), and \(p \in (0,1]\).
Their precise value might change from line to line and it only depends on the model parameters, i.e.~the interaction norms~\eqref{eq:interaction-norm}--\eqref{eq:interaction-norm-localization} of the involved SLT Hamiltonians, the lattice parameters in~\eqref{eq:Gset}, or the gap size \(g>0\) of a reference Hamiltonian.

\subsection{Extrinsic local gap conditions}
\label{subsec:extrinsic}
In this section, we describe several ways of \emph{extrinsically} expressing that an equilibrium state \(\rho\) (e.g.~the ground state) of a Hamiltonian \(H\) is locally gapped.
We call these formulations of a local gap condition \emph{extrinsic}, since the Hamiltonian \(H\) and equilibrium state \(P\) of interest are compared to another reference Hamiltonian \(H_*\), called the \emph{parent Hamiltonian}, with equilibrium state \(\rho_*\), which is globally gapped.
Throughout this section, for simplicity of the presentation, we will assume that both \(\rho\) and \(\rho_*\) have rank one.

The common core of all the different ways to extrinsically formulate a local gap conditions, is to assume that \(H\) and \(H_*\) differ only locally: That is, \(H_*\) and \(H\) are related as
\begin{equation}
    \label{eq:HHJ}
    H = H_* + J
\end{equation}
where \(J\) is some SLT operator, which is localized in \(\Lambda \setminus \Lambdag\) (e.g., in one of the senses mentioned in Lemma~\ref{lem:SLTlocal}).
This form~\eqref{eq:HHJ} of \(H\) and \(H_*\) differing only locally is then passed on to their equilibrium states, \(\rho\) and \(\rho_*\), respectively.
More precisely, we will assume that local perturbations of \(H_*\) perturb its equilibrium state \(\rho_*\) only locally.
Whenever this holds, one (usually) says that \(H_*\) satisfies the \emph{local perturbations perturb locally (LPPL)} principle, which has been shown to be the case for \emph{ground states} in several contexts~\cite{RS2015,HTW2022,BRDF2021}.

\paragraph{Local gap via \enquote*{classical} LPPL}
The simplest way to formulate an \emph{extrinsic} local gap condition for \(H\) given by~\eqref{eq:HHJ} with ground state \(\rho\), is to assume that the globally parent gapped Hamiltonian \(H_*\) with ground state \(\rho_*\) satisfies a strong form%
\footnote{%
    According to the terminology introduced in~\cite{HTW2022}, calling an LPPL principle \emph{strong}, simply means, that the perturbation \(J\) is \emph{not} required to leave the spectral gap of \(H_*\) open -- contrary to earlier works~\cite{BMNS2012,RS2015} on LPPL\@.
}
of the LPPL principle.
This means that, for \(A \in \Alg_{X}\), we have
\begin{equation}
    \label{eq:LPPL}
    \abs[\big]{\expectation{A}_{\rho} - \expectation{A}_{\rho_*}}
    \le
    C \, \norm{A} \, \diam(X)^\ell \, \chi_{b,p}\pdist{X, \Lambda \setminus \Lambdag}
    .
\end{equation}
An estimate of the form~\eqref{eq:LPPL} has been shown to hold for weakly interacting spin systems in~\cite{HTW2022} (based on ideas of~\cite{Yarotsky2005}) and -- in a similar way -- in~\cite{BRDF2021}.
Moreover, in~\cite[Theorem~4.1]{HT2020bulk}, such type of assumption was used to formulate an adiabatic theorem for large but finite systems with a gap in the bulk.
In this case (cf.~\cite[eq.~(4.1)]{HT2020bulk}), the additional error term (compared to the case of having a global gap) in the response theory expansion~\eqref{eq:linearresponseintro} becomes \(
    \mathcal{O}\paren*{\paren[\big]{\epsi^{C} \, \dist{X, \Lambda \setminus \Lambdag}}^{-\infty}}
\).
Moreover, using that \(H_*\) and \(J\) are SLT operators,~\eqref{eq:LPPL} implies that \nameref{ass:localGAP_main} holds \emph{but} with the argument of \(\chi_{b,p}\) in~\eqref{eq:local gap_main} being the \emph{minimum} of the distances (or -- up to changing constants -- the distance of the \emph{union} of the two supports to \(\Lambda \setminus \Lambdag\)) instead of their sum (i.e.~their maximum).
Hence, the decay in \(\dist{\Lambdap, \Lambda\setminus \Lambdag}\), emerging in the course of proving Proposition~\ref{prop:adswitch} in Section~\ref{sec:proof}, will eventually be lost.

In conclusion, although this way~\eqref{eq:LPPL} of saying that \(H\) has a local gap, is quite simple, the resulting error terms are considerably bad and, moreover, it does not quite allow for tracking \(\dist{\Lambdap, \Lambda\setminus \Lambdag}\)-decays.

\paragraph{Trace norm LPPL}
One could even further strengthen the LPPL assumption~\eqref{eq:LPPL} to a \emph{trace norm LPPL}.
This means, that \(\rho\) and \(\rho_*\) are not only close in the weak\(^*\) sense~\eqref{eq:LPPL}, but we have
\begin{equation}
    \label{eq:LPPLtrace}
    \norm{(\rho-\rho_*) \, A}_{\mathup{tr}}
    + \norm{A \, (\rho-\rho_*) }_{\mathup{tr}}
    \le
    C \norm{A} \, \diam(X)^\ell \, \chi_{b,p}\paren[\big]{\dist{X, \Lambda \setminus \Lambdag}}
    ,
\end{equation}
where \(\norm{\interpunct}_{\mathup{tr}} := \trace{\abs{\interpunct}}\) denotes the trace norm.
We remark, that surely~\eqref{eq:LPPLtrace} implies~\eqref{eq:LPPL}.
A somewhat weakened version of~\eqref{eq:LPPLtrace} is to ask that only the trace norm of the \emph{commutator} \([\rho-\rho_*, A]\) is small like in~\eqref{eq:LPPLtrace}, i.e.\
\begin{equation}
    \label{eq:LPPLtracecomm}
    \norm{[\rho-\rho_*, A]}_{\mathup{tr}}
    \le
    C \norm{A} \, \diam(X)^\ell \, \chi_{b,p}\paren[\big]{\dist{X, \Lambda \setminus \Lambdag}}
    ,
\end{equation}
As we show in Proposition~\ref{prop:mechforlocalgap} in Section~\ref{subsec:compare} (see also Figure~\ref{fig:impl}), both these ways,~\eqref{eq:LPPLtrace} and~\eqref{eq:LPPLtracecomm}, of saying that \(H\) has a local gap are sufficient to conditions for \(H\) having a local gap in the sense of Assumption~\nameref{ass:localGAP_main}.%
\footnote{%
    This fact is used for proving \nameref{ass:localGAP_main} (which is done in Section~\ref{app:LPPLlocalgap}) for one of the examples discussed in Section~\ref{subsec:examples} below.
}
However, it is a quite restrictive condition that \(\rho\) and \(\rho_*\) are close in such strong topology as~\eqref{eq:LPPLtrace} or~\eqref{eq:LPPLtracecomm}.

\paragraph{Intertwining automorphism}
Another, compared to~\eqref{eq:LPPL} strengthened, way of saying that \(\rho\) and \(\rho_*\) are close to each other, is to assume the following: There exists a intertwining norm-preserving \(*\)-automorphism \(\tau\) on \(\Alg\), i.e.~\(\expectation{\interpunct}_\rho = \expectation{\tau \Ab{\interpunct}}_{\rho_*}\), which satisfies
\begin{equation}
    \label{eq:AUTO}
    \norm[\big]{( \tau - \unit )\Ab{A}}
    \le
    C \, \norm{A} \, \diam(X)^\ell \, \chi_{b,p}\pdist{X,\Lambda \setminus \Lambdag}
    ,
\end{equation}
where \(\unit\) is the identity map.
Similarly to~\eqref{eq:LPPLtrace}, we surely have that~\eqref{eq:AUTO} implies~\eqref{eq:LPPL}.
Moreover,~\eqref{eq:AUTO} is a sufficient condition for having \nameref{ass:localGAP_main} (see Proposition~\ref{prop:mechforlocalgap} and Figure~\ref{fig:impl} below), as will be used for proving that both of the examples studied in Section~\ref{subsec:examples} below satisfy \nameref{ass:localGAP_main} (see Section~\ref{app:LPPLlocalgap}).

\subsection{Intrinsic local gap conditions}
\label{subsec:intrinsic}
In this section, we describe two ways of \emph{intrinsically} expressing that an equilibrium state \(\rho\) of a Hamiltonian \(H\) is locally gapped.
In contrast to the previous extrinsic formulation, they are called \emph{intrinsic}, as they do not refer to another (parent) Hamiltonian.

\paragraph{Local spectral gap conditions \emph{only} for ground states}
In case of \(\rho\) being the (unique) \emph{ground state}, a very natural way to connect the notion of locality with spectral analysis is to require that a variational condition characterizing the spectral gap is tested only locally.
Generally speaking, assume that \(\psi_0\) is the unique ground state with eigenvalue \(E_0\) of a local Hamiltonian \(H\) in some underlying physical space \(\Lambda\).
Then, a spectral gap above \(E_0\) of size (at least) \(g > 0\) is characterized by
\begin{equation}
    \label{eq:Gapchar}
    \inf_{\psi \perp \psi_0} \frac{\langle \psi, (H-E_0) \psi\rangle}{\langle \psi, \psi \rangle}
    \ge
    g
    .
\end{equation}
If the minimization in~\eqref{eq:Gapchar} is restricted a smaller set of \(\psi\)'s, i.e.~those which are localized to a region, say, \(\Lambdag \subset \Lambda\) in some appropriate sense, one could say that \(H\) is locally gapped in \(\Lambdag\).
Alternatively, for every fixed \(\psi\), the gap size \(g\) could be assumed to be non-constant but dependent on the distance of the support of \(\psi\) to the region \(\Lambdag\).

For quantum spin systems on the graph \(\Lambda\), it can easily be checked that the analog of~\eqref{eq:Gapchar} for an SLT Hamiltonian \(H\) with unique ground state \(\rho\) is
\begin{equation}
    \label{eq:spectralgap}
    \I \, \expectation{A^* \calL\Ab{A}}_\rho
    \ge
    g \, \paren[\Big]{
        \expectation{A^* A}_\rho - \abs[\big]{\expectation{A}_\rho}^2
    }
    ,
\end{equation}
for all observables \(A \in \Alg\), where \(\calL \Ab{\interpunct} := -\I \, [H, \interpunct]\) denotes the Liouvillian.
We now give two options to turn~\eqref{eq:spectralgap} into a local gap condition.

As a first option, one could require that the \emph{gap size \(g\) decays} as the support \(X\) of the observable \(A \in \Alg_X\) approaches the complement of \(\Lambdag\), e.g.~as
\begin{equation}
    \label{eq:GAPdecay}
    \I \, \expectation{A^* \calL\Ab{A}}_\rho
    \ge
    g
    \, \paren*{
        1- C \diam(X)^\ell \, \chi_{b,p}\pdist{X, \Lambda \setminus \Lambdag}
    }
    \, \paren[\Big]{
        \expectation{A^* A}_\rho - \abs[\big]{\expectation{A}_\rho}^2
    }
    .
\end{equation}
In case that \(\rho = \ket{\psi} \bra{\psi}\) with \(\psi\) being a \emph{product state}, variants of Proposition 14 and Lemma 15 in~\cite{YSL24} can be used to show the following: The gap decay condition~\eqref{eq:GAPdecay} implies a slightly weakened version (see~\eqref{eq:LDGweak_discuss} below) of our local dynamical gap condition; see Proposition~\ref{prop:mechforlocalgap}~(viii).

We further remark that, if one is interested in \emph{taking the thermodynamic limit}, \(\Lambda \nearrow \Gammar\) for some infinite graph \(\Gammar\), and \(\Lambdag \nearrow \Gammar\) in this limit (e.g.~in the scenario of gap-closing edge modes), then~\eqref{eq:GAPdecay} yields a \emph{gap in the bulk} of the naturally associated infinite system (see~\cite[Remark~4]{HW2022}).%
\footnote{%
    More precisely, this also requires that an interaction \(\Phi = \Phi_H\), associated to the SLT Hamiltonian \(H\), \emph{has a thermodynamic limit} in a suitable sense (see~\cite[Definition 2.1]{HT2020unif},~\cite[Definition~3.1]{HT2020bulk}, and~\cite[Definitions~2~and~5]{HW2022}).
    In this case, the linear functional \(A \mapsto \trace{PA}\) converges in weak\(^*\) sense to a \emph{state} on the \(C^*\)-algebra of quasi-local observables on \(\Gammar\).
    A \emph{gap in the bulk} then means, that the naturally associated GNS Hamiltonian of the infinite system has a spectral gap above zero.
}

Another, compared to~\eqref{eq:GAPdecay} weaker, option is to include a separate \emph{additive} error term on the lhs.~of~\eqref{eq:spectralgap}, e.g.~as
\begin{equation}
    \label{eq:COERdefect}
    \I \, \expectation{A^* \calL\Ab{A}}_\rho
    \ge
    g \, \paren[\Big]{
        \expectation{A^* A}_\rho - \abs[\big]{\expectation{A}_\rho}^2
    }
    - C \, \norm{A}^2 \, \diam(X)^\ell \, \chi_{b,p} \paren[\big]{\dist{X, \Lambda \setminus \Lambdag}}
    ,
\end{equation}
which we call \emph{defective coercivity} for the following reason.

The global gap characterization can equivalently be rewritten as
\begin{equation*}
    \I \, \expectation{\tilde{A}^* \calL\Ab{\tilde{A}}}_\rho
    \ge
    g \, \expectation{\tilde{A}^* \tilde{A}}_\rho
    \quad \text{with} \quad
    \tilde{A}
    :=
    A - \expectation{A}_\rho
    ,
\end{equation*}
which means that on \(\Alg^\perp:= \Set{\rho}^\perp = \Set{ \tilde{A} \given A \in \Alg} \subset \Alg\) the bounded sesquilinear form
\begin{equation*}
    \mathcal{B}\colon \Alg^\perp \times \Alg^\perp \to \C \,, (A,B) \mapsto \I \, \expectation{A^* \calL\Ab{B}}_P
\end{equation*}
is \emph{coercive} with respect to the semi-norm \(\vertiii{A}:= \sqrt{\expectation{A^* A}_\rho}\) on \(\Alg^\perp\).
Hence,~\eqref{eq:COERdefect} expresses some defect in the original coercivity of~\eqref{eq:spectralgap}.
We remark that, at least morally, one could use the Lax-Milgram theorem to deduce existence of an inverse of the Liouvillian \(\calL\) given such a coercivity estimate.
However, the problem is that the inverse obtained in this way does not necessarily have any nice locality properties, which are crucially used for practical purposes.

Finally, we point out that the defective coercivity~\eqref{eq:COERdefect} is implied by the `classical' strong LPPL~\eqref{eq:LPPL}.
More precisely, assume there exists a parent Hamiltonian \(H_*\) which is related to \(H\) like in~\eqref{eq:HHJ}.
Then, if the globally gapped ground state \(\rho_*\) of \(H_*\) satisfies the strong form~\eqref{eq:LPPL} of LPPL, then~\eqref{eq:COERdefect} holds -- modulo adjusting the constants \(b\), \(p\), \(\ell\), and \(C\); see Proposition~\ref{prop:mechforlocalgap} and Figure~\ref{fig:impl} below.

\paragraph{Local dynamical gap condition}
The -- at least in view of our application -- most important property of an SLT Hamiltonian with a gapped part of its spectrum, is that (the offdiagonal part of) its Liouvillian can be \emph{locally} inverted; see, e.g.,~\eqref{eq:invertintro} in Section~\ref{subsec:LDG} and also~\cite[Appendix C]{Teufel2020}.
In these applications, roughly said, the local invertibility of the Liouvillian guarantees that the effect of perturbations remains local.
More precisely, it allowed to prove automorphic equivalence of gapped ground states~\cite{BMNS2012, MO2020} and adiabatic theorems in cases where the perturbation is not allowed to close the gap~\cite{BRF2018,MT2019} and later also for gap closing perturbations~\cite{Teufel2020,HT2020,HT2020unif,HT2020bulk}.

As already noted in Section~\ref{subsec:LDG}, the local invertibility of the Liouvillian is in fact \emph{equivalent} to the Hamiltonian having a gapped part of its spectrum (see Proposition~\ref{prop:GDG}), which yields a \emph{dynamical} characterization of a \emph{spectral} property.
A local version of this feature is formulated in our \emph{local dynamical gap condition} in Assumptions \nameref{ass:localGAP} and \nameref{ass:localGAP_main}: Let \(\calI \equiv \calI_{H, g}\) denote the inverse Liouvillian from~\eqref{eq:invliouintro} with gap size \(g > 0\).
Then, an equilibrium state \(\rho\) (\emph{not} necessarily the ground state!) of \(H\) is locally dynamically gapped, if for all observables \(A \in \alg_X\) and \(B \in \alg_Y\) localized in \(X\subset \Lambda\) and \(Y \subset \Lambda\), it holds that
\begin{equation}
    \label{eq:LDGdiscuss}
    \abs*{\expectation[\big]{\commutator*{\calL \circ \calI\Ab{A} - A , B}}_{\rho}}
    \le
    \begin{multlined}[t]
        C \norm{A} \, \norm{B} \, \big[\diam(X)+ \diam(Y)\big]^\ell
        \\\times
        \chi_{b,p}\paren*{\dist{X, \Lambda \setminus \Lambdag}+ \dist{Y, \Lambda \setminus \Lambdag}} \,.
    \end{multlined}
\end{equation}
In the special case of \(\ell = 0\), by taking a supremum over all observables \(B\) with \(\norm{B} \le 1\) in~\eqref{eq:LDGdiscuss}, we find that -- for \(\ell = 0\) -- our local dynamical gap condition~\eqref{eq:LDGdiscuss} is actually \emph{equivalent} to
\begin{equation*}
    \norm[\big]{[\calL\circ \calI\Ab{A} - A, \rho]}_{\tr}
    \le
    C \norm{A} \, \chi_{b,p}\paren*{\dist{X, \Lambda \setminus \Lambdag}}
    .
\end{equation*}
Here, we additionally used cyclicity of the trace together with \(\sup_{\norm{D} \le 1} \abs{\trace{CD}} = \norm{C}_{\tr}\).

Moreover, we point out that, while the rhs.~of~\eqref{eq:LDGdiscuss} is obviously symmetric in \(A\) and \(B\), the lhs.~is as well.
This follows by rewriting (recall~\eqref{eq:invliouintro})
\begin{equation}
    \label{eq:Jdef}
    \calL_H \circ \calI_{H,g} \Ab{A} - A
    =
    \int_{\R} \D t \, w_g(t) \, \evol{\I H t}{A}
    =:
    \calJ_{H,g}\Ab{A}
    =
    \calJ\Ab{A}
\end{equation}
and using \([H,\rho] = 0\) together with the symmetry \(w_g(t) = w_g(-t)\).

A slightly weakened (asymmetric) version of~\eqref{eq:LDGdiscuss} would be to require that for all \(X\subset \Lambda\) satisfying \(\diam(X) \le \dist{X, \Lambda\setminus\Lambdag}^\beta\) (for some \(\beta > 0\)) and observables \(A \in \alg_X\), and \(Y \subset \Lambda\) and observables \(B \in \alg_Y\), it holds that
\begin{equation}
    \label{eq:LDGweak_discuss}
    \abs*{\expectation[\big]{\commutator*{\calL \circ \calI\Ab{A} - A , B}}_{\rho}}
    \le
    C \, \norm{A} \, \norm{B}
    \, \diam(Y)^\ell
    \, \chi_{b,p}\paren[\big]{\dist{X, \Lambda \setminus \Lambdag} }
    .
\end{equation}
This version will be introduced as the \emph{weakened local gap condition} \nameref{ass:localGAPweak} in Section~\ref{subsec:LDGweak} below.
Such an assumption will in fact be sufficient for proving our main result in Theorem~\ref{thm:linear-response}.

Finally, we remark that another symmetric (in \(A\) and \(B\)) bound on the lhs.~of~\eqref{eq:LDGdiscuss} would be to replace the sum of the distances in the exponent in~\eqref{eq:LDGdiscuss} by \(\dist{X \cup Y, \Lambda \setminus \Lambdag}\) (i.e.~take the {minimum} of the distances to \(\Lambda\setminus \Lambdag\) instead of their maximum).
Using quasi-locality estimates for \(\calJ\Ab{\interpunct}\) defined in~\eqref{eq:Jdef} (see, e.g.,~\cite[Lemma 5.1]{NSY2019}), this bound could actually be proven as a consequence of the LPPL principle~\eqref{eq:LPPL}, provided that there is a globally spectrally gapped parent Hamiltonian \(H_*\) for \(H\).
Therefore, similarly to the paragraph below~\eqref{eq:LPPL}, using this assumption, the decay in \(\dist{\Lambdap, \Lambda\setminus \Lambdag}\), emerging in the course of proving Proposition~\ref{prop:adswitch} in Section~\ref{sec:proof}, would eventually be lost.

\subsection{Summary and comparison}
\label{subsec:compare}

In the previous two Sections~\ref{subsec:extrinsic}--\ref{subsec:intrinsic} we described several different ways of expressing that a Hamiltonian is locally gapped, distinguishing between \emph{extrinsic} (Section~\ref{subsec:extrinsic}) and \emph{intrinsic} (Section~\ref{subsec:intrinsic}) formulations.

While the extrinsic conditions are easy to formulate, they rely -- by definition -- on a reference (parent) Hamiltonian with a global gap satisfying some form of LPPL principle (cf.~\eqref{eq:LPPL}--\eqref{eq:AUTO}).
Since for a system of interest, it is not guaranteed to have such a well-understood parent Hamiltonian available, it is conceptually more desirable to formulate a local gap condition in an intrinsic way.
Or, in other words, saying that a system is (or behaves as if it were) \emph{locally} gapped should not only make sense relative to another \emph{globally} gapped system.

In the intrinsic category, we formulated two local \emph{spectral} gap conditions~\eqref{eq:GAPdecay}--\eqref{eq:COERdefect}, which, however, are (i) only meaningful for (non-degenerate) ground states and (ii) although mathematically clean, hardly applicable in physical problems in a direct way (apart from the connection in Proposition~\ref{prop:mechforlocalgap}~(viii)).
These two issues are then resolved by our local dynamical gap condition (LDG) in~\eqref{eq:LDGdiscuss} and its weakened version in~\eqref{eq:LDGweak_discuss}.

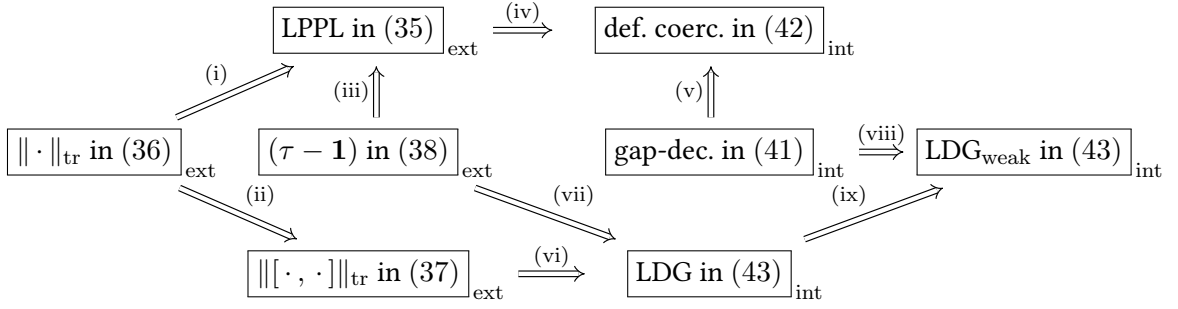
\begin{figure}[h]
    \begin{tikzcd}[column sep=small]
        & \boxed{\text{LPPL in} \ \eqref{eq:LPPL}}_{\, \mathrm{ext}} \ar[r, Rightarrow, "\mathrm{(iv)}"]
        & \quad \boxed{\text{def.~coerc.~in} \ \eqref{eq:COERdefect}}_{\, \mathrm{int}}
        \\
        \boxed{\norm{\interpunct}_{\tr} \ \text{in} \ \eqref{eq:LPPLtrace}}_{\, \mathrm{ext}} \ar[ur, Rightarrow, "\mathrm{(i)}"] \ar[dr, Rightarrow, "\mathrm{(ii)}"] \hspace{-4mm}
        & \boxed{(\tau - \unit) \ \text{in} \ \eqref{eq:AUTO}}_{\, \mathrm{ext}} \ar[u, Rightarrow, "\mathrm{(iii)}"]\ar[dr, Rightarrow, "\mathrm{(vii)}"]
        & \quad \boxed{\text{gap-dec.~in} \ \eqref{eq:GAPdecay}}_{\, \mathrm{int}} \ar[u,Rightarrow, "\mathrm{(v)}"] \ar[r, Rightarrow, "\mathrm{(viii)}"] & \boxed{\text{LDG\(_{\mathrm{weak}}\) in} \ \eqref{eq:LDGdiscuss}}_{\, \mathrm{int}}
        \\
        & \boxed{\norm{[\interpunct, \interpunct ]}_{\tr} \ \text{in} \ \eqref{eq:LPPLtracecomm}}_{\, \mathrm{ext}} \ar[r, Rightarrow, "\mathrm{(vi)}"]
        & \quad \boxed{\text{LDG in} \ \eqref{eq:LDGdiscuss}}_{\, \mathrm{int}} \ar[ur, Rightarrow, "\mathrm{(ix)}"]& \\
    \end{tikzcd}
    \caption{
        Implications among the various local gap conditions from Sections~\ref{subsec:extrinsic}--\ref{subsec:intrinsic}.
        The numbering refers to the precise statements in Proposition~\ref{prop:mechforlocalgap}.
    }
    \label{fig:impl}
\end{figure}

The following proposition formulates the implications shown in Figure~\ref{fig:impl} precisely.

\begin{prop}[Relations among the local gap conditions]
    \label{prop:mechforlocalgap}
    Fix \(g> 0\), \(d \in \N\), \(\Cvol> 0\), \(\tilde{b} > 0\), \(\tilde{p} \in (0,1]\), \(\Cint > 0\).
    Let \(\Lambda \in \mathcal{G}(d, \Cvol)\), \(\Lambdag \subset \Lambda\) and let \(H_*\) and \(J\) be SLT-operators with corresponding interactions satisfying \(\norm{\Phi_{H_*}}_{\tilde{b},\tilde{p}} < \Cint\) and \(\norm{\Phi_J}_{\tilde{b},\tilde{p};\Lambda \setminus \Lambdag} < \Cint\).
    Let \(\rho_*\) be an equilibrium state of \(H_*\) and \(\rho\) and equilibrium state of \(H := H_* + J\).

    Then the following implications hold (modulo adjusting the constants \(C,\ell, b\), and \(p\) in a way which only depends on \(g,d, \Cvol, \tilde{b}, \tilde{p}\), and \(\Cint\)):

    \begin{itemize}
        \item[(i)] Trace norm LPPL implies commutator trace norm LPPL, i.e.~we have \(\eqref{eq:LPPLtrace} \, \implies \,~\eqref{eq:LPPLtracecomm}\).
        \item[(ii)] Trace norm LPPL implies the usual LPPL, i.e.~we have \(\eqref{eq:LPPLtrace} \, \implies \,~\eqref{eq:LPPL}\).
        \item[(iii)] An intertwining automorphism implies the usual LPPL, i.e.~we have \(\eqref{eq:AUTO} \, \implies \,~\eqref{eq:LPPL}\).
        \item[(iv)] Additionally assuming that \(\rho_*\) is the unique ground states of \(H_*\) with a spectral gap of size at least \(g > 0\) above its ground state (see~\eqref{eq:spectralgap}), the usual LPPL implies defective coercivity, i.e.~\(\eqref{eq:LPPL} \, \implies \,~\eqref{eq:COERdefect}\).
        \item[(v)] A decaying gap size implies defective coercivity, i.e.~\(\eqref{eq:GAPdecay} \, \implies \,~\eqref{eq:COERdefect}\).
        \item[(vi)] Additionally assuming that \(\rho_*\) has a LDG~\eqref{eq:LDGdiscuss} (e.g.~if \(\rho_*\) is a normalized projection onto a gapped spectral patch of \(H_*\); see Proposition~\ref{prop:GDG}),
            the commutator trace norm LPPL implies LDG for \(\rho\), i.e.~\(\eqref{eq:LPPLtracecomm} \, \implies \,~\eqref{eq:LDGdiscuss}\).
        \item[(vii)] Additionally assuming that \(
                \expectation{[\calL_{H_*} \circ \calI_{H_*,g}\Ab{A} - A,B]}_{\rho_*}
                =
                0
            \) for all \(A,B \in \Alg\) (e.g.~if \(\rho_*\) is a normalized projection onto a gapped spectral patch of \(H_*\); see Proposition~\ref{prop:GDG}), an intertwining automorphism implies LDG, i.e.~\(\eqref{eq:AUTO} \, \implies \,~\eqref{eq:LDGdiscuss}\).
        \item[(viii)] Assume that \(\rho = \ket{\psi} \bra{\psi}\) and \(\psi\) is a product state.
            Then, a decaying gap size implies a weak LDG, i.e.~\(\eqref{eq:GAPdecay} \, \implies \, ~ \eqref{eq:LDGweak_discuss}\).
        \item[(ix)] A local dynamical gap implies a weakened local dynamical gap, i.e.~\(\eqref{eq:LDGdiscuss} \, \implies \, ~ \eqref{eq:LDGweak_discuss}\).
    \end{itemize}

\end{prop}

\subsection{Two exemplary systems with a local dynamical gap}
\label{subsec:examples}
In this section, we will discuss two exemplary systems, which we will show to satisfy the local gap condition~\eqref{eq:local gap} by means of Proposition~\ref{prop:mechforlocalgap}~(vi) and~(vii).
The first example in Section~\ref{subsec:Ising} is concerned with a (local perturbation of a) classical Ising model.
It is contained in a whole class of examples considered in Section~\ref{subsec:frustfree}, which are studied based on~\cite{BRDF2021}.
We nevertheless discuss it separately, as proving it to satisfy the local dynamical gap condition \nameref{ass:localGAP_main} is elementary, in particular not relying on~\cite{BRDF2021}.
The actual proofs for the two examples are given in Section~\ref{subsec:proofexample}.

\subsubsection{Perturbations of the classical Ising model with weak interaction}
\label{subsec:Ising}

As the first example we consider the classical Ising model on \(\Lambda \subset \Z^d\)
\begin{equation}
    \label{eq:classIsing}
    H_*
    =
    \sum_{x \in \Lambda} \sigma_x^3 \, + \, \frac{1}{2}\sum_{x, y \in \Lambda} \lambda(x-y) \, \sigma_x^3 \sigma_y^3
    ,
\end{equation}
where \(\sigma_x^i\) is the \(i^{\mathrm{th}}\) Pauli matrix \(\sigma_i\) acting only on the spin on site \(x \in \Lambda\).
More precisely,
\begin{subequations} \label{eq:Pauli}
    \begin{equation}
        \sigma_x^i
        =
        \unit \otimes \dotsm \otimes \unit \otimes\underbrace{ \sigma^i}_{\text{site} \ x} \otimes \unit \otimes \dotsm \otimes \unit \in \mathcal{B}\paren[\big]{\otimes_{x \in \Lambda} \HSpace_x}
        ,
    \end{equation}
    where, as usual,
    \begin{equation}
        \sigma^1 :=
        \begin{pmatrix}
            0 & 1 \\ 1 & 0
        \end{pmatrix}
        \,, \quad \sigma^2 :=
        \begin{pmatrix}
            0 & \I \\ - \I & 0
        \end{pmatrix}
        \,, \quad \sigma^3 :=
        \begin{pmatrix}
            1 & 0 \\0 &- 1
        \end{pmatrix}
        .
    \end{equation}
\end{subequations}
The Hamiltonian \(H_*\) features a magnetic field of unit strength in \(3\)-direction and (small) symmetric coupling function \(\lambda: \Gamma \to \R\) of finite range, i.e.~\(\norm{\lambda}_1 := \sum_x \abs{\lambda(x)} < 2\) and there exists some \(R > 0\) such that \(\lambda(x) = 0\) for \(\abs{x} > R\).
In particular, for any \(p\in (0,1]\) and \(b>0\), there exists a constant \(C_* > 0\), such that \(\norm{\Phi_{H_*}}_{b,p} \le C_*\), uniformly in \(\Lambda\), where \(\Phi_{H_*}\) is the canonical interaction realizing \(H_*\).

As we show in Section~\ref{subsubsec:Isingproof}, \textbf{every ground state \(\rho\) for every perturbation of the form \(H = H_* + J\) has a local dynamical gap} in the sense of Assumption \nameref{ass:localGAP_main}.
Here, \(J\) is a \emph{strictly} \(\Lambda \setminus \Lambdag\)-localized SLT Hamiltonian, i.e.~there exists a constant \(C_J > 0\) such that for some \(p\in (0,1]\) and \(b>0\), we have \(\norm{\Phi_J}_{b,p} \le C_J\) for some interaction \(\Phi_J\) realizing \(J\), for which also \(\Phi_J(X) = 0\) whenever \(X \cap \Lambdag \neq \emptyset\) (recall Lemma~\ref{lem:SLTlocal}~(i)).

The example~\eqref{eq:classIsing} can immediately be generalized to an arbitrary graph \(\Lambda \in \mathcal{G}(d, \Cvol)\).
Moreover, our above assertions remain valid for any Hamiltonian with gapped on-site terms and sufficiently weak mutually commuting finite range interactions that can be simultaneously diagonalized with the on-site terms (i.e.~one has a \emph{classical} system).

\subsubsection{Perturbations of frustration free product states}
\label{subsec:frustfree}
As the second basic example, we consider an SLT Hamiltonian \(H_*\) of the form
\begin{equation}
    \label{eq:frustfree}
    H_* = \sum_{Z \subset \Lambda} \Phi(Z)
    ,
\end{equation}
for which there exist \(p\in (0,1]\) and \(b>0\) and a constant \(C_* > 0\) such that \(\norm{\Phi}_{b,p} \le C_*\).
Moreover \(\Lambda \in \mathcal{G}(d, \Cvol)\) in~\eqref{eq:frustfree} is a finite graph as described in Section~\ref{sec:setup}.
Apart from the locality of the interaction \(\Phi\), we will impose the following further conditions (cf.~\cite{BRDF2021}).

\begin{assumption}{(A1)}{Frustration-free ground state}
    \label{ass:frustfree}
    All terms in the SLT Hamiltonian~\eqref{eq:frustfree} are non-negative, i.e.~\(\Phi(Z) \ge 0\) for all \(Z \subset \Lambda\).
    There exists a unique (up to a phase) normalized vector \(\psi_* \in \HSpace\) such that \(\psi_* \in \ker \Phi(Z)\) for all \(Z \subset \Lambda\).
    The corresponding ground state (projection) is denoted by \(\rho_* = \ket{\psi_*} \bra{\psi_*}\).
\end{assumption}

We note that the frustration free assumption depends on the explicit way~\eqref{eq:frustfree} the Hamiltonian is written, i.e.~on the interaction \(\Phi\).

\begin{assumption}{(A2)}{Product property and regularity}
    \label{ass:regular}
    The vector \(\ket{\psi_*} \in \HSpace\) factorizes as \(\ket{\psi_*} = \otimes_{z \in \Lambda} \ket{\psi_{*,z}}\) and for every \(\Omega \subset \Lambda\), the unique ground state vector of
    \begin{equation}
        \label{eq:frustfreerestricted}
        H_*\big\vert_{\Omega} = \sum_{Z \subset \Omega} \Phi(Z)
    \end{equation}
    is given by \(\ket{\psi_*\vert_{\Omega}} = \otimes_{z \in \Omega} \ket{\psi_{*,z}}\).
\end{assumption}

The latter condition can be thought of as a strong variant of the common \emph{local topological quantum order (LTQO)}.
Moreover, it is possible to relax both, the product property of \(\ket{\psi_*}\) and the strong LTQO condition, in the following way: Instead of the product property, we could only assume that \(\ket{\psi_*}\) (possibly upon adjoining an auxiliary state after doubling the Hilbert space, cf.~Assumption 4 in~\cite{BRDF2021}) is unitarily conjugate -- with SLT-generator -- to a product state; instead of the strong LTQO property, we could only assume that, upon adjoining a suitable state \(
    \ket{\psi_*\vert_{\Lambda \setminus\Omega}} \in \otimes_{z \in \Lambda \setminus \Omega} \HSpace_x
\), the unique ground state \(\ket{\psi_*\vert_{\Omega}}\) of~\eqref{eq:frustfreerestricted} is unitarily conjugate -- with SLT-generator localized at the boundary of \(\Omega\) -- to \(\ket{\psi_*}\), i.e.~\(
    \ket{\psi_*\vert_{\Omega}} \otimes \ket{\psi_*\vert_{\Lambda \setminus\Omega}}
    =
    \E^{\I F} \ket{\psi_*}
\), where \(F\) is an SLT operator localized at the boundary of \(\Omega\).
However, we refrain from going into this further generalization of our exemplary system~\eqref{eq:frustfree} for simplicity.

The final assumption on~\eqref{eq:frustfree} concerns the spectral gaps
\begin{equation*}
    \gamma
    :=
    \inf \paren[\big]{\spectrum(H_*) \setminus \{0\}}
    \qquad \text{and} \qquad
    \gamma(\Omega)
    :=
    \inf \paren[\big]{\spectrum(H_*\big\vert_{\Omega}) \setminus \{0\}}
\end{equation*}
of \(H_*\) and its restrictions \(H_*\big\vert_{\Omega}\) to some \(\Omega \subset \Lambda\), respectively.

\begin{assumption}{(A3)}{Gap condition}
    \label{ass:gap}
    The SLT Hamiltonian \(H_*\) from~\eqref{eq:frustfree} has a spectral gap, i.e.~\(\gamma > 0\).
    Moreover, the gap of restrictions~\eqref{eq:frustfreerestricted} of \(H_*\) to \emph{balls} \(\Omega = B_r(x)= \Set{y \in \Lambda \given d(x,y) \le r}\) in \(\Lambda\) shrink at most polynomially with the radius, i.e.~there exist \(C_\gamma, d_\gamma > 0\) such that for every \(x \in \Lambda\) it holds that
    \begin{equation}
        \frac{1}{\gamma(B_r(x))} \le C_\gamma r^{d_\gamma}
        .
    \end{equation}
\end{assumption}

An exemplary system satisfying all of the above assumptions (up to a constant energy shift) is given by the Heisenberg XXZ model for small enough nearest neighbor interactions \(\abs{\lambda_1}\), \(\abs{\lambda_3}\) (depending on the dimension \(d\) and the constant \(\Cvol\), cf.~\eqref{eq:dimension} and~\eqref{eq:Gset}).
The corresponding Hamiltonian is given by
\begin{equation*}
    H_*
    =
    \sum_{x\in \Lambda} \sigma_x^3 + \sum_{(x,y) \in E(\Lambda)} \lambda_1 \, {\sigma}^1_x {\sigma}^1_y + \lambda_1 \, {\sigma}^2_x {\sigma}^2_y + \lambda_3 \, {\sigma}^3_x {\sigma}^3_y
    ,
\end{equation*}
where \(E(\Lambda)\) denote the edges of \(\Lambda\) and we recall the notations~\eqref{eq:Pauli}.

Back to the general setting, as we show in Section~\ref{subsec:frustfreeproof} building on~\cite{BRDF2021}, \textbf{every ground state \(\rho\) for every perturbation of the form \(H = H_* + J\) has a local dynamical gap} in a slightly weakened sense of \nameref{ass:localGAP_main}; see Assumption \nameref{ass:localGAPweak} in Section~\ref{sec:proof}, which is, however, sufficient for obtaining our main result, Theorem~\ref{thm:linear-response}.
Here, \(J\) is a \emph{strongly} \(\Lambda \setminus \Lambdag\)-localized SLT Hamiltonian, which means that there exists a constant \(C_J > 0\) such that for some \(p\in (0,1]\) and \(b>0\), we have \(\norm{\Phi_J}_{b,p} \le C_J\) for some interaction \(\Phi_J\) realizing \(J\), for which also \(\Phi_J(Z) = 0\) whenever \(Z \cap \Lambdag \neq \emptyset\) (recall Lemma~\ref{lem:SLTlocal}~(ii)).

\section{Construction of the NEASS\@: Proofs of Propositions~\ref{prop:NEASS}--\ref{prop:expansionNEASS}}
\label{sec:proof}

The fundamental conceptual idea behind the proof of Proposition~\ref{prop:adswitch} is a perturbative scheme, which was called \emph{space-time adiabatic perturbation theory} in~\cite{PST2003,PST2003b}.
Before going into this expansion in Section~\ref{subsec:adpert}, we show that a weakened version of our local dynamical gap condition carries over to SLT operators (see Lemma~\ref{lem:LDGweak} in Section~\ref{subsec:LDGweak}).
The main technical input for carrying out the space-time adiabatic perturbation scheme, is to show that all the operations involved in the expansion preserve localization of SLT operators as required for Lemma~\ref{lem:LDGweak} to apply.
This is the content of several auxiliary technical results in Appendix~\ref{app:technical}.

\subsection{Weakened version of the local dynamical gap condition for SLT operators}
\label{subsec:LDGweak}
Throughout the proof, we will work under the following weakened version of the original local dynamical gap condition (see Assumptions \nameref{ass:localGAP} and \nameref{ass:localGAP_main}).

\begin{assumption}{(LDG\textsubscript{weak})}{Local dynamical gap condition -- weakened version}
    \label{ass:localGAPweak}
    \noindent
    We say that an equilibrium state \(\rho_0\) of an SLT operator \(H_0\), i.e.~with \([H_0, \rho_0] = 0\), is \emph{weakly locally dynamically gapped} of size at least \(g > 0\) in a region \(\Lambdag \subset \Lambda\) with respect to \(C_\gap\), \(b\), \(p\), \(\beta > 0\) and \(\ell\in \N_0\), if and only if the following holds: For all \(X\subset \Lambda\) satisfying \(\diam(X) \le \dist{X, \Lambda\setminus\Lambdag}^\beta\) and observables \(A \in \alg_X\), and \(Y \subset \Lambda\) and observables \(B \in \alg_Y\), it holds that
    \begin{equation}
        \label{eq:local gap weak}
        \abs*{\expectation[\big]{\commutator*{\calL_{H_0} \circ \calI_{H_0,g}\Ab{A} - A , B}}_{\rho_0}}
        \le
        C_\gap \, \norm{A} \, \norm{B}
        \, \diam(Y)^\ell
        \, \chi_{b,p}\paren[\big]{\dist{X, \Lambda \setminus \Lambdag} }
        .
    \end{equation}
\end{assumption}

As mentioned above, we start with the following basic lemma which will heavily be used in our proof.
It says that the local dynamical gap condition naturally carries over to SLT operators.

\begin{lem} \label{lem:LDGweak}
    Let \(b > 0\), \(p \in (0,1)\) and \(H_0\) be a \(\chi_{b, p}\)-SLT operator.
    Assume that the equilibrium state \(\rho_0\) of \(H_0\) satisfies the weakened local gap condition, Assumption \nameref{ass:localGAPweak} above, with gap size at least \(g > 0\) and with respect to \(C_{\mathrm{gap}}\), \(b\), \(p\), \(\beta\) and \(\ell\).

    Then, there exists a constant \(C\), such that for any \((b,p,\Omega)\)-localized SLT operator \(A\) and observable \(B \in \Alg_Y\), we have that
    \begin{equation}
        \label{eq:lem-LDGweak}
        \abs*{\expectation[\big]{\commutator*{\calL_{H_0} \circ \calI_{H_0,g}\Ab{A} - A , B}}_{\rho_0}}
        \le
        C \, \diam(Y)^{\ell + d} \, \norm{B} \, \norm{\Phi_A}_{b,p; \Omega} \, \chi_{b/2,p \min\{\beta, 1\}}\paren[\big]{\dist{\Omega, \Lambda \setminus \Lambdag}}
        .
    \end{equation}
\end{lem}

The proof of Lemma~\ref{lem:LDGweak} is presented in Appendix~\ref{app:proof-LDGweak}.
The principal idea is to write \(A = \sum_{Z \subset \Lambda} \Phi_A(Z)\) and then estimate only the contribution of `small' \(Z\)'s (i.e.~those with \(\diam(Z) \le \dist{Z, \Lambda \setminus \Lambdag}^\beta\)) by~\eqref{eq:local gap weak}.
Large \(Z\)'s (i.e.~those with \(\diam(Z) > \dist{Z, \Lambda \setminus \Lambdag}^\beta\)) are treated using Lieb-Robinson bounds and the smallness of \(\norm{\Phi_A(Z)}\) (by definition~\eqref{eq:interaction-norm-localization}).

\subsection{The adiabatic perturbation scheme}
\label{subsec:adpert}

For the proof of Proposition~\ref{prop:adswitch} we use the same strategy as in \textcite{Teufel2020}.
However, since we only have a local gap, the lower order terms do not vanish exactly, but can be bounded using Assumption~\nameref{ass:localGAPweak}.

The statements in Propositions~\ref{prop:adswitch}, will be deduced from a time-dependent NEASS, which is part of the next Theorem.
In contrast to the previous works, it will not include a time-dependent unperturbed Hamiltonian \(H_0(t)\), because there is no spectral flow available.
That means, we construct a time-dependent NEASS \(\Pi_n^\epsi(t)\) specifically for the switching Hamiltonian given in~\eqref{eq:Hamiltonian-switching}.
For times \(t \geq 0\) it will turn out to be time-independent.

In order to formulate the result, we introduce time-dependent interactions
\begin{equation*}
    \Phi \colon I \times \{X \subset \Lambda\} \to \Alg^N, \quad (t,X) \mapsto \Phi(t,X) = \Phi(t,X)^* \in \Alg_X
\end{equation*}
for \(I \subset \R\).
We will assume that \(t \mapsto \Phi(t,X)\) is smooth for every \(X \subset \Lambda\) and we denote the term-wise time derivatives by \(\Phi^{(k)}\), i.e.~\(\Phi^{(k)}(t,X) = \frac{\D^k}{\D t^k} \Phi(t,X)\) for every \(X \subset \Lambda\).
Moreover, we identify \(\Phi(t,X) = \paren[\big]{\Phi(t)}(X)\), such that for every fixed \(t \in I\), \(\Phi(t)\) can be viewed as a time-independent interaction.
The notion of SLT operators naturally translates to the time-dependent setting.

\begin{theorem}[Time-dependent NEASS]
    \label{thm:time-dependent-NEASS}
    Fix \(n\in \N\) and let \(d\in \N\), \(\Cvol>0\), \(b>0\), \(p\in \intervaloc{0,1}\), \(\Cint>0\) and \(g > 0\), \(C_\gap>0\), \(\beta>0\), \(\ell\in \N_0\).
    Take any \(q \in (0, p \min\{1, \beta\})\).
    Then there exist a constant \(C_{n}> 0\) (in particular depending on \(n\))
    such that for all lattices \(\Lambda\in \mathcal{G}(d,\Cvol)\) (recall~\eqref{eq:Gset}), subsets \(\Lambdap\subset \Lambda\), intervals \(I\subset \R\) and SLT-operators \(H_0\) and \(V(t)\), with corresponding interactions that satisfy \(\norm{\Phi_{H_0}}_{b,p}<\Cint\) and \(\sup_{t\in I} \norm{\Phi^{(k)}_{V}(t)}_{b,p;\Lambdap}<\Cint\) for all \(k \le n\), respectively, the following holds:

    Assume that the equilibrium state \(\rho_0\) of \(H_0\) is locally dynamically gapped in \(\Lambdag\) of size at least \(g > 0\) and with respect to \(C_\gap\), \(b\), \(p\) and \(\ell\) according to Assumption \nameref{ass:localGAPweak}.
    And let
    \begin{equation}
        \label{eq:proof-Hamiltonian-switching}
        H_\epsi(t) := H_0 + \epsi \, V(t)
    \end{equation}
    be the perturbed Hamiltonian.

    Then, there exists a sequence \((A_\mu)_{\mu \in \N}\) of time-dependent SLT operators, which are \((1,p',\Lambda^{\mathrm{pert}})\)-localized for any \(p' < p\), such that the state
    \begin{equation}
        \label{eq:NEASS-time-dependent}
        \Pi_n^{\epsi,\eta}(t)
        :=
        \evol{\I S_n^{\epsi,\eta}(t)}{\rho_0}
        \quadtext{with}
        S_n^{\epsi,\eta}(t)
        :=
        \sum_{\mu= 1}^n \epsi^\mu A_\mu(t)
        ,
    \end{equation}
    is almost-stationary for the dynamics generated by \(H_\epsi(t)\) in the following sense: Let \(t_0\in \R\) and let \(\rho^{\epsi, \eta, f}(t)\) be the solution of the time-dependent adiabatic Schrödinger equation
    \begin{equation}
        \label{eq:time-dependent-NEASS-Schrödinger-equation-switching}
        \I\eta \frac{\D }{\D t} \rho^{\epsi, \eta, f}(t)
        =
        \commutator[\big]{H_\epsi(t) , \rho^{\epsi, \eta, f}(t)}
        \quadtext{with}
        \rho^{\epsi, \eta, f}(t_0) = \Pi_n^{\epsi,\eta}(t_0)
    \end{equation}
    with adiabatic parameter \(\eta \in \intervaloc{0,1}\).

    Under these conditions, for all \(B\in \alg_Y\) with \(Y\subset \Lambda\) and \(t \in \R\), it holds that
    \begin{equation}
        \label{eq:time-dependent-NEASS-bound}
        \abs[\big]{
            \expectation{B}_{\rho\ee_{t_0}(t)}
            - \expectation{B}_{\Pi_n\ee(t)}
        }
        \le
        \begin{multlined}[t]
            C_{n}
            \, \norm{B}
            \, \diam(Y)^{3d+\ell}
            \, \paren[\bigg]{1+\frac{\eta^{n+1}}{\epsi^{n+1}}} \, \epsi^{n+1}
            \\\times\, \frac{\abs{t-t_0}}{\eta}
            \, \paren[\bigg]{1 + \frac{v\,\abs{t-t_0}}{\eta}}^{(\ell+2d)/p}
            \\\times \paren[\Big]{1 +
                \e^{-\dist{\Lambdap,\Lambda\setminus\Lambdag}^q - (n+1) \log(\epsi)}
            }
            .
        \end{multlined}
    \end{equation}
    Moreover, the operators \(A_\mu(t)\) at time \(t\) depend only on \(V\) and its first \(\mu\) derivatives at time \(t\).
\end{theorem}

Before we prove Theorem~\ref{thm:time-dependent-NEASS}, let us deduce the results from Section~\ref{sec:result}, which will follow by taking $\beta = 1$ in Theorem~\ref{thm:time-dependent-NEASS}.

\begin{proof}[Proof of Proposition~\ref{prop:NEASS}]
    We choose \(V(t) = V\), which implies that also all \(A_\mu\) are time-independent (i.e.~the time-dependent interaction is constant).
    Moreover, since there is no adiabatic timescale in~\eqref{eq:prop-NEASS-Schrödinger-equation}, we choose \(\eta=1\).
    To obtain the correct scaling, we inspect the proof of Theorem~\ref{thm:time-dependent-NEASS}:
    In equation~\eqref{eq:Q_n-before-expansion}, we expand in powers of \(\epsi\) and \(R_j^{\epsi,\eta}\) are polynomials in \(\eta/\epsi\).
    Here, since \(\dot S_n^{\epsi,\eta} = 0\), these polynomials are just constants and there is no \(\eta\) in~\eqref{eq:Q_n-before-expansion} or any of the later expressions.
    Hence, also the norm estimates \(\norm{\Phi_{\tilde R_j^{\epsi,\eta}}(s)} < C \, (1+\eta^j/\epsi^j)\) used in the end of the proof simplify to \(\norm{\Phi_{\tilde R_j^{\epsi}}} < C\) uniformly in \(\epsi\).
    Hence, \((1+\eta^{n+1}/\epsi^{n+1})\) is replaced by~\(1\) in~\eqref{eq:time-dependent-NEASS-bound}.
    All other \(\eta\) in~\eqref{eq:time-dependent-NEASS-bound} come from the adiabatic timescale and are thus replaced by \(1\).
\end{proof}

\begin{proof}[Proof of Proposition~\ref{prop:adswitch}]
    We choose \(V(t) = f(t) \, V\) and since \(V(t)\) constant for \(t \leq 0\), we obtain \(\Pi_n^{\epsi,\eta}(t) = \Pi_n^\epsi\) for all \(t \geq 0\).
    To compare with the solution of~\eqref{eq:Schrödinger-equation-switching}, we choose \(t_0=-1\) such that \(\Pi_n^{\epsi,\eta}(t_0) = \rho_0\).
    Then \(\abs{t-t_0}=1+t\) and~\eqref{eq:time-dependent-NEASS-bound} gives the statement.
\end{proof}

\begin{proof}[Proof of Proposition~\ref{prop:expansionNEASS}]
    To prove the asymptotic expansion, we first expand~\eqref{eq:NEASS} and obtain
    \begin{equation*}
        \trace[\big]{\Pi_n^\epsi \, B}
        =
        \trace[\big]{\rho_0 \, \e^{\epsi \calL_{S_n^\epsi}}\Ab{B}}
        =
        \sum_{k=0}^m \frac{1}{k!} \, \expectation[\big]{\calL^k_{S_n^\epsi}\Ab{B}}_{\rho_0}
        + \frac{1}{(m+1)!}
        \, \trace[\big]{\rho_0 \, \e^{\tilde\epsi \calL_{S_n^\epsi}}\circ\calL^{m+1}_{S_n^\epsi}\Ab{B}}
    \end{equation*}
    for some \(\tilde\epsi\in \intervalcc{0,\epsi}\).
    Since the \(A_\mu\) and thus also the \(S_n^\epsi\) are \(\Lambdap\)-localized, we use Lemma~\ref{lem:commutator with local observable} to bound the remainder by
    \begin{equation*}
        \frac{1}{(m+1)!}
        \, \norm[\big]{\calL^{m+1}_{S_n^\epsi}\Ab{B}}
        \leq
        C
        \, \norm{B}
        \, \abs{Y}^{m+1}
        \, \chi_{b',p}\pdist{Y,\Lambdap}
        \, \norm{S_n^\epsi}_{b',p,\Lambdap}^{m+1}
        ,
    \end{equation*}
    where \(\norm{S_n^\epsi}_{b',p;\Lambdap} \leq \epsi \, \tilde{C}\).
    It is clear from the proof of Theorem~\ref{thm:time-dependent-NEASS}, that \(\tilde{C}\) depends only on \(n, d,\Cvol, b, p, \Cint\), and \(g\).
    We now expand \(S_n^\epsi\) in the first term and group the terms according to the powers in \(\epsi\).
    The zero order term clearly is \(\expectation{B}_{\rho_0}\).
    In first order, we obtain \(
        \epsi \, \expectation{\calL_{A_1}}_{\rho_0}
        =
        - \epsi \,\I \, \expectation*{\commutator[\big]{\calI_{H_0,g}\Ab{V} , B}}_{\rho_0}
    \) as can be read of from~\eqref{eq:proof-remainder-order-1}.
    All \(\mathcal{K}_j\) for \(j \leq m\) are constructed in this way.
    In the end, some higher order terms are left.
    They all come with multi-commutators of \(\Lambdap\)-localized \(A_\mu\) with \(B\) and can be bounded as the remainder above using Lemma~\ref{lem:commutator with local observable}.
\end{proof}

We now prove the time-dependent NEASS from Theorem~\ref{thm:time-dependent-NEASS}.

\begin{proof}[Proof of Theorem~\ref{thm:time-dependent-NEASS}]
    For the proof we first assume the form of~\eqref{eq:NEASS-time-dependent} and then iteratively choose \(A_\mu(t)\) so that the statement holds.
    Therefore, let \(U_{t,t_0}^{\epsi,\eta}\) be the solution of
    \begin{equation}
        \label{eq:Schrödinger-equation-full-evolution-propagater}
        \I \, \eta \, \frac{\D}{\D t} \, U_{t,t_0}^{\epsi,\eta}
        =
        H_\epsi(t) \, U_{t,t_0}^{\epsi,\eta}
        \qquadtext{with}
        U_{t_0,t_0}^{\epsi,\eta}
        =
        \unit
        \quadtext{for all}
        t,t_0 \in I
        ,
    \end{equation}
    with \(H_\epsi\) given in~\eqref{eq:Hamiltonian-switching}.
    Then, \(
        \rho\ee_{t_0}(t)
        :=
        U\ee_{t,t_0} \, \Pi_n\ee(t_0) \, U\ee_{t_0,t}
    \) is the solution of~\eqref{eq:time-dependent-NEASS-Schrödinger-equation-switching}.
    To obtain~\eqref{eq:time-dependent-NEASS-bound} we use the fundamental theorem of calculus, and get
    \begin{equation}
        \label{eq:proof-expansion-fundamental-theorem-of-calculus}
        \expectation{B}_{\rho\ee_{t_0}(t)}
        - \expectation{B}_{\Pi_n\ee(t)}
        =
        - \int_{t_0}^t \D s \, \frac{\D}{\D s} \, \expectation[\Big]{\evoli{\I S_n\ee(s)}{U\ee_{s,t}\,B\,U\ee_{t,s}}}_{\rho_0}
        .
    \end{equation}
    By product rule and Duhamel’s formula, the derivative evaluates as
    \begin{equation*}
        \frac{\D}{\D s} \, \evoli{\I S_n\ee(s)}{U\ee_{s,t}\,B\,U\ee_{t,s}}
        =
        -\frac{\I}{\eta} \, \commutator[\Big]{
            Q\ee_n(s) ,
            \evoli{\I S_n\ee(s)}{U\ee_{s,t}\,B\,U\ee_{t,s}}
        }
        ,
    \end{equation*}
    where
    \begin{align}
        \nonumber
        Q\ee_n(s)
        &=
        \eta \, \int_0^1 \D \lambda \, \evoli{\I \lambda S\ee_n(s)}{\dot S\ee_n(s)}
        + \evoli{\I S_n\ee(s)}{\paren[\big]{H_0 + \epsi \, V(s)}}
        \\&=
        \label{eq:Q_n-before-expansion}
        \eta \, \int_0^1 \D \lambda \, \e^{\lambda \calL_{S\ee_n(s)}} \Ab[\big]{\dot S\ee_n(s)}
        + \e^{\calL_{S_n\ee(s)}}\Ab[\big]{H_0 + \epsi \, V(s)}
        \\&=
        \nonumber
        H_0
        + \sum_{j=1}^n \epsi^j \, R_j\ee(s)
        + \epsi^{n+1} \, R_{n+1}\ee(s)
        .
    \end{align}
    In the last line we expanded in powers of \(\epsi\) and \(\eta\) such that \(R_j\ee(s)\) are polynomials in \(\eta/\epsi\) of order \(j\) with \(\epsi\)- and \(\eta\)-independent SLT operators as coefficients.
    In this way, the joint power of \(\epsi\) and \(\eta\) in front of the SLT operators collected in \(R_j\ee\) is \(j\).
    By Taylor formula with mean-value form of the remainder, there exist \(\theta \in \intervalcc{0,1}\) such that
    \begin{equation}
        \label{eq:proof-expansion-expansion-of-derivative-of-full-evolution}
        \e^{\calL_{S\ee_n(s)}} \Ab[\big]{H_0 + \epsi \, V(s)}
        =
        \begin{aligned}[t]
            &\sum_{k=0}^n \frac{1}{k!} \, \calL_{S\ee_n(s)}^k \Ab[\big]{H_0 + \epsi \, V(s)}
            \\&+
            \frac{1}{(n+1)!} \, \e^{\theta \calL_{S\ee_n(s)}} \circ \calL_{S\ee_n(s)}^{n+1} \Ab[\big]{H_0 + \epsi \, V(s)}
            .
        \end{aligned}
    \end{equation}
    Similarly, for the first term in~\eqref{eq:Q_n-before-expansion} we expand the integrand using the integral form of the remainder and obtain
    \begin{align}
        \nonumber
        \Alignindent
        \eta \, \int_0^1 \D \lambda \, \e^{\lambda \calL_{S\ee_n(s)}} \Ab[\big]{\dot S\ee_n(s)}
        \\&=
        \nonumber
        \begin{aligned}[t]
            &
            \eta \, \sum_{k=0}^{n-2} \frac{\int_0^1 \D \lambda \, \lambda^k }{k!} \, \calL_{S\ee_n(s)}^k\Ab[\big]{\dot S_n\ee(s)}
            \\&+
            \eta
            \, \int_0^1 \D \lambda
            \, \int_0^\lambda \D \mu
            \, \frac{(\lambda-\mu)^{n-2}}{(n-2)!}
            \, \e^{\mu \calL_{S\ee_n(s)}} \circ \calL_{S\ee_n(s)}^{n-1}
            \Ab[\big]{\dot S\ee_n(s)}
            ,
        \end{aligned}
        \\&=
        \label{eq:proof-expansion-expansion-of-derivative-of-NEASS-automorphism}
        \begin{aligned}[t]
            &
            \eta \, \sum_{k=0}^{n-2} \frac{1}{(k+1)!} \, \calL_{S\ee_n(s)}^k\Ab[\big]{\dot S_n\ee(s)}
            \\&+
            \eta
            \, \int_0^1 \D \mu
            \, \frac{(1-\mu)^{n-1}}{(n-1)!}
            \, \e^{\mu \calL_{S\ee_n(s)}} \circ \calL_{S\ee_n(s)}^{n-1}
            \Ab[\big]{\dot S\ee_n(s)}
            ,
        \end{aligned}
    \end{align}
    From this expansion and~\eqref{eq:NEASS-time-dependent}, we can read of
    \begin{align}
        \label{eq:proof-remainder-order-1}
        R_1\ee(s)
        &=
        - \calL_{H_0}\Ab{A_1\ee(s)}
        + V(s)
        ,
        \\
        R_2\ee(s)
        &=
        - \calL_{H_0}\Ab{A_2\ee(s)}
        + \calL_{A_1\ee(s)}\Ab{V(s)}
        + \tfrac{\eta}{\epsi} \dot A_1\ee(s)
        ,
    \shortintertext{and more generally}
        R_j\ee(s)
        &=
        - \calL_{H_0}\Ab{A_j\ee(s)} + \Rt_j\ee(s)
        ,
    \end{align}
    where the \(\Rt_j\ee(s)\) are sums of iterated commutators of the operators \(A_i\ee(s)\) and \(\dot A_i\ee(s)\) for \(i<j\leq n\) and \(V(s)\).
    We can now iteratively choose
    \begin{equation*}
        A_j\ee(s) = \calI_{H_0,g}\Ab[\big]{\Rt_j\ee(s)}
        .
    \end{equation*}
    Clearly, for all \(p'<p\), it holds that \(\Rt_1\ee(s) = V(s)\) is a \((b,p',\Lambdap)\)-localized SLT operator.
    Hence, by Lemma~\ref{lem:inverse-liouvillian-preserves-SLT}, \(A_1\ee(s)\) is \((b',p',\Lambdap)\)-localized for any \(b'<b\).
    This step only works for \(p'<p\), because Lemma~\ref{lem:inverse-liouvillian-preserves-SLT} requires a slightly better localization of the Hamiltonian~\(H_0\) compared to the argument of the inverse Liouvillian.
    Finally, by Lemma~\ref{lem:multi-commutators-of-SLT-are-SLT} also \(\calL_{H_0}\Ab{A_1\ee(s)}\) and thus \(R_1\ee(s)\) are \((b',p',\Lambdap)\)-localized SLT operator for any slightly smaller~\(b'\).
    The same arguments hold for the higher orders \(R_j\ee(s)\) with \(j \leq n\) as well, at each step lowering \(b'\).
    We point out that clearly the smallest~\(b'\) can be chosen independently of~\(n\) by choosing all intermediate~\(b'\) in an \(n\)-dependent equidistant way.

    With this expansion it is also immediate, that the \(A_\mu(s)\) only depend on \(V\) and its derivatives at time \(s\).

    Putting everything back together and denoting \(\tau_{t,s} \Ab{B} := U\ee_{s,t}\,B\,U\ee_{t,s}\), we find
    \begin{align}
        \label{eq:proof-time-dependent-NEASS-bound}
        \Alignindent
        \abs[\Big]{
            \expectation{B}_{\rho\ee_{t_0}(t)}
            - \expectation{B}_{\Pi_n\ee(t)}
        }
        \\&\leq
        \nonumber
        \frac{\abs{t-t_0}}{\eta}
        \sup_{s\in \intervalcc{t_0,t}}
        \paren*{
            \begin{aligned}
                &
                \, \sum_{j=1}^{n} \epsi^j \, \abs[\bigg]{
                    \expectation[\Big]{
                        \commutator[\Big]{
                            \calL_{H_0} \circ \calI_{H_0,g}\Ab[\big]{\Rt_j\ee(s)}
                            - \Rt_j\ee(s)
                            ,
                            \e^{\calL_{S_n\ee(s)}} \circ \tau_{t,s}\Ab{B}
                        }
                    }_{\rho_0}
                }
                \\&+
                \epsi^{n+1}
                \, \abs[\bigg]{
                    \expectation[\Big]{
                        \commutator[\Big]{
                            R_{n+1}\ee(s)
                            ,
                            \e^{\calL_{S_n\ee(s)}} \circ \tau_{t,s}\Ab{B}
                        }
                    }_{\rho_0}
                }
                .
            \end{aligned}
        }
    \end{align}
    The remainder \(R_{n+1}\ee(s)\) collects all the remaining terms which are (a) the higher order terms from the first lines and (b) the remainder terms from the second lines of~\eqref{eq:proof-expansion-expansion-of-derivative-of-full-evolution} and~\eqref{eq:proof-expansion-expansion-of-derivative-of-NEASS-automorphism}.
    The former are local by the previous arguments.
    The latter additionally include an evolution by the local operator \(S_n\ee(s)\) and are local by Lemma~\ref{lem:conjugation-with-unitaries-evolved-SLT-is-SLT}.

    To apply the local gap assumption in the form given in Lemma~\ref{lem:LDGweak} to the lower order terms, we need to decompose the second entry of the commutator into strictly local operators.
    Therefore, we use the same decomposition as in the proof of Lemma~\ref{lem:commutator-SLT-operator-evolved-observable}
    \begin{equation*}
        \tau_{t,s}\Ab{B}=\sum_{k=0}^\infty \Delta_k
        ,
    \end{equation*}
    where
    \begin{equation*}
        \begin{aligned}
            \Delta_0 &\in \alg_{Y_{\delta}} \subset \alg_{Y_{\delta^{1/p}}},
            & \norm{\Delta_0} &\leq \norm{B},
            \\
            \text{and for \(k\in \N\)}\qquad
            \Delta_k &\in \alg_{Y_{(\delta+k)^{1/p}}},
            & \norm{\Delta_k} &\leq C_\LR \, \E \, \norm{B} \, \abs{Y} \, \E^{-b k}
            .
        \end{aligned}
    \end{equation*}
    For better readability we abbreviate \(\delta=1+\vts/\eta\) here.
    The extra \(1/\eta\) is due to the scaling in~\eqref{eq:Schrödinger-equation-full-evolution-propagater}.
    For the outer automorphism, we can use the decomposition from the proof of Lemma~\ref{lem:conjugation-with-unitaries-evolved-SLT-is-SLT}, just replace \(O=\Delta_k\) there, to obtain, for each \(k\), the decomposition
    \begin{equation*}
        \e^{\calL_{S_n\ee(s)}}\Ab{\Delta_k} = \sum_{l=0}^\infty \Delta_{k,l}
    \end{equation*}
    with
    \begin{equation*}
        \begin{aligned}
            \Delta_{k,0} &\in \alg_{Y_{(\delta+k)^{1/p}}},
            & \norm{\Delta_{k,0}}
            &\leq
            \norm{\Delta_k},
            \quad\text{and}
            \\
            \Delta_{k,l} &\in \alg_{Y_{l+(\delta+k)^{1/p}}},
            & \norm{\Delta_{k,l}}
            &\leq
            C_\LR \, \norm{\Delta_k} \, \abs[\big]{Y_{(\delta+k)^{1/p}}} \, \paren[\big]{\e^{c \, \norm{\Phi_S}_{d,q}} - 1} \, \chi_{d',j}(l)
        \end{aligned}
    \end{equation*}
    for \(l \in \N\).
    Thus, in total we have
    \begin{equation*}
        \e^{\calL_{S_n\ee(s)}} \circ \tau_{t,s}\Ab{B}
        =
        \sum_{k,l=0}^\infty \Delta_{k,l}
        ,
    \end{equation*}
    where the sum is actually finite on finite lattices as discussed in Lemmata~\ref{lem:commutator-SLT-operator-evolved-observable} and~\ref{lem:conjugation-with-unitaries-evolved-SLT-is-SLT}.
    Then we use the triangle inequality and Lemma~\ref{lem:LDGweak} to apply~\eqref{eq:lem-LDGweak} and bound
    \begin{align*}
        \Alignindent
        \abs[\Big]{
            \expectation[\Big]{
                \commutator[\Big]{
                    \calL_{H_0} \circ \calI_{H_0,g}\Ab[\big]{\Rt_j\ee(s)}
                    - \Rt_j\ee(s)
                    ,
                    \e^{\calL_{S_n\ee(s)}} \circ \tau_{t,s}\Ab{B}
                }
            }_{\rho_0}
        }
        \\&\leq
        C
        \, \norm{\Phi_{\Rt_j\ee}(s)}_{b',p'; \Lambdap}
        \, \sum_{k,l=0}^\infty
        \diam(Y_{l+(\delta+k)^{1/p}})^{\ell + d}
        \, \norm{\Delta_{k,l}}
        \, \chi_{b'/2,q'}\pdist{\Lambdap, \Lambda \setminus \Lambdag}
        \\&\leq
        \begin{aligned}[t]
            &
            C
            \, \norm{\Phi_{\Rt_j\ee}(s)}_{b',p'; \Lambdap}
            \, \sum_{k=0}^\infty
            \diam(Y_{(\delta+k)^{1/p}})^{\ell + 2d}
            \, \norm{\Delta_{k}}
            \\&\times \paren[\bigg]{
                1
                + \paren[\big]{\e^{c \, \norm{\Phi_S}_{d,q}} - 1}
                \sum_{l=1}^\infty l^{\ell+d} \, \chi_{d',j}(l)
            }
            \, \chi_{b'/2,q'}\pdist{\Lambdap, \Lambda \setminus \Lambdag}
        \end{aligned}
        \\&\leq
        \begin{aligned}[t]
            &
            C
            \, \norm{\Phi_{\Rt_j\ee}(s)}_{b',p'; \Lambdap}
            \, \e^{c \, \norm{\Phi_S}_{d,q}}
            \, \diam(Y)^{\ell + 3d}
            \, \norm{B}
            \, \delta^{(\ell+2d)/p}
            \\&\times \sum_{k=0}^\infty
            \, k^{(\ell+2d)/p}
            \, \e^{-bk}
            \, \chi_{b'/2,q'}\pdist{\Lambdap, \Lambda \setminus \Lambdag}
        \end{aligned}
        \\&\leq
        C
        \, \norm{\Phi_{\Rt_j\ee}(s)}_{b',p'; \Lambdap}
        \, \e^{c \, \norm{\Phi_S}_{d,q}}
        \, \diam(Y)^{\ell + 3d}
        \, \norm{B}
        \, \delta^{(\ell+2d)/p}
        \, \chi_{b'/2,q'}\pdist{\Lambdap, \Lambda \setminus \Lambdag}
        ,
    \end{align*}
    where we abbreviated \(q' := p' \, \min\List{\beta,1}\) (according to Lemma~\ref{lem:LDGweak}) and~\(b'\) is the smallest of the~\(b'\) such that all \(\Phi_{\Rt_j\ee}(s)\) are \((b',p',\Lambdap)\)-localized.
    Finally, for any \(q < q'\) we bound
    \begin{equation}
        \label{eq:backto1}
        \chi_{b'/2,q'}\pdist{\Lambdap, \Lambda \setminus \Lambdag}
        \le
        C \, \chi_{1,q}\pdist{\Lambdap, \Lambda \setminus \Lambdag}
    \end{equation}
    for some constant \(C > 0\) depending only on \(b'\), \(q'\), and~\(q\).
    Since \(p'<p\) was arbitrary, this~\(q\) can be chosen arbitrarily in the interval \((0, p \min\{\beta, 1\})\).
    \normalcolor

    With the same local decomposition argument as above and using Lemma~\ref{lem:commutator with local observable} we can also bound the remainder
    \begin{align*}
        \Alignindent
        \abs[\bigg]{
            \expectation[\Big]{
                \commutator[\Big]{
                    R_{n+1}\ee(s)
                    ,
                    \e^{\calL_{S_n\ee(s)}} \circ \tau_{t,s}\Ab{B}
                }
            }_{\rho_0}
        }
        \\&\leq
        C
        \, \norm{\Phi_{\Rt_{n+1}\ee}(s)}_{b',p'; \Lambdap}
        \, \e^{c \, \norm{\Phi_S}_{d,q}}
        \, \diam(Y)^{3d}
        \, \norm{B}
        \, \delta^{2d/p}
        .
    \end{align*}
    We recall, that like the \(R_j\ee(s)\), also the \(\tilde{R}_j\ee(s)\) are polynomials of order \(j\) in \(\eta/\epsi\) with \(\Lambdap\)-localized SLT operators as coefficients.
    Hence, the interaction norms in the above equations can be bounded by a constant depending on the interaction norms of \(H\) and \(V(t)\) times \(1+(\eta/\epsi)^j\).
    With this observation, we can insert these bounds into~\eqref{eq:proof-time-dependent-NEASS-bound} and conclude~\eqref{eq:time-dependent-NEASS-bound}.

    To reduce the number of constants in the formulation of the statement, we actually do the proof for \((p+p')/2\) instead of \(p'\).
    In the end we then estimate \(\chi_{b',(p+p')/2} \leq C \, \chi_{1,p'}\) for the locality of the operators \(A_\mu(t)\), similarly to~\eqref{eq:backto1}.
\end{proof}

\section{Local gap conditions: Proofs for Proposition~\ref{prop:GDG} and Section~\ref{sec:localgap}}
\label{app:LPPLlocalgap}

This section collects several proofs concerning the local dynamical gap condition \nameref{ass:localGAP} (or \nameref{ass:localGAP_main} for the formal version), which were skipped in earlier sections.
More precisely, we will prove the dynamical characterization of a spectral gap in Proposition~\ref{prop:GDG} and the relations among the various local gap conditions formulated in Proposition~\ref{prop:mechforlocalgap} in Section~\ref{subsec:compare}.
Finally, we show that the examples from Section~\ref{subsec:examples} satisfy \nameref{ass:localGAP_main} by means of Proposition~\ref{prop:mechforlocalgap}.

\subsection{Dynamical characterization of a spectral gap: Proof of Proposition~\ref{prop:GDG}}
\label{subsec:dyncharGDG}
Deriving the lhs.~from the rhs.~is standard material; see, e.g.,~\cite[Lemma~6.8 and Proposition~6.9]{NSY2019}.
Therefore, dropping the subscripts \(H\) and \(g\) for ease of notation, suppose that for all \(A,B \in \mathcal{B}(\HSpace)\) we have (recall~\eqref{eq:Jdef})
\begin{equation*}
    0
    =
    \expectation[\big]{\commutator[\big]{\calL \circ \calI\Ab{A} - A, B}}_P
    =
    \expectation[\big]{\commutator[\big]{\calJ\Ab{A},B}}_{P}
    =
    - \trace[\big]{\commutator[\big]{\calJ\Ab{A}, P} \, B}
    .
\end{equation*}
Since \(B\) is arbitrary, this means \([\calJ\Ab{A}, P] = 0\).
Moreover, inserting the spectral decomposition \(H = \sum_{n} E_n \, P_n\), this can be written as
\begin{equation}
    \label{eq:GDGcontradict}
    0
    =
    \commutator[\big]{\calJ\Ab{A}, P}
    =
    \sqrt{2 \pi} \, \paren*{
        \sumstack[r]{E_n \in \sigma_2 \\ E_m \in \sigma_1}
        \widehat{w}(E_m - E_n) \, P_n \, A \, P_m
        -
        \sumstack[lr]{E_n \in \sigma_1 \\ E_m \in \sigma_2}
        \widehat{w}(E_m - E_n) \, P_n \, A \, P_m
    }
    .
\end{equation}

For \emph{contradiction}, we now assume that \(\dist{\sigma_1, \sigma_2} < g\).
Then, since \(\widehat{w}\vert_{(-g,g)} > 0\), one can easily construct an observable \(A\), which violates~\eqref{eq:GDGcontradict}, e.g.~\(A = P_{n_*} P_{m_*}\) with \(E_{n_*} \in \sigma\), \(E_{m_*} \in \sigma_2\) satisfying \(\abs{E_{n_*} - E_{m_*}} = \dist{\sigma_1, \sigma_2} < g\).

\subsection{Relations among local gap conditions: Proof of Proposition~\ref{prop:mechforlocalgap}}
\label{subsec:proofmechforlocalgap}
We prove the seven implications gathered in Proposition~\ref{prop:mechforlocalgap} one by one.
Unless differently stated, we will use the constants \(C, \ell, b\) and \(p\) from the formulation of Proposition~\ref{prop:mechforlocalgap} generically, i.e.~their precise value might change from line to line.
Some technical details are kept brief in this section, more detailed proofs using similar arguments are given in Appendix~\ref{app:technical}.

\medskip

\noindent \underline{Proofs of (i)--(iii):} All of these are obvious, by application of the estimates
\begin{align*}
    \norm{[\rho-\rho_*, A]}_{\mathup{tr}}
    &\le
    \norm{(\rho-\rho_*) \, A}_{\mathup{tr}} + \norm{A \, (\rho-\rho_*) }_{\mathup{tr}}
    ,
    \\
    \abs*{\trace*{(\rho-\rho_*)A}}
    &\le
    \norm{(\rho-\rho_*) A}_{\mathup{tr}}\,,
\shortintertext{and}
    \abs*{\trace*{\rho_* \paren[\big]{\tau - \unit}\Ab{A}}}
    &\le
    \norm[\big]{\paren[\big]{\tau - \unit}\Ab{A}}
    ,
\end{align*}
for (i), (ii), and (iii), respectively.
\\[2mm]
\noindent \underline{Proof of (iv):} By assumption, it holds that
\begin{equation}
    \label{eq:LPPLtoCOERdef}
    \I \, \expectation{A^* \calL_{H_*}\Ab{A}}_{\rho_*}
    \ge
    g \, \left[\expectation{A^* A}_{\rho_*} - \abs[\big]{\expectation{A}_{\rho_*}}^2\right]
\end{equation}
for all observables \(A \in \Alg\), where \(\calL_{H_*} \Ab{\interpunct} := -\I \, [H_*, \interpunct]\) denotes the Liouvillian of \(H_*\).
The idea is now to replace \(H_* \to H\) and \(\rho_* \to \rho\) in~\eqref{eq:LPPLtoCOERdef} and estimate the resulting error in such a way that we arrive at~\eqref{eq:COERdefect}.

First, by application of~\eqref{eq:LPPL}, we replace \([\expectation{A^* A}_{\rho_*} - \abs{\expectation{A}_{\rho_*}}^2]\) by \([\expectation{A^* A}_{\rho} - \abs{\expectation{A}_{\rho}}^2]\) on the rhs.~of~\eqref{eq:LPPLtoCOERdef} at the cost of an error bounded by \(
    C \, \norm{A}^2 \diam(X)^\ell \, \chi_{b,p}\paren[\big]{\dist{X, \Lambda \setminus \Lambdag}}
\).

For the lhs.~of~\eqref{eq:LPPLtoCOERdef}, we estimate
\begin{equation}
    \label{eq:LPPLtoCOERdef2}
    \begin{aligned}
        \Alignindent
        \abs*{\expectation{A^* \calL_{H_*}\Ab{A}}_{\rho_*} - \expectation{A^* \calL_{H}\Ab{A}}_{\rho}}
        \\&\le
        \abs*{\expectation[\big]{A^* \paren[\big]{\calL_{H_*}\Ab{A} - \calL_H\Ab{A}}}_{\rho_*}} + \abs*{\trace[\big]{(\rho_* - \rho) A^* \calL_{H}\Ab{A}}}
    \end{aligned}
\end{equation}
by means of the triangle inequality.
The first term on the rhs.~of~\eqref{eq:LPPLtoCOERdef2} can now be bounded as (recall that \(H = H_* + J\), \(J\) is \(\Lambda \setminus \Lambdag\)-localized, and \(A \in \Alg_X\))
\begin{equation}
    \label{eq:LPPLCOERterm1}
    \abs*{\expectation[\big]{A^* \paren[\big]{\calL_{H_*}\Ab{A} - \calL_H\Ab{A}}}_{\rho_*}}
    \le
    2 \norm{\Phi_J}_{b,p;\Lambda \setminus \Lambdag} \, \norm{A}^2 \diam(X)^d \chi_{b,p}\pdist{X, \Lambda \setminus \Lambdag}
\end{equation}
by application of~\eqref{eq:lem-commutator-SLT-with-observable} in Lemma~\ref{lem:commutator with local observable}.

For the second term on the rhs.~of~\eqref{eq:LPPLtoCOERdef2}, we write \(
    \calL_H\Ab{A}
    =
    \mathbb{E}_{X_n} \calL_H\Ab{A} + (\unit - \mathbb{E}_{X_n} ) \calL_H\Ab{A}
\) for some \(n\) to be chosen below, where \(X_n := \Set{x \in \Lambda \given \dist{x,X} \le n}\) denotes the \(n\)-fattening of the set \(X \subset \Lambda\).
We now estimate the two terms separately.
For the first term, we employ~\eqref{eq:LPPL} to bound
\begin{equation}
    \label{eq:LPPLCOERterm21}
    \abs*{\trace[\big]{(\rho_* - \rho) A^*\mathbb{E}_{X_n} \calL_H\Ab{A}}}
    \le
    C \norm{A}^2 \paren[\big]{\diam(X) + n}^\ell \chi_{b,p}\paren[\big]{\dist{X_n, \Lambda\setminus \Lambdag}}
\end{equation}
where we used that \(\mathbb{E}_{X_n} \calL_H\Ab{A} \in \Alg_{X_n}\) (by definition) and \(\norm{\mathbb{E}_{X_n} \calL_H\Ab{A}} \le C \abs{X} \, \norm{A}\) (by~\eqref{eq:lem-commutator-SLT-with-observable} from Lemma~\ref{lem:commutator with local observable} and Lemma~\ref{lem:conditional-expectation}~(c)).
For the second term, we use~\eqref{eq:comm-SLT-two-obs} from Lemma~\ref{lem:commutator with local observable} and Lemma~\ref{lem:conditional-expectation}~(e) for estimating the difference \(\norm{(\unit - \mathbb{E}_{X_n} ) \calL_H\Ab{A}}\) to get
\begin{equation}
    \label{eq:LPPLCOERterm22}
    \abs*{\trace[\big]{(\rho_* - \rho) A^*(\unit - \mathbb{E}_{X_n} ) \calL_H\Ab{A}}}
    \le
    C \norm{A}^2 \, \diam(X)^\ell \, \chi_{b,p}(n)
\end{equation}
Using \(
    \dist{X_n, \Lambda\setminus \Lambdag}
    \ge
    \dist{X, \Lambda \setminus \Lambdag} - n
\) for~\eqref{eq:LPPLCOERterm21}, we can pick \(n = \dist{X, \Lambda \setminus \Lambdag}/2\), say, to estimate
\begin{equation*}
    {\eqref{eq:LPPLCOERterm21}} + {\eqref{eq:LPPLCOERterm22}}
    \le
    C \norm{A}^2 \, \diam(X)^\ell \, \chi_{b/2,p}\pdist{X, \Lambda\setminus \Lambdag}
    .
\end{equation*}

Combining this with~\eqref{eq:LPPLCOERterm1}, we estimate~\eqref{eq:LPPLtoCOERdef2} by \(
    C \, \norm{A}^2 \diam(X)^\ell \, \chi_{b,p}\paren[\big]{\dist{X, \Lambda \setminus \Lambdag}}
\) and we thus arrive at~\eqref{eq:COERdefect}.
\\[2mm]
\noindent \underline{Proof of (v):}
This is obvious, because \(
    \abs[\big]{\expectation{A^* A}_{\rho} - \abs[\big]{\expectation{A}_{\rho}}^2}
    \le
    2 \norm{A}^2
\).
\\[2mm]
\noindent \underline{Proof of (vi):} By assumption~\eqref{eq:LDGdiscuss}, it holds that, for all observables \(A \in \alg_X\) and \(B \in \alg_Y\) localized in \(X\subset \Lambda\) and \(Y \subset \Lambda\), it holds that (recall~\eqref{eq:Jdef} for the definition of \(\calJ\))
\begin{equation}
    \label{eq:commLPPLtoLDG}
    \abs*{\expectation[\big]{\commutator*{\calJ_{H_*}\Ab{A} , B}}_{\rho_*}}
    \le
    \begin{multlined}[t]
        C \norm{A} \, \norm{B} \, \big[\diam(X)+ \diam(Y)\big]^\ell
        \\\times
        \chi_{b,p}\paren*{\dist{X, \Lambda \setminus \Lambdag}+ \dist{Y, \Lambda \setminus \Lambdag}} \,.
    \end{multlined}
\end{equation}
Similarly to the proof of (iv), the idea is now to replace \(H_* \to H\) and \(\rho_* \to \rho\) on the lhs.~\eqref{eq:commLPPLtoLDG} at the price of an error that is bounded in terms of the rhs.~of~\eqref{eq:commLPPLtoLDG}.
In order to do so, we will heavily exploit the symmetry of~\eqref{eq:commLPPLtoLDG} in \(A\) and \(B\) (recall the discussion around~\eqref{eq:Jdef}).
That is, we will prove the bound first only with \(\dist{X, \Lambda \setminus \Lambdag}\) in the argument of \(\chi_{b,p}\) and later obtain their sum (like on the rhs.~of~\eqref{eq:commLPPLtoLDG}) by symmetry (modulo changing \(b \to b/2\)).

To begin with, by the triangle inequality, we have
\begin{equation}
    \label{eq:commLPPLtoLDG2}
    \begin{aligned}
        \Alignindent
        \abs*{\expectation[\big]{\commutator*{\calJ_{H_*}\Ab{A} , B}}_{\rho_*} - \expectation[\big]{\commutator*{\calJ_{H}\Ab{A} , B}}_{\rho}}
        \\&\le
        \abs*{\expectation[\big]{\commutator*{\paren[\big]{\calJ_{H_*} - \calJ_{H}}\Ab{A} , B}}_{\rho} } + \abs*{\trace[\big]{(\rho - \rho_*) \commutator[\big]{B , \calJ_{H_*}\Ab{A}}}} \,.
    \end{aligned}
\end{equation}

For the first term on the rhs.~of~\eqref{eq:commLPPLtoLDG2}, we estimate
\begin{equation}
    \label{eq:invlioudiff}
    \norm*{\paren[\big]{\calJ_{H_*} - \calJ_{H}}\Ab{A}}
    \le
    \int_{\R} \D t \, w_g(t) \, \norm*{\E^{\I t H_*} A \E^{-\I t H_*} - \E^{\I t H} A \E^{-\I t H}}
    .
\end{equation}
Recalling \(H = H_* + J\), the difference between the two time evolutions can be written as
\begin{equation*}
    \E^{\I t H_*} \, A \, \E^{-\I t H_*} - \E^{\I t H} \, A \, \E^{-\I t H}
    =
    - \I \int_0^t \D s \, \E^{\I t H} \, \commutator[\big]{J , \E^{\I s H_*} \, A \, \E^{-\I s H_*}} \, \E^{-\I t H}
    .
\end{equation*}
We thus find that
\begin{equation}
    \label{eq:invlioudiff2}
    \begin{split}
        {\eqref{eq:invlioudiff}}
         & \le
        \int_{\R} \D t \, w_g(t) \, \abs{t} \, \sup_{s\in \intervalcc{0,t}} \norm[\big]{\commutator[\big]{J , \E^{\I s H_*} \, A \, \E^{-\I s H_*}}}
        \\&\le
        C
        \, \Norm{A}
        \, \diam(X)^{2d}
        \begin{aligned}[t]
            \paren[\bigg]{
                &\chi_{b,p}\pdist{X, \Lambda\setminus \Lambdag} \int_{I} \D t \, w_g(t) \, \abs{t} \, \paren[\big]{1 + \abs{t}}^{d/p}
                \\&+ \int_{\R\setminus I} \D t \, w_g(t) \, \abs{t} \, \paren[\big]{1 + \abs{t}}^{d/p}
            }
        \end{aligned}
        \\&\le
        C
        \, \Norm{A}
        \, \diam(X)^{2d}
        \, \chi_{b,p}\pdist{X, \Lambda\setminus \Lambdag}
        ,
    \end{split}
\end{equation}
where we denoted \(
    I
    :=
    \Set[\big]{t \in \R \given \abs{t} \le \paren[\big]{\dist{X, \Lambda\setminus \Lambdag}/(2v)}^p/2}
\).
Here, \(v\) is the Lieb-Robinson velocity from Lemma~\ref{lem:commutator-SLT-operator-evolved-observable}, which we employed in the second step.
In the final step, we used the stretched exponential decay of \(w_g\) (see~\eqref{eq:wbound} and Lemma~\ref{lem:weightfunctions} later) and possibly adjusted the constants \(C,b\), and \(p\).

The second term on the rhs.~of~\eqref{eq:commLPPLtoLDG2} can be estimated by means of~\eqref{eq:LPPLtracecomm}, since
\begin{equation*}
    \abs*{\trace[\big]{(\rho - \rho_*) \commutator[\big]{B , \calJ_{H_*}\Ab{A}}}}
    \le
    \norm{B} \, \norm[\big]{\commutator[\big]{\rho - \rho_* , \calJ_{H_*}\Ab{A}}}_{\tr}
    .
\end{equation*}
For \(\calJ_{H_*}\Ab{A}\) we now apply the local decomposition technique, analogously to the arguments around~\eqref{eq:LPPLCOERterm21}--\eqref{eq:LPPLCOERterm22}.
More precisely, taking \(A \in \Alg_X\), we now write \(
    \calJ_{H_*}\Ab{A}
    =
    \mathbb{E}_{X_n} \calJ_{H_*}\Ab{A} + (\unit - \mathbb{E}_{X_n} ) \calJ_{H_*}\Ab{A}
\) for some \(n\) to be chosen below, where \(X_n := \Set{ x \in \Lambda \given \dist{x,X} \le n }\) denotes the \(n\)-fattening of the set \(X \subset \Lambda\).
We now estimate the two terms separately.
For the first term, we employ~\eqref{eq:LPPLtracecomm} to bound
\begin{equation}
    \label{eq:commLPPLtoLDG2a}
    \norm{\commutator[\big]{\rho - \rho_* , \mathbb{E}_{X_n}\calJ_{H_*}\Ab{A}}}_{\tr}
    \le
    C \, \norm{A} \, \paren[\big]{\diam(X) + n}^\ell \, \chi_{b,p}\paren[\big]{\dist{X_n, \Lambda\setminus \Lambdag}}
\end{equation}
where we used that \(\mathbb{E}_{X_n} \calJ_{H_*}\Ab{A} \in \Alg_{X_n}\) (by definition) and \(\norm{\mathbb{E}_{X_n} \calJ_{H_*}\Ab{A}} \le \norm{A}\).
For the second term, we simply use Lemma~\ref{lem:conditional-expectation} together with~\eqref{eq:invliouQL} applied with \(\calI \to \calJ\) for estimating the difference \(\norm{(\unit - \mathbb{E}_{X_n} ) \calJ_{H_*}\Ab{A}}\) to get
\begin{equation}
    \label{eq:commLPPLtoLDG2b}
    \begin{aligned}
        \Alignindent
        \norm[\big]{\commutator[\big]{\rho - \rho_* , (\unit - \mathbb{E}_{X_n})\calJ_{H_*}\Ab{A}}}_{\tr}
        \\&\le
        2\norm{\rho - \rho_*}_{\tr} \, \norm{(\unit - \mathbb{E}_{X_n})\calJ_{H_*}\Ab{A}}
        \le
        C \norm{A} \, \diam(X)^\ell \, \chi_{b,p}(n) \,.
    \end{aligned}
\end{equation}
Using \(
    \dist{X_n, \Lambda\setminus \Lambdag}
    \ge
    \dist{X, \Lambda \setminus \Lambdag} - n
\) for~\eqref{eq:commLPPLtoLDG2a}, we can pick \(n = \dist{X, \Lambda \setminus \Lambdag}/2\), say, to estimate
\begin{equation*}
    {\eqref{eq:commLPPLtoLDG2a}} + {\eqref{eq:commLPPLtoLDG2b}}
    \le
    C \norm{A} \, \diam(X)^\ell \, \chi_{b/2,p}\pdist{X, \Lambda\setminus \Lambdag}
    .
\end{equation*}
Finally, as mentioned above, interchanging the roles of \(A\) and \(B\) (by symmetry of the original expression; recall the discussion around~\eqref{eq:Jdef}), we thus arrive at a bound on \(\abs[\big]{\expectation[\big]{\commutator*{\calJ_{H}\Ab{A} , B}}_{\rho}}\) of the form given by the rhs.~of~\eqref{eq:commLPPLtoLDG} (modulo changing \(b \to b/2\)).
Therefore, combining~\eqref{eq:commLPPLtoLDG} with~\eqref{eq:commLPPLtoLDG2}, and~\eqref{eq:invlioudiff}--\eqref{eq:invlioudiff2} as well as~\eqref{eq:commLPPLtoLDG2a}--\eqref{eq:commLPPLtoLDG2b}, we conclude the desired.
\\[2mm]
\noindent \underline{Proof of (vii):} Instead of~\eqref{eq:commLPPLtoLDG}, we start with (by assumption)
\begin{equation}
    \label{eq:exactinvert}
    \expectation[\big]{\commutator*{\calJ_{H_*}\Ab{A} , B}}_{\rho_*}
    =
    0
    \qquad \text{for all} \qquad
    A,B \in \Alg
    .
\end{equation}
Apart from this, the idea is identical to (iv) and (vi).
Hence, by means of the triangle inequality, we obtain the same two terms from~\eqref{eq:commLPPLtoLDG2}.
The first term can be estimated in exactly the same way as in~\eqref{eq:invlioudiff}--\eqref{eq:invlioudiff2}.
The second term in~\eqref{eq:commLPPLtoLDG2} has to be treated a bit differently as in~\eqref{eq:commLPPLtoLDG2a}--\eqref{eq:commLPPLtoLDG2b}, since we now assumed~\eqref{eq:AUTO} instead of~\eqref{eq:LPPLtracecomm}.

In fact, for this term, using \(\expectation{\interpunct}_\rho = \expectation{\tau \Ab{\interpunct}}_{\rho_*}\), that \(\tau\) is a \(*\)-automorphism, and~\eqref{eq:exactinvert}, we get
\begin{equation*}
    \abs*{\trace[\big]{(\rho - \rho_*) \commutator[\big]{B , \calJ_{H_*}\Ab{A}}}}
    =
    \abs*{\trace[\big]{\rho_* \commutator[\big]{\tau \Ab{B} , (\tau -\unit)\circ \calJ_{H_*}\Ab{A}}}}
    \le
    \norm{(\tau -\unit)\circ \calJ_{H_*}\Ab{A}}
    .
\end{equation*}
This can now be treated exactly as done in the argument around~\eqref{eq:commLPPLtoLDG2a}--\eqref{eq:commLPPLtoLDG2b} (i.e.~taking \(A \in \Alg_X\), writing \(
    \calJ_{H_*}\Ab{A}
    =
    \mathbb{E}_{X_n} \calJ_{H_*}\Ab{A} + (\unit - \mathbb{E}_{X_n} ) \calJ_{H_*}\Ab{A}
\), and estimating the two terms separately while optimizing in \(n\)).
\\[2mm]
\noindent \underline{Proof of (viii):} Without loss of generality, we may assume that \(X \subset \Lambda\) in~\eqref{eq:GAPdecay} is such that \(C \diam(X)^\ell \chi_{b, p}\pdist{X, \Lambda \setminus \Lambdag} < 1/2\), say (otherwise there is nothing to prove), and \(\diam(X)\le \dist{X, \Lambda \setminus \Lambdag}^\beta\) for some \(\beta <1\). By assumption, for such \(X \subset \Lambda\) and \(A \in \Alg_X\), we have
\begin{equation}
    \label{eq:localgap}
    \I \langle \psi, A^* \calL\Ab{A} \psi \rangle
    \ge
    \frac{g}{2}\paren*{\langle \psi, A^* A \psi \rangle - \abs[\big]{\langle \psi, A \psi \rangle}^2}
    \quadtext{with}
    \ket{\psi}
    =
    \otimes_{x \in \Lambda} \ket{\psi_x}
    ,
    \
    \norm{\psi_x}
    =
    1
    .
\end{equation}
Our claim will be a consequence of the following lemma.

\begin{lem}[cf.~Proposition 14 in~\cite{YSL24}]
    \label{lem:stabilizer}
    Take a Hermitian \(A \in \Alg_X\) as above and assume~\eqref{eq:localgap}.
    Denote the \(n\)-fattening of \(X\) by \(X_n := \Set{z \in \Lambda \given \dist{z, X} \le n}\).
    Then, there exists a Hermitian operator \(\widetilde{A} \in \Alg_{X_n}\) with \(n := \lfloor \dist{X, \Lambda \setminus \Lambdag} /2 \rfloor\) such that (i) \(\psi\) is an eigenvector of \(\widetilde{A}\), i.e.~\(\widetilde{A} \psi = \widetilde{E} \psi\) for some \(\widetilde{E} \in \R\), and (ii) we have the bound
    \begin{equation}
        \label{eq:stabilizer}
        \norm[\big]{\mathbb{E}_{X_n} \calJ\Ab{A} - \widetilde{A}}
        \le
        C \norm{A} \chi_{b,p}\pdist{X, \Lambda \setminus \Lambdag}
    \end{equation}
    with \(\calJ\) being defined in~\eqref{eq:Jdef}.
\end{lem}

Armed with Lemma~\ref{lem:stabilizer}, we now turn to estimating the lhs.~of~\eqref{eq:LDGweak_discuss}, which is given by \(
    \expectation{[\calJ \Ab{A}, B]}_\rho
    =
    \langle \psi, [\calJ \Ab{A}, B] \psi \rangle
\).
By the triangle inequality, we have
\begin{equation}
    \label{eq:stabilizersplit}
    \begin{split}
        \Alignindent
        \abs[\big]{\big\langle \psi, [\calJ\Ab{A}, B] \psi \big\rangle}
        \\&\le
        \abs[\big]{\big\langle \psi, [(\unit - \mathbb{E}_{X_n})\calJ\Ab{A}, B] \psi \big\rangle}
        + \abs[\big]{\big\langle \psi, [\mathbb{E}_{X_n}\calJ\Ab{A} - \widetilde{A}, B] \psi \big\rangle}
        + \abs[\big]{\big\langle \psi, [\widetilde{A}, B] \psi \big\rangle}
    \end{split}
\end{equation}
and estimate the three terms separately. The first term in~\eqref{eq:stabilizersplit} can be treated as in~\eqref{eq:commLPPLtoLDG2b}, yielding the bound
\begin{equation}
    \label{eq:targetbound}
    C \norm{A} \, \norm{B} \, \chi_{b,p}\pdist{X, \Lambda\setminus \Lambdag}
    ,
\end{equation}
where we additionally used that \(\diam(X)\le \dist{X, \Lambda \setminus \Lambdag}^\beta\).
For the second term, we employ~\eqref{eq:stabilizer}, yielding the same bound as for the first term.
The third term in~\eqref{eq:stabilizersplit} vanishes since \(\widetilde{A} \psi = \widetilde{E} \psi\).
Therefore, \eqref{eq:stabilizersplit} is bounded by~\eqref{eq:targetbound} and we have proven Proposition~\ref{prop:mechforlocalgap}~(viii).

It thus remains to give the proof of Lemma~\ref{lem:stabilizer}.
\begin{proof}[Proof of Lemma~\ref{lem:stabilizer}] The principal idea is similar to~\cite[Proposition 14]{YSL24}.
    To start with, we assume w.l.o.g.~that \(H \psi = 0\), i.e.~\(\psi\) is an eigenvector to the eigenvalue zero.
    Then, we note that, since \(\diam(X) \le \dist{X, \Lambda \setminus \Lambdag}^\beta\), the bound in~\eqref{eq:commLPPLtoLDG2b} implies that
    \begin{equation}
        \label{eq:1-E}
        \norm[\big]{(\unit - \mathbb{E}_{X_n})\calJ\Ab{A}}
        \le
        C \norm{A} \, \chi_{b,p}\pdist{X, \Lambda\setminus \Lambdag}
        .
    \end{equation}

    We continue by decomposing
    \begin{equation}
        \label{eq:decomp}
        \paren[\big]{\calJ\Ab{A} - \langle \psi, \calJ\Ab{A} \psi \rangle} \ket{\psi}
        =
        \ket{\phi_{\le n}} + \ket{\phi_{> n}}
    \end{equation}
    where we defined
    \begin{equation*}
        \ket{\phi_{\le n}}
        =
        \ket{\phi_{\le n}}_{X_n} \otimes \ket{\psi}_{X_n^c} :=\paren[\big]{\mathbb{E}_{X_n} \calJ\Ab{A} - \langle \psi, \mathbb{E}_{X_n} \calJ\Ab{A} \psi \rangle} \ket{\psi}
    \end{equation*}
    and \(\ket{\phi_{> n}}\), which is defined such that~\eqref{eq:decomp} holds, satisfies the bound
    \begin{equation*}
        \begin{split}
            \norm{\ket{\phi_{> n}}}
             & \le
            \norm[\big]{\paren[\big]{(\unit - \mathbb{E}_{X_n})\calJ\Ab{A} - \langle \psi , (\unit - \mathbb{E}_{X_n})\calJ\Ab{A} \psi \rangle} \ket{\psi}}
            \\&\le
            \norm[\big]{(\unit - \mathbb{E}_{X_n})\calJ\Ab{A}}
            \le
            C \norm{A} \, \chi_{b,p}\pdist{X, \Lambda\setminus \Lambdag}
            ,
        \end{split}
    \end{equation*}
    where in the last step we employed~\eqref{eq:1-E}.
    Then, similarly to~\cite[Lemma~15, eq.~(B49)]{YSL24}, one can compute \(\langle \phi_{\le n}, H^2 \phi_{\le n} \rangle\) and use the Payley-Zygmund inequality to show that the norm of \(\ket{\phi_{\le n}}\) is essentially bounded by the norm of \(\ket{\phi_{> n}}\).
    That is, in our case, we find
    \begin{equation}
        \label{eq:philenbound}
        \norm{\ket{\phi_{\le n}}}
        \le
        C \norm{A} \, \chi_{b,p}\pdist{X, \Lambda\setminus \Lambdag}
        ,
    \end{equation}
    as always modulo adjusting the constants \(C, b\) and \(p\).
    Hence, defining the Hermitian operator
    \begin{equation*}
        \widetilde{A}
        :=
        \mathbb{E}_{X_n} \calJ \Ab{A} - \paren[\big]{\ket{\phi_{\le n}}_{X_n} \bra{\psi} + \mathrm{h.c.}}
        ,
    \end{equation*}
    supported in \(X_n\) we easily see that \(
        \widetilde{A} \ket{\psi}
        =
        \langle \psi, \mathbb{E}_{X_n} \calJ\Ab{A} \psi \rangle \ket{\psi}
        =:
        \widetilde{E} \ket{\psi}
    \) and the bound~\eqref{eq:stabilizer} follows from~\eqref{eq:philenbound}.
\end{proof}
\noindent \underline{Proof of (ix):} This is obvious from the definitions~\eqref{eq:LDGdiscuss} and~\eqref{eq:LDGweak_discuss}.
\\[2mm]
This concludes the proof Proposition~\ref{prop:mechforlocalgap}. \qed

\subsection{Local dynamical gap for the examples in Section~\ref{subsec:examples}}
\label{subsec:proofexample}
In this section, we prove the systems considered in Section~\ref{subsec:examples} to have a local dynamical gap.

\subsubsection{Perturbations of the classical Ising model with weak interactions}
\label{subsubsec:Isingproof}
In this section, we prove the claim of a local dynamical gap from Section~\ref{subsec:Ising}, where we considered perturbations of the classical Ising model with weak interactions.
First, the (unique) ground state vector of~\eqref{eq:classIsing} and the associated ground state energy is easily found as
\begin{equation*}
    \psi_*
    =
    \otimes_{x \in \Lambda} \ket{\downarrow}
    \quad \text{satisfying} \quad
    H_* \psi_*
    =
    \paren*{- \abs{\Lambda} + \frac{1}{2}\sum_{x, y} \lambda(x-y)} \psi_*
    ,
\end{equation*}
and the associated spectral projection (ground state) is simply given by \(\rho_* = P_* = \ket{\psi_*} \bra{\psi_*}\).
We note that this is a globally gapped eigenstate of \(H_*\), since the ground state energy corresponding to \(\psi_*\) is separated by a spectral gap \(g \ge 2 - \norm{\lambda}_1 > 0\) from the first excited state.

For the following argument, it is important to observe that, for any given \(\Lambdag \subset \Lambda\), the ground state projection factorizes, i.e.\
\begin{equation}
    \label{eq:P*Ising}
    \rho_*
    =
    \paren[\bigg]{\bigotimes_{x \in \Lambdag} \ket{\downarrow} \bra{\downarrow}} \otimes \paren[\bigg]{\bigotimes_{x \in \Lambda \setminus \Lambdag} \ket{\downarrow} \bra{\downarrow}}
    =:
    \rho_*^\Lambdag \otimes \rho_*^{\Lambda \setminus \Lambdag}
    .
\end{equation}
Indeed, since \(\norm{\lambda}_1 < 2\), every ground state \(\rho\) of \(H = H_* + J\),
where \(J\) is a \emph{strictly} \(\Lambda \setminus \Lambdag\)-localized SLT Hamiltonian, as described in Section~\ref{subsec:Ising}, also factorizes as
\begin{equation}
    \label{eq:PIsing}
    \rho
    =
    \rho_*^\Lambdag \otimes \rho^{\Lambda \setminus \Lambdag}
    .
\end{equation}

In order to see this, first note that there exists an eigenbasis%
\footnote{%
    This is simply a common eigenbasis of \(H\) and \(H\big\vert_{\Lambdag} \otimes \unit_{\Alg_{\Lambda\setminus \Lambdag}}\), which commute.
} of \(H\) for which every eigenvector \(\psi\) of \(H\) can be written as a linear combination \(\sum_j c_j e_j \otimes \varphi_j\), where \(e_j \in \HSpace_{\Lambdag}\) are eigenvectors of
\begin{equation*}
    H\big\vert_{\Lambdag}
    :=
    \sum_{x \in \Lambdag} \sigma_x^3 \, + \, \frac{1}{2}\sum_{x, y \in \Lambdag} \lambda(x-y) \, \sigma_x^3 \sigma_y^3
\end{equation*}
to a \emph{common} eigenvalue and \(\varphi_j \in \HSpace_{\Lambda \setminus \Lambdag}\).
Then, to see~\eqref{eq:PIsing}, it suffices to realize that, for every \(x \in \Lambdag\), starting from the \emph{unique} ground state vector \(\psi_*\vert_{\Lambdag} = \otimes_{x \in \Lambdag} \ket{\downarrow}\) of \(H\big\vert_{\Lambdag}\), the energy cost for flipping the spin \(\ket{\downarrow}\) to \(\ket{\uparrow}\) in the first term of~\eqref{eq:classIsing} is two, whereas the potential gain stemming from the second summand in~\eqref{eq:classIsing} is bounded by \(\norm{\lambda}_1 <2\), yielding~\eqref{eq:PIsing}.
In particular, any \emph{overall} pure ground state of \(H\) can be obtained by tensorizing the unique separate ground state vector of \(H\vert_{\Lambdag}\), i.e.~\(\psi\vert_\Lambdag\), with an appropriate (not necessarily unique) minimizer \(\varphi^{\Lambda\setminus\Lambdag}\) of
\begin{equation*}
    \min_{\substack{\varphi \in \HSpace_{\Lambda \setminus \Lambdag} \\ \norm{\varphi} = 1}}
    \braket[\Big]{
        \psi_*\vert_{\Lambdag} \otimes \varphi ,
        (H - H\vert_{\Lambdag}) \, \psi_*\vert_{\Lambdag} \otimes \varphi
    }
    ,
\end{equation*}
i.e.~by conditioning on the first factor \(\psi_*\vert_{\Lambdag}\).
The pure ground state is then obtained as \(
    \ket{\psi_*\vert_{\Lambdag}} \bra{\psi_* \vert_\Lambdag} \otimes \ket{\varphi^{\Lambda\setminus\Lambdag}} \bra{\varphi^{\Lambda\setminus\Lambdag}}
\).

Therefore, combining~\eqref{eq:P*Ising} and~\eqref{eq:PIsing}, we have the following: For \(\rho\) being a \emph{pure state}, i.e.~\(\rho^{\Lambda\setminus\Lambdag}\) from~\eqref{eq:PIsing} can be written as \(
    \rho^{\Lambda\setminus\Lambdag}
    =
    \ket{\varphi^{\Lambda\setminus\Lambdag}} \bra{\varphi^{\Lambda\setminus\Lambdag}}
\),
there exists a unitary \(U \equiv U^{\Lambda\setminus\Lambdag} \in \Alg_{\Lambda\setminus \Lambdag}\) such that \(
    \ket{\varphi^{\Lambda\setminus\Lambdag}}
    =
    U \otimes_{x \in \Lambda \setminus \Lambdag} \ket{\downarrow}
\).
In particular, \(\rho = U \rho_* U^*\) and hence we have a norm-preserving \(*\)-automorphism \(\tau\Ab{A} := U^*AU\) on \(A \in \Alg\), which intertwines the ground states, i.e.~\(\expectation{\interpunct}_{\rho} = \expectation{\tau\Ab{\interpunct}}_{\rho_*}\), and satisfies~\eqref{eq:AUTO}.
By means of Proposition~\ref{prop:mechforlocalgap}~(vii) (note that, since \(\rho_*\) is spectrally gapped, it fulfills the additional assumption of~\ref{prop:mechforlocalgap}~(vii) by means of Proposition~\ref{prop:GDG}), we thus find that \(\rho\) satisfies Assumption \nameref{ass:localGAP_main}.
Finally, for a general (\emph{mixed}) state \(\rho\), we conclude the desired after noticing that Assumption \nameref{ass:localGAP_main} is invariant under taking convex combinations.

\subsubsection{Perturbations of gapped frustration free product states}
\label{subsec:frustfreeproof}
In this section, we prove the claim of a local dynamical gap from Section~\ref{subsec:frustfree}, where we considered perturbations of gapped frustration free Hamiltonians with a product ground state.

Similarly to Section~\ref{subsubsec:Isingproof}, one can easily verify that \(H_*\) from~\eqref{eq:frustfree} is globally gapped with its ground state vector being given by \(\otimes_{x \in \Lambda}\ket{\psi_{*,x}}\).
The same is true for all restrictions \(H_*\vert_{\Lambda'}\).

Moreover, for \(\Lambdag \subset \Lambda\) and a fixed exponent \(\beta >0\), consider \(X \subset \Lambda\) satisfying \(\diam(X) \le \dist{X, \Lambda\setminus\Lambdag}^\beta\).
Under Assumptions \nameref{ass:frustfree}, \nameref{ass:regular}, and \nameref{ass:gap}, the authors of~\cite{BRDF2021} have proven the following: Let \(\ket{\psi}\) be a ground state vector of \(H = H_* + J\), where \(J\) is a \emph{strongly} \(\Lambda \setminus \Lambdag\)-localized SLT Hamiltonian, as described in Section~\ref{subsec:frustfree}, and \(P_{*,X}\) denote the projection onto the ground state vector \(\ket{\psi_*\vert_{X}}\) of \(H_*\vert_X\).
Then it holds that
\begin{equation}
    \label{eq:tracenorm}
    \norm{\rho - \tilde{\rho}}_{\mathup{tr}}
    \le
    C \exp\paren[\big]{- \paren[\big]{\dist{X, \Lambda\setminus \Lambdag}}^q}
\end{equation}
for some \(C, q>0\) and \(\norm{\interpunct}_{\mathup{tr}}\) being the trace norm.
Here, \(\rho \equiv P\) and \(\tilde{\rho} \equiv \tilde{P}\) denote the orthogonal projections on \(\ket{\psi}\) and \(P_{*,X} \ket{\psi}\), respectively, i.e.~they are \emph{pure states}.

Due to the product structure of the ground state vector of \(H_*\) and its restrictions, we easily see that \(\tilde{\rho}\) can be written as
\begin{equation*}
    \tilde{\rho}
    =
    \ket{\psi_*\vert_{X}} \bra{\psi_*\vert_X} \otimes \tilde{\rho}^{\Lambda\setminus X}
\end{equation*}
for some state \(\tilde{\rho}^{\Lambda\setminus X}\) on \(\Lambda\setminus X\).
This means that, analogously to Section~\ref{subsubsec:Isingproof},
\(\expectation{\interpunct}_{\rho_*}\) and \(\expectation{\interpunct}_{\tilde{\rho}}\) can be intertwined by a norm preserving \(*\)-automorphism \(\tau\) satisfying~\eqref{eq:AUTO}.
In particular, by means of Proposition~\ref{prop:mechforlocalgap}~(vii) (note that, since \(\rho_*\) is spectrally gapped, it fulfills the additional assumption of~\ref{prop:mechforlocalgap}~(vii) by means of Proposition~\ref{prop:GDG}), we thus find that \(\tilde{\rho}\) satisfies Assumption \nameref{ass:localGAP_main} -- but only for observables supported in \(X \subset \Lambda\) with \(\diam(X) \le \dist{X, \Lambda\setminus\Lambdag}^\beta\) and without \(\dist{Y, \Lambda \setminus \Lambdag}\) in the argument of \(\chi_{b,p}\); that is, Assumption \nameref{ass:localGAPweak}.
This implies, by means of Proposition~\ref{prop:mechforlocalgap}~(i)+(vi) (trivially modified to the setting of \nameref{ass:localGAPweak}) and~\eqref{eq:tracenorm}, that \(\rho\) satisfies Assumption \nameref{ass:localGAPweak}.
Finally, for a general (\emph{mixed}) state \(\rho\), we conclude the desired by taking convex combinations (as at the end of the argument in Section~\ref{subsubsec:Isingproof}).

\statement{Acknowledgments}
It is a pleasure to thank Stefan Teufel for numerous interesting discussions, fruitful collaboration and helpful comments on an earlier version of the manuscript.
J.~H.~acknowledges partial financial support by the ERC Advanced Grant ‘RMTBeyond’ No.~101020331.
T.~W.~acknowledges financial support from the \foreignlanguage{ngerman}{Deutsche Forschungsgemeinschaft} (DFG, German Research Foundation) –
465199066
.

\statement{Conflict of interest}
The authors have no conflicts to disclose.

\statement{Author contributions}
\emph{Joscha~Henheik:} writing -- original draft (equal); writing -- review and editing (equal).
\emph{Tom~Wessel:} writing -- original draft (equal); writing -- review and editing (equal).

\statement{Data availability}
Data sharing is not applicable to this article as no new data were created or analyzed in this study.

\appendix

\section{Technical lemmata}
\label{app:technical}

In this section, we prove the technical lemmata required for the construction of the NEASS\@.
We begin with some general properties of the functions \(\chi_{b,p}\) and Lieb-Robinson bounds for \((b,p)\)-localized SLT-operators in Appendix~\ref{app:properties-of-chi-b-p-and-LR-bounds}.
In Appendix~\ref{app:commutators-and-dynamics-of-SLT-operators}, we prove that the various operations used in the construction of the NEASS preserve locality.
In Appendix~\ref{app:inverse-liouvillian} we recall the construction of the quasi-local inverse of the Liouvillian and prove that it also preserves locality of SLT-operators.
Finally, Appendices~\ref{app:proof-SLTlocal} and~\ref{app:proof-LDGweak} are concerned with the proofs of Lemmata~\ref{lem:SLTlocal} and~\ref{lem:LDGweak}, respectively.

In all proofs, \(C>0\) is a generic constant that might change within the computations.
It can in particular depend on all the parameters chosen in the statements, but it is uniform in the chosen lattice and the operators appearing.

\subsection{Properties of the decay function\texorpdfstring{ \(\chi_{b,p}\)} {}}
\label{app:properties-of-chi-b-p-and-LR-bounds}

Let us first collect some properties of the decay function \(\chi_{b,p}\) we use in the definition of the interaction norm.
From~\cite[Lemma~7.2.3]{Maier2022} we have the following Lemma, where we simplified the statements.

\begin{lem} \label{lem:chifnct}
    For any \(b \geq 0\) and \(s\in \intervaloc{0,1}\), the function \(\chi_{b,p}\) satisfies the following properties:
    \begin{enumerate}[label=\textup{(\alph*)}, ref=\thetheorem\,(\alph*)]
        \item \label{item:lem-chifnct-logarithmically-superadditive}
            \(\chi_{b,p}\) is logarithmically superadditive, i.e.\ \(\chi_{b,p}(x+y) \geq \chi_{b,p}(x) \, \chi_{b,p}(y)\) for all \(x\), \(y \geq 0\).
        \item \label{item:lem-chifnct-superpolynomial-decay}
            For every \(b>0\) and \(k \geq 0\) there exists a constant \(C>0\) such that
            \begin{equation*}
                \sup_{x \geq 0} x^k \, \chi_{b,p}(x)
                =
                C
                .
            \end{equation*}
    \end{enumerate}
\end{lem}

As a direct consequence of Lemma~\ref{item:lem-chifnct-superpolynomial-decay}, we get the following Lemma, which we write out to fix the constant and recall it in later proofs.

\begin{lem}
    \label{lem:cardinality-over-diameter-bound}
    Let \(d\in \N\), \(\Cvol>0\), \(b>0\), \(p\in \intervaloc{0,1}\) and \(k\in \N\).
    Then there exists a constant \(\Ccod{b,p,k}>0\) such that for all lattices \(\Lambda\in \mathcal{G}(d,\Cvol)\) and sets \(Z\subset \Lambda\)
    \begin{equation*}
        \abs{Z}^k \, \chi_{b,p}\pdiam{Z}
        \leq
        \Ccod{b,p,k}
        .
    \end{equation*}
\end{lem}

Before we state the Lieb-Robinson bound, which is a crucial ingredient in the proof, let us briefly recall the time-dependent Heisenberg evolution.
For a time-dependent interaction defined on an interval \(I\subset \R\) with corresponding SLT operator \(H(t)\), let \(\tau_{t,s}\) be the unique solution of
\begin{equation*}
    -\I \, \frac{\D}{\D t} \, \tau_{t,s}(A)
    =
    \tau_{t,s}\paren[\big]{\commutator{H(t), A}}
    \qquadtext{and}
    \tau_{s,s} = \operatorname{Id}
    \qquad\text{for all \(s\), \(t\in I\)}
    .
\end{equation*}
This was already used with a different time scaling for the Hamiltonian \(H_\epsi(t)\) in the proof of Theorem~\ref{thm:time-dependent-NEASS}.
Under locality assumptions on the Hamiltonian, one finds the following Lieb-Robinson bound.

\begin{lem}[{Lieb-Robinson bound~\cite[Theorem~7.3.3]{Maier2022}}]
    \label{lem:LRB}
    Let \(d\in \N\), \(\Cvol>0\), \(b'>b>0\), \(p\in \intervaloc{0,1}\) and \(k\in \N\).
    There exists constants \(C\) and \(c>0\) such that for all lattices \(\Lambda\in \mathcal{G}(d,\Cvol)\), intervals \(I\subset \R\), time-dependent interactions \(\Phi\), disjoint subsets \(X\), \(Y \subset \Lambda\), observables \(A\in \alg_X\) and \(B\in \alg_Y\), and \(s\), \(t\in I\) it holds that
    \begin{equation}
        \label{eq:LR-bound-general}
        \norm[\big]{\commutator[\big]{\tau_{t,s}(A),B}}
        \leq
        C \, \norm{A} \, \norm{B}
        \, \paren[\big]{
            \E^{b \, v \, \abs{t-s}}-1
        }
        \, D(X,Y)
        ,
    \end{equation}
    where \(v = c \, \norm{\Phi}_{b',p} / b\) is the Lieb-Robinson velocity and
    \begin{align*}
        D(X,Y)
        &:=
        \min \List[\bigg]{
            \sum_{x\in X} \chi_{b,p}\pdist{x,Y},
            \sum_{y\in Y} \chi_{b,p}\pdist{y,X}
        }
        \\&\phantom{:}\leq
        \min \List[\big]{\abs{X},\abs{Y}}
        \, \chi_{b,p}\pdist{X,Y}
        .
    \end{align*}
\end{lem}

The Lieb-Robinson velocity is defined including the \(1/b\) because~\eqref{eq:LR-bound-general} can be bounded by
\begin{equation*}
    C \, \norm{A} \, \norm{B}
    \, \min \List[\big]{\abs{X},\abs{Y}}
    \, \E^{b \, \paren[\big]{ v \, \abs{t-s} - \dist{X,Y}^p}}
    .
\end{equation*}

\subsection{Commutators and dynamics of localized SLT-operators}
\label{app:commutators-and-dynamics-of-SLT-operators}

\begin{lem}[Commutator with local observable]
    \label{lem:commutator with local observable}
    Let \(d\in \N\), \(\Cvol>0\), \(b>0\), \(p\in \intervaloc{0,1}\) and \(k\in \N\).
    There exists constants \(C\) and \(C_k>0\) such that for all lattices \(\Lambda\in \mathcal{G}(d,\Cvol)\), subsets \(\Omega\), \(X \subset \Lambda\), SLT operators \(A_1\) and observables \(O\in \alg_X\) it holds that
    \begin{equation}
        \label{eq:lem-commutator-SLT-with-observable}
        \Norm[\big]{\commutator{A_1,O}}
        \leq
        2 \, \Norm{O} \, \abs{X} \, \chi_{b,p}\pdist{X,\Omega} \, \Norm{\Phi_{A_1}}_{b,p;\Omega}
        .
    \end{equation}
    For a second observable \(\tilde{O} \in \Alg_Y\), it holds that
    \begin{equation}
        \label{eq:comm-SLT-two-obs}
        \Norm[\big]{\commutator{\commutator{A_1, O}, \tilde{O}}}
        \leq
        4 \, \Norm{O} \, \Norm{\tilde{O}} \, \abs{X} \, \chi_{b,p}\pdist{X,Y} \, \chi_{b,p} \pdist{X,\Omega}\, \Norm{\Phi_{A_1}}_{b,p;\Omega}
        .
    \end{equation}

    Finally, if additionally also \(A_2,\dotsc,A_k\) are SLT operators, then
    \begin{equation}
        \label{eq:lem-multi-commutator-SLT-with-observable}
        \Norm[\big]{\adjoint_{A_k}\dotsb\adjoint_{A_1}(O)}
        \leq
        C_k
        \, \Norm{O}
        \, \abs{X}^k
        \, \chi_{b,p}\pdist{X,\Omega}
        \, \Norm{\Phi_{A_1}}_{b,p;\Omega}
        \, \prod_{j=2}^k \, \norm{\Phi_{A_j}}_{b,p}
        .
    \end{equation}
    All three bounds, in particular, also hold for \(\Omega=\Lambda\), where \(\norm{\interpunct}_{b,p;\Omega} = \norm{\interpunct}_{b,p}\) and \(\dist{X,\Omega} = 0\).
\end{lem}

\begin{proof}
    We begin with the first statement and write \(A=A_1\).
    Since \(\commutator{\Phi_A(Z),O}\) vanishes whenever \(Z\cap X=\emptyset\), we find
    \begin{align*}
        \Norm[\big]{\commutator{A,O}}
        &\leq
        \sumstack[lr]{Z\subset\Lambda \colonpunct\\ Z\cap X \neq \emptyset} \, 2\, \Norm[\big]{\Phi_A(Z)} \,\Norm{O}
        \\&\leq
        2 \, \Norm{O} \, \sum_{z\in X}
        \sumstack{Z\subset\Lambda \colonpunct\\ z\in Z}
        \frac{\Norm[\big]{\Phi_A(Z)}}{\chi_{b,p}\pdiam{Z}\,\chi_{b,p}\pdist{z,\Omega}}
        \, \chi_{b,p}\pdist{z,\Omega}
        \\&\leq
        2 \, \Norm{O} \, \sum_{z\in X}\, \chi_{b,p}\pdist{z,\Omega}\, \Norm{\Phi_A}_{b,p;\Omega}
        \\&\leq
        2 \, \Norm{O} \, \abs{X}\, \chi_{b,p}\pdist{X,\Omega}\, \Norm{\Phi_A}_{b,p;\Omega}
        ,
    \end{align*}
    where we just overcount in the second inequality.
    Clearly, the same statement also holds with \(\Omega=\Lambda\).
    The proof of~\eqref{eq:comm-SLT-two-obs} is analogous to the proof of~\eqref{eq:lem-commutator-SLT-with-observable} and so omitted.

    We conclude by proving~\eqref{eq:lem-multi-commutator-SLT-with-observable} using induction.
    Note that the outer operators are all SLT-operators on \(\Lambda\).
    The \(k=1\) case is given in~\eqref{eq:lem-commutator-SLT-with-observable}.
    We now assume~\eqref{eq:lem-multi-commutator-SLT-with-observable} for some fixed \(k\) and with \(\Omega=\Lambda\).
    Then, we add a further commutator with \(A_0\) to conclude
    \begin{align*}
        \Alignindent
        \Norm[\big]{\adjoint_{A_k}\dotsb\adjoint_{A_1}\adjoint_{A_0}(O)}
        \\&\leq
        \sumstack[lr]{Z\subset\Lambda \colonpunct\\ Z\cap X \neq \emptyset}
        \Norm[\big]{\adjoint_{A_k}\dotsb\adjoint_{A_1}\commutator{\Phi_{A_0}(Z),O}}
        \\&\leq
        \sumstack[lr]{Z\subset\Lambda \colonpunct\\ Z\cap X \neq \emptyset}
        C_k
        \, \Norm[\big]{\commutator[\big]{\Phi_{A_0}(Z),O}}
        \, \abs{X\cup Z}^k
        \, \prod_{j=1}^k \, \norm{\Phi_{A_j}}_{b,p}
        \\&\leq
        \begin{aligned}[t]
            &C_k \, 2^k \, \norm{O} \, \abs{X}^k \, \sum_{z\in X}
            \chi_{b,p}\pdist{z,\Omega} \, \prod_{j=1}^k \, \norm{\Phi_{A_j}}_{b,p}
            \\&\times\sumstack{Z\subset\Lambda \colonpunct\\ z\in Z}
            \frac{\Norm[\big]{\Phi_{A_0}(Z)}}{\chi_{b,p}\pdiam{Z}\,\chi_{b,p}\pdist{z,\Omega}}
            \, \chi_{b,p}\pdiam{Z}
            \, \abs{Z}^k
        \end{aligned}
        \\&\leq
        \Ccod{b,p,k} \, C_k \, 2^k
        \, \norm{O}
        \, \abs{X}^{k+1}
        \, \chi_{b,p}\pdist{X,\Omega}
        \, \norm{\Phi_{A_0}}_{b,p;\Omega}
        \, \prod_{j=1}^k \, \norm{\Phi_{A_j}}_{b,p}
        ,
    \end{align*}
    where we used Lemma~\ref{lem:cardinality-over-diameter-bound} in the last step.
    This finishes the induction.
\end{proof}

\begin{lem}[Multi-commutators]
    \label{lem:multi-commutators-of-SLT-are-SLT}
    Let \(d\in \N\), \(\Cvol>0\), \(b>0\), \(p\in \intervaloc{0,1}\), \(\varepsilon>0\) and \(k\in \N\).
    There exists a constant \(C>0\) such that for all lattices \(\Lambda\in \mathcal{G}(d,\Cvol)\), subsets \(\Omega\subset \Lambda\) and SLT operators \(A_0,\dotsc,A_k\) it holds that
    \begin{equation}
        \label{eq:lem-multi-commutators-of-SLT-are-SLT}
        \Norm[\big]{\Phi_{\adjoint_{A_k}\dotsb\adjoint_{A_1}(A_0)}}_{b,p;\Omega}
        \leq
        C
        \, \Norm{\Phi_{A_0}}_{b+\varepsilon,p;\Omega}
        \, \prod_{j=1}^k \, \norm{\Phi_{A_j}}_{2b+\varepsilon,p}
        .
    \end{equation}
\end{lem}

\begin{proof}
    For the proof we first need to constructed an interaction for the commutator of two SLT operators \(A\) and \(B\).
    It turns out that it can be given as
    \begin{equation*}
        \Phi_{\commutator{A,B}}(Z)
        =
        \sumstack[lr]{X,Y\subset \Lambda\colonpunct\\ X\cup Y=Z\\ X\cap Y\neq\emptyset}
        \commutator[\big]{\Phi_A(X),\Phi_B(Y)}
        .
    \end{equation*}
    Then,
    \begin{align*}
        \Alignindent
        \sumstack{Z\subset\Lambda\colonpunct\\z\in Z}
        \frac
            {\Norm[\big]{\Phi_{\commutator{A,B}}(Z)}}
            {\chi_{b,p}\pdiam{Z}\,\chi_{b,p}\pdist{z,\Omega}}
        \\&\leq
        \sumstack[l]{Z\subset\Lambda\colonpunct\\z\in Z}
        \sumstack{X,Y\subset\Lambda\colonpunct\\ X\cup Y=Z\\ X\cap Y\neq\varnothing}
        \frac
            {2 \, \Norm[\big]{\Phi_{A}(X)} \, \Norm[\big]{\Phi_{B}(Y)}}
            {\chi_{b,p}\pdiam{X}\,\chi_{b,p}\pdiam{Y}\,\chi_{b,p}\pdist{z,\Omega}}
        ,
    \end{align*}
    where we used \(\diam(Z) \leq \diam(X) + \diam(Y)\) and the properties of \(\chi_{b,p}\).
    The above sum can be bounded by the sum of the terms where \(z\in X\) or \(z\in Y\).
    The latter can be upper bounded by
    \begin{align*}
        \Alignindent
        2\,
        \sumstack[l]{Y\subset\Lambda\colonpunct\\z\in Y}
        \frac
            {\Norm[\big]{\Phi_{B}(Y)}}
            {\chi_{b,p}\pdiam{Y}\,\chi_{b,p}\pdist{z,\Omega}}
        \sum_{x\in Y} \sumstack{X\subset\Lambda\colonpunct\\x\in X}
        \frac
            {\Norm[\big]{\Phi_{A}(X)}}
            {\chi_{b,p}\pdiam{X}}
        \\&\leq
        2\,
        \sumstack[l]{Y\subset\Lambda\colonpunct\\z\in Y}
        \frac
            {\Norm[\big]{\Phi_{B}(Y)}}
            {\chi_{b+\varepsilon,p}\pdiam{Y} \, \chi_{b+\varepsilon,p}\pdist{z,\Omega}}
        \, \chi_{\varepsilon,p}\pdiam{Y}
        \, \abs{Y}
        \, \Norm[\big]{\Phi_{A}}_{b,p}
        \\&\leq
        C \, \Ccod{\varepsilon,p,1} \, \sumstack[l]{Y\subset\Lambda\colonpunct\\z\in Y}
        \frac
            {\Norm[\big]{\Phi_{B}(Y)}}
            {\chi_{b+\varepsilon,p}\pdiam{Y} \, \chi_{b+\varepsilon,p}\pdist{z,\Omega}}
        \, \Norm[\big]{\Phi_{A}}_{b,p}
        \\&\leq
        C
        \, \Norm{\Phi_{B}}_{b+\varepsilon,p;\Omega}
        \, \Norm{\Phi_{A}}_{b,p}
        ,
    \end{align*}
    where we used Lemma~\ref{lem:cardinality-over-diameter-bound} in the third step.
    Using \(
        \chi_{b,p}\pdist{z,\Omega}
        \geq
        \chi_{b,p}\pdiam{X} \, \chi_{b,p}\pdist{x,\Omega}
    \) for all \(z,x\in X\), the part of the sum where \(z\in X\) can be bounded by
    \begin{align*}
        \Alignindent
        2 \, \sumstack[l]{X\subset\Lambda\colonpunct\\z\in X}
        \frac
            {\Norm[\big]{\Phi_{A}(X)}}
            {\chi_{2b,p}\pdiam{X}}
        \sum_{y\in X}
        \sumstack{Y\subset\Lambda\colonpunct\\y\in Y}
        \frac{1}{\chi_{b,p}\pdist{y,\Omega}}
        \frac
            {\Norm[\big]{\Phi_{B}(Y)}}
            {\chi_{b,p}\pdiam{Y}}
        \\ &\leq
        2 \, \sumstack[l]{X\subset\Lambda\colonpunct\\z\in X}
        \frac
            {\Norm[\big]{\Phi_{A}(X)}}
            {\chi_{2b+\varepsilon,p}\pdiam{X}}
        \, \chi_{\varepsilon,p}\pdiam{X}
        \, \abs{X}
        \, \Norm[\big]{\Phi_{B}}_{b,p;\Omega}
        \\ &\leq
        C
        \, \Norm[\big]{\Phi_{A}}_{2b+\varepsilon,p}
        \, \Norm[\big]{\Phi_{B}}_{b,p;\Omega}
        .
    \end{align*}
    Both bounds together prove the claim for \(k=1\).
    To proceed by induction we assume that the statement holds for some fixed \(k\).
    Then, applying first the statement for \(k=1\) and then \(k=k\) both with \(\varepsilon/2\) we obtain
    \begin{align*}
        \Norm[\big]{\Phi_{\adjoint_{A_{k+1}}\dotsb\adjoint_{A_1}(A_0)}}_{b,p;\Omega}
        &\leq
        C
        \, \Norm{\Phi_{\adjoint_{A_{k}}\dotsb\adjoint_{A_1}(A_0)}}_{b+\varepsilon/2,p;\Omega}
        \, \norm{\Phi_{A_{k+1}}}_{2b+\varepsilon/2,p}
        \\&\leq
        C
        \, \Norm{\Phi_{A_0}}_{b+\varepsilon,p;\Omega}
        \, \prod_{j=1}^{k+1} \, \norm{\Phi_{A_j}}_{2b+\varepsilon,p}
        .
    \end{align*}
\end{proof}

For the following statements we need to approximate the time evolution of local operators, which in principle live on the whole lattice.
This can be done by a so called conditional expectation, which is just the partial trace in our case of finite spin systems.
We collect its properties in the following lemma.

\begin{lem}[{\cite[Lemma~4.1]{NSY2019}}]
    \label{lem:conditional-expectation}
    Let \(\Lambda\) be a lattice and \(X\subset\Lambda\).
    Then there exists a unit-preserving, completely positive linear map \(\cmeansym_{X}\colon \alg_\Lambda\rightarrow\alg_\Lambda\) satisfying
    \begin{enumerate}[label=\textup{(\alph*)}]
        \item
            \(\cmean{X}{\alg_\Lambda}\subset\alg_X\);
        \item \label{lem:cmean E(ABC)=AE(B)C}
            \(\cmean{X}{ABC} = A \, \cmean{X}{B} \, C\) for all \(B\in\alg_\Lambda\) and \(A,C\in\alg_{X}\);
            This in particular implies \(\cmean{X}{A}=A\) for all \(A\in\alg_{X}\);
        \item
            \(\Norm{\cmeansym_{X}}=1\);
        \item
            \(\cmeansym_{X} \circ \cmeansym_{Y} = \cmeansym_{X\cap Y}\), for \(X,Y\subset\Lambda\);
        \item \label{lem:cmean bound for difference A-E(A)}
            If \(A\in\alg_\Lambda\) satisfies
            \begin{equation}
                \Norm[\big]{\commutator{A,B}}
                \leq
                \eta \, \Norm{A} \, \Norm{B}
                \quadtext{for all}
                B\in\alg_{\Lambda\setminus X}
                ,
            \end{equation}
            for some \(\eta>0\), then
            \begin{equation}
                \Norm{A-\cmean{X}{A}} \leq \eta \, \Norm{A}
                .
            \end{equation}
    \end{enumerate}
\end{lem}

Together with the Lieb-Robinson bound we can now obtain the following.

\begin{lem}[Dynamics]
    \label{lem:commutator-SLT-operator-evolved-observable}
    Let \(d\in \N\), \(\Cvol>0\), \(b,b'>0\), and \(p,p'\in \intervaloc{0,1}\) satisfying \(p'<1\) or \(b'<b\).
    There exists constants \(C\) and \(c>0\) such that for all lattices \(\Lambda\in \mathcal{G}(d,\Cvol)\), the following holds:
    Let \(I\subset \R\) an interval, the interaction \(\Phi\) generate the dynamics \(\tau_{t,s}\) with Lieb-Robinson velocity \(v = c \, \norm{\Phi}_{b,p} / b\).
    For every \(\chi_{b',p'}\)-SLT operator \(A\), subsets \(\Omega\), \(X\subset \Lambda\), observables \(O\in \alg_X\), and \(t,s\in I\) it holds that
    \begin{equation*}
        \Norm[\Big]{\commutator[\big]{A,\tau_{t,s}(O)}}
        \leq
        C
        \, \Norm{O}
        \, \abs{X}
        \, \abs{X_{(2\vts)^{1/p}}}
        \, \Norm{\Phi_A}_{b',p';\Omega}
        \, \chi_{b',p'}\pdist{X_{(2\vts)^{1/p}},\Omega}
    \end{equation*}
\end{lem}

\begin{proof}
    We use the local decomposition technique similar to~\cite[Section 5]{NSY2019}.
    Therefore, let
    \begin{align*}
        \Delta_0
        &:=
        \cmean[\big]{X_{\vts}}{\tau_{t,s}(O)}\\
    \shortintertext{and}
        \Delta_k
        &:=
        \cmean[\big]{X_{(\vts+k)^{1/p}}}{\tau_{t,s}(O)} - \cmean[\big]{X_{(\vts+k-1)^{1/p}}}{\tau_{t,s}(O)}
        ,
    \end{align*}
    so that \(\tau_{t,s}(O)=\sum_{k=0}^\infty \Delta_k\), where the sum is finite since eventually \(X_{(\vts+k)^{1/p}}=\Lambda\).
    By the properties of the conditional expectation
    \begin{align*}
        \Norm{\Delta_0}
        &\leq
        \Norm{\tau_{t,s}(O)}
        =
        \Norm{O}
    \intertext{and for \(k\geq1\) it holds that}
        \Delta_k
        &=
        \cmean[\Big]{X_{(\vts+k)^{1/p}}}{\paren[\big]{1-\cmeansym_{X_{(\vts+k-1)^{1/p}}}}\,\tau_{t,s}(O)}
    \shortintertext{and thus}
        \Norm{\Delta_k}
        &\leq
        \Norm[\big]{\paren[\big]{1-\cmeansym_{X_{(\vts+k-1)^{1/p}}}}\,\tau_{t,s}(O)}
        .
    \end{align*}
    Furthermore, for all \(B\in\alg_{\Lambda\setminus X_{(\vts+k-1)^{1/p}}}\) by the Lieb-Robinson bound (Lemma~\ref{lem:LRB})
    \begin{equation*}
        \Norm[\big]{\commutator{\tau_{t,s}(O),B}}
        \leq
        C_\LR \, \Norm{O} \, \Norm{B} \, \abs{X} \, \E^{b\paren[\big]{\vts-(\vts+k-1)}}
        =
        C_\LR \, \E \, \Norm{O} \, \Norm{B} \, \abs{X} \, \E^{-b k}
    \end{equation*}
    because \({\dist{X,\Lambda\setminus X_{\vts+k-1}}}\geq\vts+k-1\), and thus by Lemma~\ref{lem:conditional-expectation}
    \begin{equation*}
        \norm{\Delta_k}
        \leq
        C_\LR \, \E \, \norm{O} \, \abs{X} \, \E^{-b k}
        .
    \end{equation*}
    Now we can apply Lemma~\ref{lem:commutator with local observable} to each of the summands in the decomposition
    \begin{align*}
        \Norm[\big]{\commutator[\big]{A,\tau_{t,s}(O)}}
        &\leq
        \sum_{k=0}^\infty
        \, \Norm[\big]{\commutator[\big]{A,\Delta_k}}
        \\&\leq
        \sum_{k=0}^\infty
        2 \, \Norm{\Delta_k} \, \abs[\big]{X_{(\vts+k)^{1/p}}} \, \chi_{b',p'}\pdist{X_{(\vts+k)^{1/p}},\Omega} \, \Norm{\Phi_A}_{b',p';\Omega}
        \\&\leq
        \tilde{C} \, \Norm{O} \, \abs{X} \, \Norm{\Phi_A}_{b',p';\Omega}
        \sum_{k=0}^\infty
        \abs[\big]{X_{(\vts+k)^{1/p}}} \, \chi_{b',p'}\pdist{X_{(\vts+k)^{1/p}},\Omega} \, \E^{-b k}.
        \\&\leq
        \begin{aligned}[t]
            &
            \tilde{C} \, \Norm{O} \, \abs{X} \, \abs{X_{(2\vts)^{1/p}}} \, \Norm{\Phi_A}_{b',p';\Omega} \, \chi_{b',p'}\pdist{X_{(2\vts)^{1/p}},\Omega}
            \\&\times
            \sum_{k=0}^\infty
            \paren{1+\Cvol \, (2k)^{d/p}} \, \chi_{b',p'}^{-1}(k) \, \E^{-b k},
        \end{aligned}
    \end{align*}
    where we abbreviated \(\tilde{C} = 2 \max\List{1,C_\LR \, \E}\) and used \((\vts+k)^{1/p} \leq (2\vts)^{1/p} + (2k)^{1/p}\).
    To conclude the result, we observe, that the series is bounded for \(p'<1\) or \(b'<b\) if \(p'=1\).
\end{proof}

\begin{lem}[Conjugation with unitaries]
    \label{lem:conjugation-with-unitaries-evolved-SLT-is-SLT}
    Let \(d\in \N\), \(\Cvol>0\), \(\varepsilon>0\), \(a\), \(b>0\), and \(p,q\in \intervaloc{0,1}\) satisfying \(p<q\) or \(p=q\) and \(a > (2^p+1) \, b\).
    There exists constants \(C\) and \(c>0\) such that for all lattices \(\Lambda\in \mathcal{G}(d,\Cvol)\), the following holds:
    For SLT operators \(D\) and \(S\) it holds that \(A := \evol{\I \, S}{D}\) is an SLT operator as well.
    More precisely, there exists an interaction \(\Phi_A\) such that
    \begin{equation*}
        \Norm{\Phi_A}_{b,p;\Omega}
        \leq
        C \, \e^{c \, \norm{\Phi_S}_{a,q}}
        \, \norm[\big]{\Phi_D}_{b+\varepsilon,p;\Omega}
    \end{equation*}
\end{lem}

\begin{proof}
    The proof uses the same technique as the proof of Lemma~\ref{lem:commutator-SLT-operator-evolved-observable}.

    First fix \(X\subset \Lambda\) and \(O\in \alg_X\) and denote \(\tau(O) = \evol{\I S}{O}\).
    Then define
    \begin{align*}
        \Delta_0(O)
        &:=
        \cmean[\big]{X}{\tau(O)}\\
    \shortintertext{and}
        \Delta_k(O)
        &:=
        \cmean[\big]{X_k}{\tau(O)} - \cmean[\big]{X_{k-1}}{\tau(O)}
        =
        \cmean[\Big]{X_k}{\paren[\big]{\unit - \cmeansym_{X_{k-1}}} \, \tau(O)}
        .
    \end{align*}
    By properties of the conditional expectation, Lemma~\ref{lem:conditional-expectation}, and the Lieb-Robinson bound, Lemma~\ref{lem:LRB}, one can bound
    \begin{equation*}
        \Norm{\Delta_0(O)}
        \leq
        \Norm{O}
    \end{equation*}
    and
    \begin{equation}
        \label{eq:proof-lem-conjugation-with-unitaries-bound-Delta-k}
        \Norm{\Delta_k(O)}
        \leq
        C_\LR \, \norm{O} \, \abs{X}
        \, \paren[\big]{\e^{c \, \norm{\Phi_S}_{a,q}}-1} \, \chi_{a',q}(k)
    \end{equation}
    for \(k \geq 1\) because \(\dist{X,\Lambda\setminus X_{k-1}}=k\) in our geometry, where we chose \(a'<a\).

    We now construct an interaction for \(A\).
    First note, that
    \begin{equation*}
        A
        =
        \tau(A)
        =
        \sumstack[lr]{Z\subset\Lambda}
        \tau\paren[\big]{\Phi_D(Z)}
        =
        \sumstack[l]{Z\subset\Lambda}
        \sumstack[r]{k=0}^\infty
        \Delta_k\paren[\big]{\Phi_D(Z)}
    \end{equation*}
    where \(\Delta_k\paren[\big]{\Phi_A(Z)}\in\alg_{Z_k}\) and the sum is actually finite.
    For any function \(f\colon \Set{\Omega \subset \Lambda} \rightarrow \alg_{\Lambda}\) and \(k\geq0\) it holds that
    \begin{equation*}
        \sumstack[lr]{Y\subset\Lambda} f(Y)
        =
        \sumstack[l]{Y\subset\Lambda}
        \paren[\Big]{\,
            \sumstack[r]{Z\subset\Lambda}
            \unit_{Z=Y_k}
        }
        \, f(Y)
        =
        \sumstack[l]{Z\subset\Lambda}
        \sumstack[r]{Y\subset\Lambda\colonpunct\\Y_k=Z}
        f(Y)
        .
    \end{equation*}
    Applying this with \(f\colon Z \mapsto \Delta_k(\Phi_D(Z))\) for each \(k\) we find
    \begin{equation*}
        A
        =
        \sumstack{k=0}^\infty
        \sumstack{Z\subset\Lambda}
        \sumstack[r]{Y\subset\Lambda\colonpunct\\Y_k=Z}
        \Delta_k\paren[\big]{\Phi_D(Y)}
        =
        \sumstack{Z\subset\Lambda} \Phi_A(Z)
        \quadtext{with}
        \Phi_A(Z)
        :=
        \sumstack{k=0}^\infty
        \sumstack[r]{Y\subset\Lambda\colonpunct\\Y_k=Z}
        \Delta_k\paren[\big]{\Phi_D(Y)}
        .
    \end{equation*}
    With this interaction for \(A\) and any \(z\in \Lambda\) we bound
    \begin{align*}
        \sumstack{Z\subset \Lambda\colonpunct\\z\in Z}
        \frac
            {\norm{\Phi_A(Z)}}
            {\chi_{b,p}\pdiam{Z} \, \chi_{b,p}\pdist{z, \Omega}}
        &\leq
        \sumstack{Z\subset \Lambda\colonpunct\\z\in Z}
        \sumstack{k=0}^\infty
        \sumstack{Y\subset\Lambda\colonpunct\\Y_k=Z}
        \frac
            {\norm[\big]{\Delta_k\paren[\big]{\Phi_D(Y)}}}
            {\chi_{b,p}\pdiam{Z} \, \chi_{b,p}\pdist{z, \Omega}}
        \\&=
        \sumstack{k=0}^\infty
        \sumstack{Y\subset\Lambda}
        \unit_{z\in Y_k}
        \, \frac
            {\norm[\big]{\Delta_k\paren[\big]{\Phi_D(Y)}}}
            {\chi_{b,p}\pdiam{Y_k} \, \chi_{b,p}\pdist{z, \Omega}}
        .
    \end{align*}
    The \(k=0\) term is bounded by \(\norm{\Phi_D}_{p,b;\Omega}\).
    For \(k \geq 0\) we use~\eqref{eq:proof-lem-conjugation-with-unitaries-bound-Delta-k}.
    Moreover, \(\diam(Y_k) \leq \diam(Y) + 2k\) and, since \(z\in Y_k\), there exists \(y\in B_z(k) \cap Y\), such that \(\dist{z,\Omega} \leq k + \dist{y,\Omega}\).
    Hence, the remaining sum is bounded by
    \begin{align*}
        \Alignindent
        C_\LR
        \, \paren[\big]{\e^{c \, \norm{\Phi_S}_{a,q}}-1}
        \sum_{k=1}^\infty
        \sum_{y\in B_z(k)}
        \sumstack{Y\subset\Lambda\colonpunct\\y\in Y}
        \, \frac
            {\abs{Y} \, \norm[\big]{\Phi_D(Y)}}
            {\chi_{b,p}\pdiam{Y} \, \chi_{b,p}\pdist{y, \Omega}}
        \, \frac
            {\chi_{a',q}(k)}
            {\chi_{b,p}(2k) \, \chi_{b,p}(k)}
        \\&=
        C_\LR
        \, \paren[\big]{\e^{c \, \norm{\Phi_S}_{a,q}}-1}
        \, C \, \norm[\big]{\Phi_D}_{b+\varepsilon,p;\Omega}
        \sum_{k=1}^\infty
        \paren[\big]{ 1+\Cvol \, k^d }
        \, \frac
            {\chi_{a',q}(k)}
            {\chi_{(2^p+1)b,p}(k)}
        ,
    \end{align*}
    for some \(C>0\).
    The remaining sum is bounded if \(q>p\), or \(q=p\) and \(a'>(2^p+1)\,b\).
    The last condition is equivalent to \(a>(2^p+1)\,b\), by our choice of \(a'\).
    To conclude the statement, we choose the total constant larger than \(1\) and add the \(k=0\) term \(
        \norm[\big]{\Phi_D}_{b,p;\Omega}
        \leq
        \norm[\big]{\Phi_D}_{b+\varepsilon,p;\Omega}
    \).
\end{proof}

\subsection{Quasi-local inverse of the Liouvillian}
\label{app:inverse-liouvillian}
In this section, we briefly recall the construction of the quasi-local inverse of the Liouvillian \(\calI\) (see~\eqref{eq:invliouintro} and~\eqref{eq:invlioudef}) and a related operator \(\calJ\) used in~\eqref{eq:Jdef}.
Both of them use use certain properties (recall~\eqref{eq:Fouriercpct}--\eqref{eq:wbound}) of a \emph{weight function} \(w_g\), which one can construct explicitly.

\begin{lem}[Explicit weight function, cf.~Lemma~2.3 from~\cite{BMNS2012}]
    \label{lem:weightfunctions}
    Let \(g > 0\) and consider the sequence \((a_n)_{n \ge 1}\) of positive numbers, defined as \(a_n = a_1 (n (\log n)^2)^{-1}\) for \(n \ge 2\) and \(a_1\) chosen such that \(\sum_{n \ge 1} a_n = \gamma/2\).
    Then, the positive function \(w_g \in L^1(\R)\) defined via the infinite product
    \begin{equation}
        \label{eq:wexpl}
        w_g(t) := c_g \prod_{n=1}^{\infty} \paren*{\frac{\sin (a_nt)}{a_nt}}^2
    \end{equation}
    and \(c_g > 0\) chosen such that \(\int_\R \D t \, w_g(t) = 1\), has Fourier transform \(\widehat{w_g}\) with compact support \(\support(\widehat{w_g}) \subset [-g,g]\) (cf.~\eqref{eq:Fouriercpct}) and satisfies the bound \(\abs{w_g(t)} \le C \e^{-\abs{t}^q}\) for every \(q<1\) (cf.~\eqref{eq:wbound}).
\end{lem}

Given the explicit weight function~\eqref{eq:wexpl}, the quasi-local inverse of the Liouvillian \(\calI_{H,g} \Ab{\interpunct}: \Alg \to \Alg\) of the Hamiltonian \(H\) with gap parameter \(g>0\), acting on \(A \in \Alg\), is then defined as%
\begin{equation}
    \label{eq:invlioudef}
    \calI_{H,g}\Ab{A}
    :=
    \int_{\R} \D t \, w_g(t) \int_{0}^{t} \D s \, \e^{\I Hs} \, A \, \e^{-\I Hs}
    .
\end{equation}
\begin{rmk}[On the weight function]
	\label{rmk:weight}
	We point out, that, in principle and unless additional conditions are given, any map \(\calI_{H,g}\) with the properties~\eqref{eq:invliouintro}--\eqref{eq:wbound} would work for all of our proofs in this paper, in particular including the statements from Section~\ref{sec:localgap}.
\end{rmk}
Together with the map \(\calJ_{H,g} \colon \Alg \to \Alg\), again depending on the Hamiltonian \(H\) and gap parameter \(g> 0\), with action on \(A \in \Alg\) defined as
\begin{equation*}
    \calJ_{H,g}\Ab{A}
    :=
    \int_{\R} \D t \, w_g(t) \, \E^{\I Ht} \, A \, \E^{- \I Ht}
    ,
\end{equation*}
one then has (recalling the Liouvillian \(\calL_H\Ab{\interpunct} = -\I \, [H, \interpunct]\)) the identity \(\calL_H \circ \calI_{H,g}\Ab{A} - A = \calJ_{H,g}\Ab{A}\) for all \(A \in \Alg\); see~\eqref{eq:Jdef}.

The inverse Liouvillian \(\calI_{H,g}\) is called \emph{quasi-local}, since, if \(H\) satisfies the Lieb-Robinson bound from Lemma~\ref{lem:LRB}, then it holds that, for \(A \in \Alg_X\) and \(B \in \Alg_Y\) (see, e.g., \cite[Example~5.7]{NSY2019})
\begin{equation}
    \label{eq:invliouQL}
    \norm*{\commutator*{\calI_{H,g}\Ab{A} , B}}
    \le
    C \norm{A} \, \norm{B} \, \min\{ \abs{X}, \abs{Y}\} \chi_{\tilde{b}, \tilde{p}}\pdist{X, Y}
\end{equation}
for some \(\tilde{b} > 0\) (depending on the Lieb-Robinson velocity \(v\) from Lemma~\ref{lem:LRB}) and \(\tilde{p} \in (0,1)\), which can be chosen as \(p\) from~\eqref{eq:LR-bound-general} if \(p < 1\).
The estimate~\eqref{eq:invliouQL} holds verbatim with \(\calI\) replaced by \(\calJ\).

Beside the classical quasi-locality estimate~\eqref{eq:invliouQL}, the inverse Liouvillian even preserves locality of SLT operators.
This is the content of the following lemma, the special case for \(p=1\) without localization (i.e.\ for \(\Lambda\)-localized SLT operators) already appeared in~\cite[Theorem~7.5.6]{Maier2022} and is based on~\cite[Theorem~4.8]{BMNS2012}.

\begin{lem}[Quasi-local inverse of the Liouvillian on SLT operators]
    \label{lem:inverse-liouvillian-preserves-SLT}
    Let \(d\in \N\), \(\Cvol>0\), \(\varepsilon>0\), \(a\), \(b>0\), \(p\), \(q\in \intervaloc{0,1}\) satisfying \(q<p\), and \(\Cint>0\).
    There exist a constant \(C>0\) such that for all lattices \(\Lambda\in \mathcal{G}(d,\Cvol)\), the following holds:
    For SLT operators \(H\) and \(D\) it holds that \(\calI_{H,g}(D)\) is an SLT operator as well.
    More precisely if \(\Norm{\Phi_{H}}_{b,p} \leq \Cint\), then there exists an interaction \(\Phi_{\calI_{H,g}(D)}\) such that
    \begin{equation*}
        \Norm{\Phi_{\calI_{H,g}(D)}}_{a,q;\Omega}
        \leq
        C
        \, \norm[\big]{\Phi_D}_{a+\varepsilon,q;\Omega}
        .
    \end{equation*}
    The statement holds verbatim when replacing \(\calI\) by \(\calJ\) from~\eqref{eq:Jdef}.
\end{lem}

In the proof of the Lemma we use the equality
\begin{equation*}
    \calI_{H,g}\Ab{A}
    =
    \int_{\R} \D t \, \Wf_g(t) \, \e^{\I Ht} A \e^{-\I Ht}
    ,
\end{equation*}
where the function \(\Wf_g\) is given by \(\Wf_g(t) = - \int_{- \infty}^{t} \D s \, w_g(s) + \unit_{[0, \infty )}(t)\) with \(\unit_{[0,\infty )}\) being the characteristic function of \([0,\infty)\).
In many works, e.g.~\cite{BMNS2012, BRF2018, MT2019, Teufel2020, NSY2019}, this is used as a definition for the inverse Liouvillian.
It can easily be checked that, by Lemma~\ref{lem:weightfunctions}, also \(\Wf_g\) satisfies \(\abs{\Wf_g(t)} \le C \e^{-\abs{t}^q}\).

\begin{proof}
    The proof uses the same technique as the proofs of Lemma~\ref{lem:commutator-SLT-operator-evolved-observable} and~\ref{lem:conjugation-with-unitaries-evolved-SLT-is-SLT}.

    First fix \(X\subset \Lambda\) and \(O\in \alg_X\) and denote \(\tau(O) = \evol{\I H}{O}\).
    Then, define
    \begin{align*}
        \Delta_0(O)
        &=
        \I \int_{\R} \D s \, \Wf_g(s) \, \cmean[\big]{X}{\tau_s(O)}
    \shortintertext{and}
        \Delta_k(O)
        &=
        \I \int_{\R} \D s \, \Wf_g(s) \, \paren[\Big]{
            \cmean[\big]{X_{k}}{\tau_s(O)}
            - \cmean[\big]{X_{k-1}}{\tau_s(O)}
        }
        \\&=
        \I \int_{\R} \D s \, \Wf_g(s) \, \cmean[\Big]{X_{k}}{
            \paren[\big]{\unit - \cmeansym_{X_{k-1}}}
            \paren[\big]{{\tau_s(O)}}
        }
    \end{align*}
    for \(k\geq1\).
    Then \(\calI_{H,g}(O)=\sum_{k\in{\N}} \Delta_k(O)\) where the sum is eventually finite.

    For \(k=0\) we have
    \begin{equation*}
        \Norm{\Delta_0(O)}
        \leq
        \norm{\Wf_g}_{L^1} \, \Norm{O}
        .
    \end{equation*}

    For \(b'\in \intervaloo{0,b}\), \(k\geq1\) and some \(T>0\) to be chosen, Lemma~\ref{lem:LRB} and the properties of the conditional expectation, yield
    \begin{align*}
        &\alignindent \Norm*{
            \I \int_{-T}^T \D s \, \Wf_g(s) \, \cmean[\Big]{X_{k}}{
                \paren[\big]{\unit - \cmeansym_{X_{k-1}}}
                \paren[\big]{{\tau_s(O)}}
            }
        }
        \\&\leq
        \Norm{\Wf_g}_{L^\infty} \int_{-T}^T \D s \, \norm[\big]{
            \paren[\big]{\unit - \cmeansym_{X_{k-1}}}
            \paren[\big]{{\tau_s(O)}}
        }
        \\&\leq
        C_\LR \, \Norm{O} \, \abs{X} \, \Norm{\Wf_g}_{L^\infty}
        \, \int_{-T}^T \D s \,
        \paren{\e^{b' \, v \, \abs{s}} - 1}
        \, \chi_{b',p}(k)
        \\&=
        2 \, C_\LR \, \Norm{O} \, \abs{X} \, \Norm{\Wf_g}_{L^\infty}
        \frac{\e^{b' v T} - b' v T - 1}{b' \, v}
        \, \chi_{b',p}(k)
        \\&\leq
        2 \, C_\LR \, \Norm{O} \, \abs{X} \, \Norm{\Wf_g}_{L^\infty}
        \, \chi_{b'/2,p}(k) / (b' \, v)
    \end{align*}
    where we chose \(T=k^p/(2v)\) for the last step.
    Furthermore, by Lemma~\ref{lem:weightfunctions} and after integrating twice, for any \(0<\tilde{p}<1\) there exists \(C\) and \(\tilde{b}>0\) such that
    \begin{align*}
        &\alignindent \Norm*{
            \I \int_{\abs{s}\geq T} \D s \, \Wf_g(s) \, \cmean[\Big]{X_{k}}{
                \paren[\big]{\unit - \cmeansym_{X_{k-1}}}
                \paren[\big]{{\tau_s(O)}}
            }
        }
        \\&\leq
        2 \, \Norm{O} \int_{\abs{s}\geq T} \D s \, \abs{\Wf_g(s)}
        \\&\leq
        C \, \chi_{\tilde{b},\tilde{p}}(T)
        \\&\leq
        C \, \chi_{\tilde{b}',\tilde{p}p}(k)
        ,
    \end{align*}
    where \(\tilde{b}'=\paren[\big]{1/(2v)}^{\tilde{p}} \, \tilde{b}\).
    Then, for all \(p'\in \intervaloo{0,p}\) we can choose \(\tilde{p}=p'/p<1\) and combine the two bounds.
    Then, there exists \(C\) and \(\eta>0\) such that
    \begin{equation}
        \label{eq:bound-Delta_k(O)-for-qliL}
        \Norm{\Delta_k(O)}
        \leq
        C \, \abs{X} \, \Norm{O} \, \chi_{\eta,p'}(k)
        \qquadtext{for all}
        k \geq 0
        .
    \end{equation}

    An interaction for \(A:=\calI_{H,g}(D)\) is given by
    \begin{equation*}
        \Phi_A(Z)
        =
        \sum_{k=0}^\infty \sumstack{Y\subset\Lambda\colonpunct\\Y_k=Z} \Delta_k(\Phi_D(Y))
        .
    \end{equation*}
    It follows that
    \begin{align*}
        \sumstack[lr]{Z\subset\Lambda\colonpunct\\z\in Z}
        \, \frac{\Norm[\big]{\Phi_A(Z)}}{\chi_{a,q}\pdiam{Z}\,\chi_{a,q}\pdist{z,\Omega}}
        &\leq
        \sumstack[l]{Z\subset\Lambda\colonpunct\\z\in Z}
        \sumstack{k=0}^\infty
        \sumstack{Y\subset\Lambda\colonpunct\\Y_k=Z}
        \, \frac{\Norm[\big]{\Delta_k\paren[\big]{\Phi_D(Y)}}}{\chi_{a,q}\pdiam{Z}\,\chi_{a,q}\pdist{z,\Omega}}
        \\&=
        \sumstack{k=0}^\infty
        \sumstack[r]{Y\subset\Lambda}
        \unit_{z\in Y_k}
        \, \frac{\Norm[\big]{\Delta_k\paren[\big]{\Phi_D(Y)}}}{\chi_{a,q}\pdiam{Y_k}\,\chi_{a,q}\pdist{z,\Omega}}
        .
    \end{align*}
    The \(k=0\) term is bounded by \(\norm{\Wf_g}_{L^1} \Norm{\Phi_D}_{a,q;\Omega}\).
    For \(k\geq1\) and \(z\in Y_k\) there exists \(y\in B^{\Lambda}_z(k)\cap Y\) such that \(\dist{z,\Omega}\leq k+\dist{y,\Omega}\).
    Furthermore, \(\diam(Y_k) \leq 2k+\diam(Y)\).
    Hence, using~\eqref{eq:bound-Delta_k(O)-for-qliL} the rest of the sum is bounded by
    \begin{align*}
        &\alignindent
        C\,
        \sumstack{k=1}^\infty
        \, \chi_{\eta,p'}(k)
        \sumstack{y\in B^{\Lambda}_z(k)}
        \sumstack{Y\subset\Lambda\colonpunct\\y\in Y}
        \frac{\Norm[\big]{\Phi_D(Y)} \, \abs{Y}}{\chi_{a,q}\pdiam{Y_k}\,\chi_{a,q}\pdist{z,\Omega}}
        \\&\leq
        C
        \, \sumstack{k=1}^\infty
        \, \frac{\chi_{\eta,p'}(k)}{\chi_{a,q}(2k)\,\chi_{a,q}(k)}
        \sumstack{y\in B^{\Lambda}_z(k)}
        \Ccod{\varepsilon,q,1}
        \sumstack{Y\subset\Lambda\colonpunct\\y\in Y}
        \frac{\Norm[\big]{\Phi_D(Y)}}{\chi_{a+\varepsilon,q}\pdiam{Y}\,\chi_{a,q}\pdist{y,\Omega}}
        \\&\leq
        C
        \, \Ccod{\varepsilon,q,1}
        \, \Norm[\big]{\Phi_D}_{a+\varepsilon,q;\Omega}
        \sumstack{k=1}^\infty
        \frac{\chi_{\eta,p'}(k)}{\chi_{(2^q+1)a,q}(k)}
        \, \paren[\big]{1 + \Cvol \, k^d}
        .
    \end{align*}
    The remaining sum is bounded if \(p'>q\), which we can ensure if \(p>q\) by choosing \(p'\in \intervaloo{q,p}\).
    Thus, there exists \(C\) such that
    \begin{equation*}
        \sumstack[l]{Z\subset\Lambda\colonpunct\\z\in Z}
        \, \frac{\Norm[\big]{\Phi_A(Z)}}{\chi_{a,q}\pdiam{Z}\,\chi_{a,q}\pdist{z,\Omega}}
        \leq
        C \, \Norm{\Phi_D}_{a+\varepsilon,q;\Omega}
    \end{equation*}
    for some constant \(C\) depending on \(\Norm{\Phi_H}_{p,b}\), \(p\), \(b\), \(q\), \(a\), \(\varepsilon\) and \(d\) which finishes the proof.
\end{proof}

\begin{rmk}[Abstract properties of \(\calI\) needed in the proof of our main result]
	\label{rmk:abstract}
	For the purpose of proving our main result, Theorem~\ref{thm:linear-response}, it is not necessary to work with the explicit \(\calI_{H,g}\) from~\eqref{eq:invlioudef}.
	In fact, by inspecting the proof of Proposition~\ref{prop:adswitch} in Section~\ref{sec:proof}, which is the key input for our main result, we realize the following: For Theorem~\ref{thm:linear-response} being valid (up to minor adjustments of constants), one only needs that there exists some operator \(\widetilde{\calI}: \Alg \to \Alg\) for which Assumption \nameref{ass:localGAP} is satisfied and such that
	\begin{itemize}
		\item[(i)] for \(A \in \Alg_X\) and \(B \in \Alg_Y\) it holds that (cf.~\eqref{eq:smallZ2})
		\begin{equation*}
			\abs*{\expectation*{\commutator*{\calL_{H} \circ \widetilde{\calI}\Ab{A} - A , B}}_{\rho_0}}
			\le
			C \norm{A} \, \norm{B} \, \diam(X)^\ell \exp(-\dist{X,Y}^q)
		\end{equation*}
		for some positive constants \(C,q, \ell >0\), i.e.~the composition \(\calL_{H} \circ \widetilde{\calI}\) behaves as a quasi-local operator if tested in the above way;
		\item[(ii)] Lemma~\ref{lem:inverse-liouvillian-preserves-SLT} holds, i.e.~\(\widetilde{\calI}\) maps localized SLT operators to localized SLT operators.
	\end{itemize}
	
	These relaxed abstract conditions are, however, not sufficient for showing the relations among the various gap conditions outlined in Section~\ref{sec:localgap}.
\end{rmk}

\subsection{Localized SLT operators: Proof of Lemma~\ref{lem:SLTlocal}}
\label{app:proof-SLTlocal}

For (i), it suffices to realize that, since for strictly \(\Omega\)-localized \(\Phi\) it holds that \(\Phi(Z) = 0\) whenever \(Z \cap (\Lambda \setminus \Omega) \neq \emptyset\), we have
\begin{equation*}
    \norm{\Phi}_{b,p; \Omega}
    =
    \sup_{z \in \Lambda} \sumstack{Z\subset \Omega: \\ z\in Z}
    \frac{\norm{\Phi(Z)}}{\chi_{b,p}\pdiam{Z} \, \chi_{b,p}\pdist{z,\Omega}}
    \le
    \sup_{z \in \Lambda} \sumstack{Z\subset \Lambda: \\ z\in Z}
    \frac{\norm{\Phi(Z)}}{\chi_{b,p}\pdiam{Z} }
    =
    \norm{\Phi}_{b,p}
    \le
    C
    .
\end{equation*}
Next, for (ii) and strongly \(\Omega\)-localized \(\Phi\), we have that
\begin{equation*}
    \begin{split}
        \norm{\Phi}_{b/2,p; \Omega}
         & =
        \sup_{z \in \Lambda} \sumstack{Z\subset \Lambda: \\ Z \cap \Omega \neq \emptyset , \, z\in Z}
        \frac{\norm{\Phi(Z)}}{\chi_{b/2,p}\pdiam{Z} \, \chi_{b/2,p}\pdist{z,\Omega}}
        \\ & \le
        \sup_{z \in \Lambda} \sumstack{Z\subset \Lambda: \\ z\in Z}
        \frac{\norm{\Phi(Z)}}{\chi_{b,p}\pdiam{Z} }
        =
        \norm{\Phi}_{b,p}
        \le
        C
    \end{split}
\end{equation*}
since \(\dist{z,\Omega} \le \diam(Z)\) for \(z \in Z\) and \(Z \cap \Omega \neq \emptyset\), together with monotonicity of \(\chi_{b/2, p}\) and using \((\chi_{b/2, p})^2 = \chi_{b, p}\).
This concludes the proof of Lemma~\ref{lem:SLTlocal}.
\qed

\subsection{Assumption \nameref{ass:localGAPweak} and SLT operators: Proof of Lemma~\ref{lem:LDGweak}}
\label{app:proof-LDGweak}

We write \(A = \sum_{Z \subset \Lambda} \Phi_A(Z)\) (cf.~\eqref{eq:intSLT}) and estimate
\begin{equation}
    \label{eq:LDGweakmain}
    \begin{aligned}
        \Alignindent
        \abs*{\expectation[\big]{\commutator*{\calL_{H_0} \circ \calI_{H_0,g}\Ab{A} - A , B}}_{\rho_0}}
        \\&\le
        \sum_{Z \subset \Lambda} \abs[\Big]{
            \expectation[\Big]{
                \commutator[\big]{\calL_{H_0} \circ \calI_{H_0,g}\Ab{\Phi_A(Z)} - \Phi_A(Z), B}
            }_{\rho_0}
        }
        .
    \end{aligned}
\end{equation}

For the ‘small’ \(Z \subset \Lambda\) satisfying \(\diam(Z) \le \dist{Z, \Lambda \setminus \Lambdag}^\beta\), we bound
\begin{equation}
    \label{eq:smallZ}
    \begin{aligned}
        \Alignindent
        \abs[\Big]{
            \expectation[\Big]{
                \commutator[\big]{\calL_{H_0} \circ \calI_{H_0,g}\Ab{\Phi_A(Z)} - \Phi_A(Z), B}
            }_{\rho_0}
        }
        \\&\le
        C_\gap \, \norm{\Phi_A(Z)} \, \norm{B}
        \, \diam(Y)^\ell
        \, \chi_{b,p}\pdist{Z, \Lambda \setminus \Lambdag}
    \end{aligned}
\end{equation}
by means of Assumption \nameref{ass:localGAPweak}.
Additionally, we need the following alternative estimate on~\eqref{eq:smallZ} (recall~\eqref{eq:Jdef}):
\begin{align}
    \label{eq:smallZ2}
    \Alignindent
    \abs[\Big]{
        \expectation[\Big]{
            \commutator[\big]{\calL_{H_0} \circ \calI_{H_0,g}\Ab{\Phi_A(Z)} - \Phi_A(Z), B}
        }_{\rho_0}
    }
    \nonumber
    \\&\le
    \int_\R \D t \, w_g(t)
    \, \norm[\big]{\commutator[\big]{\E^{\I t H_0} \, \Phi_A(Z) \, \E^{-\I t H_0}, B}}
    \nonumber
    \\&\le
    C \, \norm{\Phi_A(Z)} \, \norm{B} \, \diam(Z)^d
    \, \paren*{
        \chi_{b/2,p}\pdist{Z,Y} \int_I \D t\, w_g(t)
        + \int_{\R\setminus I} \D t\, w_g(t)
    }
    \nonumber
    \\&\le
    C \, \norm{\Phi_A(Z)} \, \norm{B} \, \diam(Z)^d
    \, \chi_{b/2,p}\pdist{Z,Y}
\end{align}
where we denoted \(I := \Set[\big]{t \in \R \given v\abs{t} \le \dist{Z,Y}^p/3 }\).
Here, \(v\) is the Lieb-Robinson velocity from Lemma~\ref{lem:LRB},
which we employed in the second step with \(b \to 3b/4\).
In the final step, we used the stretched exponential decay of \(w_g\) (see~\eqref{eq:wbound} and Lemma~\ref{lem:weightfunctions}).
Note that~\eqref{eq:smallZ} and~\eqref{eq:smallZ2} track two different relevant distances, namely those of \(Z\) to \(\Lambda \setminus \Lambdag\) and \(Y\), respectively.

In fact, a weighted geometric mean of~\eqref{eq:smallZ} and~\eqref{eq:smallZ2}, that combines these two effects, can now be summed up as (neglecting the factor \(C\,\norm{B}\,\diam(Y)^{\ell}\), which will be put back in~\eqref{eq:LDGweakfinal})
\begin{align}
    \label{eq:smallZfinal}
    \Alignindent
    \sumstack[r]{Z \subset \Lambda:\\ \diam(Z) \leq \dist{Z, \Lambda \setminus \Lambdag}^\beta}
    \diam(Z)^{d}
    \, \chi_{b-\epsi,p}\pdist{Z, \Lambda \setminus \Lambdag}
    \, \chi_{\epsi/2,p}\pdist{Z,Y}
    \, \norm{\Phi_A(Z)}
    \nonumber
    \\&\le
    \sum_{z \in \Lambda}
    \sumstack{Z \subset \Lambda:\\ z \in Z}
    \norm{\Phi_A(Z)} \frac
        {\chi_{b-\epsi,p}\paren[\big]{\diam(Z) + \dist{z, \Omega} + \dist{Z, \Lambda \setminus \Lambdag}}}
        {\chi_{b,p}\pdiam{Z} \, \chi_{b,p}\pdist{z, \Omega}}
    \, \chi_{\epsi/2,p}\pdist{z,Y}
    \nonumber
    \\&\le
    \norm{\Phi_A}_{b,p; \Omega}
    \, \chi_{b-\epsi,p}\pdist{\Omega, \Lambda \setminus \Lambdag}
    \sum_{z \in \Lambda}
    \chi_{\epsi/2,p}\pdist{z,Y}
    \nonumber
    \\&\le
    C
    \, \diam(Y)^d
    \, \norm{\Phi_A}_{b,p; \Omega}
    \, \chi_{b-\epsi,p}\pdist{\Omega, \Lambda \setminus \Lambdag}
    .
\end{align}
For the first bound, we used logarithmic superadditivity of \(\chi_{b,p}\) together with elementary monotonicity properties from Lemma~\ref{item:lem-chifnct-logarithmically-superadditive} and estimated \(\diam(Z)^{d}\) by \(1/\chi_{\epsi,p}\pdiam{Z}\).
For the second bound, we used the definition of \(\norm{\Phi_A}_{b,p; \Omega}\) from~\eqref{eq:interaction-norm-localization} and the fact that, for \(z \in Z\), we have \(
    \diam(Z) + \dist{z, \Omega} + \dist{Z, \Lambda \setminus \Lambdag}
    \ge
    \dist{\Omega, \Lambda \setminus \Lambdag}
\).
In the final step, we employed summability of \(\chi_{\epsi/2,p}\pdist{z,Y}\).

Therefore, combining~\eqref{eq:LDGweakmain} with~\eqref{eq:smallZfinal}, the contribution of those \(Z \subset \Lambda\), for which \(\diam(Z) \le \dist{Z, \Lambda \setminus \Lambdag}^\beta\), to~\eqref{eq:LDGweakmain} is bounded by
\begin{equation}
    \label{eq:LDGweakfinal}
    C \, \diam(Y)^{d+ \ell} \, \norm{B} \, \norm{\Phi_A}_{b,p; \Omega} \, \chi_{b-\epsi,p}\pdist{\Omega, \Lambda \setminus \Lambdag}
    .
\end{equation}

For the ‘large’ \(Z \subset \Lambda\) that satisfy \(\diam(Z) > \dist{Z, \Lambda \setminus \Lambdag}^\beta\), we simply use the estimate from~\eqref{eq:smallZ2}, which we can sum up as
\begin{equation}
    \label{eq:largeZfinal}
    \begin{aligned}
        \Alignindent
        \sumstack[r]{Z \subset \Lambda:\\ \diam(Z) > \dist{Z, \Lambda \setminus \Lambdag}^\beta}
        \diam(Z)^d
        \, \chi_{b/2,p}\pdist{Z,Y}
        \, \norm{\Phi_A(Z)}
        \\&\le
        \sum_{z \in \Lambda} \sumstack{Z \subset \Lambda:\\ z \in Z}
        \, \norm{\Phi_A(Z)}
        \, \frac
            {\chi_{b/2,p\beta}\paren[\big]{\diam(Z) + \dist{z, \Omega} + \dist{Z, \Lambda \setminus \Lambdag}}}
            {\chi_{b,p}\pdiam{Z} \, \chi_{b,p}\pdist{z, \Omega}}
        \, \chi_{b/2,p}\pdist{z,Y}
        \\&\le
        C \, \diam(Y)^d \, \norm{\Phi_A}_{b,p; \Omega} \, \chi_{b/2,p\beta}\pdist{\Omega, \Lambda \setminus \Lambdag}
    \end{aligned}
\end{equation}
analogously to~\eqref{eq:smallZfinal}.
In the second step we used that \(\diam(Z) > \dist{Z, \Lambda \setminus \Lambdag}^\beta\) and elementary monotonicity properties of \(\chi_{b,p}\) in \(b,p\).

Therefore, by means of~\eqref{eq:largeZfinal}, also the large \(Z\)'s contribute only in a way that is controlled in terms of~\eqref{eq:LDGweakfinal} (but with worse constants \(b/2\) and \(p \min\{\beta, 1 \}\)).
This concludes the proof of Lemma~\ref{lem:LDGweak}.
\qed

\printbibliography[heading=bibintoc]

\end{document}